\documentclass[11pt]{article}

\usepackage{amsmath,amssymb,amsthm}
\usepackage[authoryear,round,sort]{natbib}
\setcitestyle{aysep={}}
\usepackage{hyperref}
\usepackage{graphicx}
\usepackage[font=small,skip=4pt]{caption}
\usepackage{authblk}
\usepackage[letterpaper,margin=1in]{geometry}
\newcommand{\Pp}{\mathbb{P}}
\newcommand{\E}{\mathbb{E}}
\newcommand{\Var}{\mathrm{Var}}
\newcommand{\Cov}{\mathrm{Cov}}
\newcommand{\Corr}{\mathrm{Corr}}

\theoremstyle{plain}
\newtheorem{theorem}{Theorem}
\newtheorem{corollary}{Corollary}
\theoremstyle{remark}

\title{Majority Correctness in Social Networks:\\
From Well-Mixed Electorates to Complex Networks}

\author[1,3]{Dan Braha\thanks{Corresponding author. Email: \href{mailto:braha@necsi.edu}{braha@necsi.edu}}}
\author[2,3]{Marcus A. M. de Aguiar}

\affil[1]{University of Massachusetts, Dartmouth, Massachusetts, United States of America}
\affil[2]{Gleb Wataghin Institute of Physics, State University of Campinas (UNICAMP), Campinas, Brazil}
\affil[3]{New England Complex Systems Institute, Cambridge, Massachusetts, United States of America}

\date{}

\begin{document}

\begin{titlepage}
\thispagestyle{empty}
\maketitle

\begin{abstract}
\noindent We study majority correctness when voting is preceded by sustained social interaction on a social network. Motivated by the Condorcet Jury Theorem, we consider a binary choice with an objectively correct alternative, where uninformed voters revise their vote intentions through repeated interaction in the presence of competing committed leaders (zealots). In this zealot--contrarian voter model, voters may either imitate or oppose the views they encounter. For fully mixed electorates, we characterize the long-run distribution of votes and the correlation structure induced among voters, and we show, under a standard approximation for random networks, that Erd\H{o}s--R\'enyi networks exhibit the same majority-correctness behavior after an appropriate rescaling of leader influence. Building on these results, we establish a finite-electorate Condorcet-type guarantee: when post-deliberation individual correctness exceeds random choice, a strict majority is more likely to select the correct alternative than a randomly chosen voter. At the same time, we identify an aggregation failure: social interaction can reduce majority accuracy relative to a no-deliberation benchmark in which voters respond only to zealots. As the electorate size tends to infinity, this finite-electorate advantage disappears unless social updating is purely conformist, revealing a tipping point at full conformity: any persistent contrarian updating drives both individual and majority correctness to the random choice level of one half. Simulations on scale-free, ring, and small-world networks further show that topology matters because it shapes the vote correlations generated by social influence: hub-dominated structures generate stronger positive correlations and lower majority accuracy, whereas spatially structured networks generate weaker correlations, preserve a larger effective number of independent judgments, and improve majority accuracy.
\end{abstract}



\end{titlepage}
\setcounter{page}{2}

\section{Introduction}
\label{intro}

The Condorcet Jury Theorem (CJT) provides a natural benchmark for thinking about when voting yields accurate collective decisions. In its classical form, each voter independently selects the objectively correct alternative with probability $p>0.5$. Under these assumptions, simple majority rule is more likely to choose the correct alternative than any individual voter, and the probability of a correct collective decision increases with electorate size, approaching one as the electorate becomes large \citep{condorcet1785}. For this reason, the CJT has long served as an epistemic justification for democratic decision-making: when voters are individually competent and their judgments are independent, majority rule can track truth more reliably than individual judgment \citep{listgoodin2001,estlund2008}.

Modern work on voting and public opinion, however, emphasizes two features of real electorates that depart from these idealizations. First, voter competence is heterogeneous: individuals differ substantially in how well informed they are about the factual merits of competing alternatives. Second, opinions are not formed in isolation. Voters are exposed to common signals from elites, media, and social ties, and they deliberate, communicate, and influence one another in ways that create dependence across beliefs and votes \citep{zaller1992}. A substantial literature therefore examines how heterogeneity in competence and dependence among voters can weaken, modify, or even overturn the classical CJT conclusions (see Section~\ref{literat}). 

These observations motivate a shift from one-shot aggregation to the communication dynamics that precede a vote. Elections and referendums are typically preceded by sustained exposure to messages, cues, and social discussion. Even when voters begin with limited substantive knowledge, they may learn indirectly through elite discourse and interpersonal communication. The central question is therefore dynamic: under what conditions does opinion formation increase the accuracy of the eventual collective decision, and under what conditions does it instead generate fragility through herding and path dependence?

The ``Condorcet-like'' process studied here makes this shift explicit while remaining deliberately parsimonious. Consider $n$ free voters facing a binary choice between two alternatives, one of which is objectively correct, namely alternative~$1$. A small set of \emph{committed leaders} holds fixed vote intentions throughout the campaign; for example, there may be five leaders, three favoring alternative~$1$ and two favoring alternative~$2$. The remaining individuals, the \emph{free voters}, begin without reliable private information. The population then enters a prolonged campaign phase in which free voters repeatedly revise their vote intentions through interpersonal exposure, so that dependence in votes arises endogenously through social interaction. In the fully mixed baseline, represented by a complete interaction graph, a free voter is selected uniformly at random at each time step. That voter consults another individual chosen uniformly at random from the rest of the population, whether a free voter or a leader, and then updates by either adopting that individual’s current vote or switching to the opposite vote. We refer to this dynamics as the \emph{zealot--contrarian voter model}. In networked versions, the same rule applies, but interactions among free voters are constrained by an underlying graph: each free voter consults only its graph neighbors, while all free voters remain connected to the same fixed set of committed leaders. 

This ``who talks to whom'' structure shapes both how influence spreads and how dependence among voters develops. Degree heterogeneity, clustering, community structure, and path length can affect both the diffusion of opinions and the concentration of influence. After a sufficiently long horizon that brings the system close to stationarity, on the order of $n^{2}$ update attempts in the fully mixed idealization, a referendum is held. This long horizon is not merely a technical convenience. It captures the idea that publics typically have time to observe, discuss, and update before a decision is taken, an intuition that also motivates deliberative democracy proposals \citep{fishkin1991}. From a network perspective, it separates transient dynamics from the long-run regime, allowing accuracy to be assessed in or near the stationary distribution induced by the committed leaders on a given topology. The model does not assume that discussion improves judgment by default. Rather, it provides a clean baseline for studying how sustained exposure to persistent committed voices, together with endogenous social influence, interacts with network structure to shape collective accuracy. In this way, it offers a simple bridge between jury theorem arguments about tracking the truth and computational studies of influence, correlation, and opinion dynamics on networks. 

This setup brings together three literatures that are often studied separately. Epistemic approaches ask when democratic decisions are likely to be correct \citep{condorcet1785,listgoodin2001,estlund2008}. In parallel, research on mass opinion formation emphasizes elite cues, interpersonal diffusion, and the endogenous emergence of correlated beliefs \citep{zaller1992,lenz2012}. A third perspective, from network science and dynamical processes on networks, studies how interaction topology, meaning who is connected to whom and how heterogeneous, clustered, or modular the connections are, shapes the spread of opinions and the resulting long-run collective outcome. The present model links these perspectives by shifting the source of accuracy. The relevant mechanism is no longer a large number of independent voters, each slightly more likely than not to be correct, but the long-run influence of a small set of fixed, competing committed leaders acting on an initially uninformed public as opinions diffuse through the social network. The referendum therefore does not aggregate private signals; it reflects population-level influence dynamics, and its accuracy can depend strongly on network topology.

The updating rule is intentionally behaviorally minimal. It does not model deliberation as the exchange and evaluation of reasons or evidence, and it does not assume Bayesian inference. Instead, it compresses pre-vote communication into a simple heuristic: when a voter updates, the voter either imitates the opinion of a randomly encountered individual or, with some probability, adopts the opposite opinion; otherwise the voter retains its current view. This mechanism is closely related to the voter model introduced by \citet{clifford_sudbury1973} and further developed in foundational work by \citet{holley_liggett1975}. It is also complementary in spirit to classic social learning models such as \citep{degroot1974}, but replaces continuous belief averaging with a discrete imitation rule and focuses on the resulting stationary distribution rather than on convergence to consensus. In voter model terms, the committed leaders are \emph{zealots}, exogenous and unchanging sources that continually inject fixed opinions into the dynamics of the free voters. Unlike the voter model without zealots, which admits absorbing consensus states, the presence of zealots prevents absorption and typically yields a nondegenerate stationary distribution over vote counts, not necessarily symmetric, together with a characteristic relaxation time to that regime. Voter models with zealots have therefore been used to study when small committed groups can shift population-level outcomes, and how the stationary vote share distribution depends on zealot influence and interaction structure \citep{mobilia2003,chinellato2007,xie2011,chinellato2015}. In political terms, zealots can represent unwavering elites, activists, or committed leaders whose positions remain stable while the broader public adapts through repeated exposure.

Under this framing, the substantive questions are sharp and Condorcet in spirit. Does pre-vote social interaction under the influence of committed leaders increase the probability that a strict majority selects the correct alternative beyond the stationary correctness of a randomly chosen free voter? If so, under what conditions, if any, does majority correctness approach certainty as the electorate grows? These questions naturally lead to benchmark comparisons. In particular, how does majority correctness under the interaction dynamics compare with the counterfactual in which free voters do not influence one another and respond only to the committed leaders, so that votes reflect independent exposure to the same leader environment rather than dependence generated by social interaction? For networked electorates, an equally central issue is structural: how does collective accuracy depend on the topology of social influence? The same set of committed leaders can produce different outcomes on Erd\H{o}s--R\'enyi, scale-free, small-world, and ring networks because these architectures differ in clustering, path length, and the concentration of influence. More broadly, the aim is to understand how accuracy changes as we move away from the CJT idealization of independent signals toward the empirically more typical setting of socially influenced opinions. By modeling the pre-referendum phase explicitly, the process reflects a basic political science reality: collective choices are often the endpoint of campaigns, conversations, and cue-driven learning, rather than a one-shot aggregation of independent private judgments. 

\section{Related Literature}
\label{literat}

The ``Condorcet-like'' process studied in this paper treats the eventual vote as the endpoint of communication dynamics on a social network under external influence from zealots. This framework departs from the classical jury theorem in two natural ways: voters may differ in their effective competence, and their opinions and votes may be statistically dependent. We therefore organize the literature review around three themes that are central to the present analysis: heterogeneity in voter competence, dependence in opinions and votes, and models of voting shaped by pre-vote communication and social network dynamics.

Condorcet's jury theorem (CJT), first articulated in 1785, remains a foundational result in social choice theory \citep{black1958, kazmann1973, grofman1975, mccannon2015}. More broadly, it has long served as a central formal expression of the ``wisdom of crowds,'' influencing work in fields ranging from philosophy and political science to law, economics, and finance \citep{mccannon2015}. Its appeal lies in the idea that, under suitable conditions, collective decisions can outperform individual judgments and become more reliable in larger groups \citep{mccannon2015}.

In its classical form, the CJT rests on several strong assumptions: homogeneous voter competence with $p>0.5$, meaning that every voter is equally more likely than not to choose the correct alternative; statistical independence, so that voting errors are uncorrelated; and simple majority rule applied to these independent judgments, with no pre-vote communication that could induce dependence \citep{kazmann1973, grofman1975, mccannon2015}. Under these conditions, the probability that the majority is correct exceeds that of any individual voter, the nonasymptotic part of the theorem, and rises with group size, converging to $1$ as the number of voters tends to infinity, the asymptotic part. Real electorates, however, rarely satisfy all of these assumptions simultaneously. A large literature has therefore examined how heterogeneity in competence and dependence among votes affect this conclusion, identifying settings in which collective accuracy remains high as well as settings in which it can fail.
 
One major line of research relaxes the assumption of homogeneous competence. Condorcet's original theorem assumes identical voters, but real electorates differ in ability. Early extensions showed that its asymptotic conclusion can remain valid under heterogeneous success probabilities, provided voters are, on average, better than random. More precisely, under independence and mild regularity conditions, majority accuracy still converges to $1$ as group size grows whenever mean competence exceeds $0.5$ \citep{boland1989, owen1989, berend1998}. Finite electorate results are more delicate. If additional voters are substantially less competent, the dilution of average ability can outweigh the benefit of larger numbers, causing majority accuracy to decline initially as the group expands \citep{karotkin2003}. This raises the possibility of an ``optimal'' committee size in some settings, beyond which adding voters reduces performance \citep{karotkin2003}. Nevertheless, majority voting often remains robust: under many competence profiles, the majority's accuracy exceeds the average voter's accuracy and can sometimes even exceed that of the most competent individual in the group \citep{grofman1983thirteen, berend1998, boland1989, benyashar2000nonasymptotic}.

Perhaps the most difficult CJT assumption to relax, and arguably the most consequential in real-world settings, is the independence of voters' information and votes. Once voters' information or choices are correlated, the familiar CJT probability calculations no longer apply in the same way, and the epistemic gains from large electorates can be substantially weakened \citep{nitzan1982optimal, boland1989modelling, ladha1992condorcet, berg1993condorcet, ladha1993condorcet, kaniovski2010aggregation, pivato2017epistemic}. Dependence can arise through shared information sources, common cues, mutual influence, or the contagion of opinions generated by deliberation. One simple and important mechanism is common influence, such as an opinion leader or an external source like advertising or a promotional campaign: many voters respond to the same message, so their decisions move together rather than reflecting independent information. \citet{boland1989modelling} formalize this idea and show that even a single highly influential source can impose an upper bound on group accuracy, no matter how large the electorate becomes. In the limiting case, if many voters simply follow the same source, the collective decision effectively reproduces that source, and enlarging the electorate contributes little additional information. Along similar lines, \citet{dietrich2004model} and \citet{dietrich2013epistemic, dietrich2013independent} show that when voters rely on only a small number of imperfect shared information sources, the asymptotic CJT can fail: enlarging the electorate does not make the collective decision more reliable than those underlying sources. 

More generally, as positive dependence strengthens, the informational advantage of majority rule diminishes. Negative dependence, by contrast, can improve aggregation because voters are more likely to make different mistakes, allowing one voter's error to be offset by another's correct choice. \citet{ladha1992condorcet} makes this point precise in an extension of the CJT with correlated votes, deriving conditions, in terms of average competence and overall dependence, under which majority rule still outperforms the average voter, and showing how negative dependence across groups can increase majority accuracy even when some voters are poorly informed. A related conclusion appears in \citet{berg1993condorcet}, where dependence is modeled using a P\'olya--Eggenberger framework. In that setting, strong positive correlation can prevent majority accuracy from approaching certainty as the electorate grows, whereas weaker or negative dependence can improve the probability of a correct majority. In a complementary approach, \citet{ladha1993condorcet} models dependence through a shared latent informational component. Assuming exchangeable, symmetrically dependent votes, and applying de Finetti's theorem, he represents votes as conditionally i.i.d.\ Bernoulli variables given a random common parameter, which can be interpreted as a shared information source or as a common component in the way evidence is processed. Within this framework, he shows that majority rule can still outperform an individual voter under broad conditions, although the advantage declines as correlation increases. \citet{berend2007monotonicity} provide general conditions under which the nonasymptotic CJT guarantee holds for committees of correlated voters, and \citet{peleg2012extending} extend this line of work by deriving broad necessary and sufficient conditions under which correlated electorates satisfy the CJT. 

Later work shows that the effect of correlation cannot be summarized by a single number without additional structural assumptions. \citet{kaniovski2010aggregation} emphasizes an ``aggregation'' problem: specifying individual competence and pairwise correlation does not uniquely determine the joint distribution of votes, and different dependence structures consistent with the same pairwise correlations can produce very different majority accuracies. He derives two contrasting jury theorem results for homogeneous electorates based on different resolutions of this aggregation problem, and develops a bounds approach that computes the best- and worst-case majority accuracy compatible with a given competence and dependence specification. Complementing this, \citet{pivato2017epistemic} develops a general large population theory of epistemic democracy with correlated voters, showing that truth-tracking can still emerge under correlation when the average covariance between voters becomes small as the population grows. He further shows that this condition is compatible with dependence generated by social networks, for example through mainly local correlations, and with DeGroot-type deliberation when influence is sufficiently diffuse. Overall, these results clarify both why dependence can undermine CJT-style guarantees and when, despite correlation, large electorates can remain reliably accurate. 

Another crucial CJT assumption is that voters form judgments in isolation, without strategic behavior or pre-vote communication. In practice, however, voters observe others' choices, exchange opinions, and sometimes share, withhold, or strategically disclose evidence. These interactions also induce dependence, but through mechanisms that differ from those discussed above. A useful way to organize this literature is around four representative strands discussed below: sequential observational learning (cascades), repeated exchange on networks (na\"ive and Bayesian learning), evidence-based deliberation and disclosure protocols, and reduced-form local-update opinion dynamics (majority dynamics and voter-type rules). The last of these provides the conceptual lens most closely related to our model. These four strands can be compared along two dimensions: first, what is communicated, such as actions, stated beliefs, or underlying evidence; and second, how that information is processed, whether through Bayesian or strategic inference or through reduced-form heuristic updating.

\noindent\textbf{(i)} \emph{Sequential observational learning} models pre-vote communication through the observation of others' actions. A central example is herding through informational cascades. In sequential settings, early actions can shape later choices, and \citet{bikhchandani1992theory} show that once an initial sequence of decisions points in one direction, later voters may rationally disregard their own signals and follow the apparent consensus, potentially locking the group into an incorrect outcome; see also \citep{banerjee1992simple}. This highlights an important limitation of the CJT intuition: large electorates need not aggregate information well when individuals learn from others' actions, because observed choices generate dependence even in the absence of deception or a single common source. Subsequent theoretical, experimental, and empirical work documents cascade and herding dynamics in settings ranging from laboratory markets and financial trading to elections and polling environments \citep{bikhchandani2024information,shachat2022informational, granzier2023coordination,harmon2015anticipating, braha_deaguiar2017}. 

\noindent\textbf{(ii)} \emph{Repeated belief exchange on networks} treats deliberation as iterative opinion exchange over a social graph under either non-Bayesian (na\"ive) or Bayesian learning. Two broad frameworks have been studied. In the non-Bayesian, or na\"ive, approach, often represented by the DeGroot averaging process \citep{degroot1974}, individuals repeatedly take weighted averages of their neighbors' opinions or beliefs and may converge to a consensus. 
A central question is whether that consensus tracks the true state. \citet{golub2010naive} show that if long-run influence is sufficiently diffuse, so that no individual or small group remains disproportionately influential, then as society grows, the consensus is correct with high probability. In fully Bayesian models of network learning, agents update posteriors based on neighbors' reports; \citet{acemoglu2011bayesian} show that even under Bayesian rationality some observation structures lead to learning failures, whereas others support asymptotic learning. More broadly, these results highlight a recurring theme: communication can spread useful information, but it can also create dependence that undermines information aggregation.

\noindent\textbf{(iii)} \emph{Evidence-based deliberation and disclosure} focuses on the exchange of reasons or private evidence, and is commonly organized into normative-democratic, social-choice-theoretic, and game-theoretic perspectives. A distinct literature asks when \emph{deliberation itself} is epistemically beneficial, emphasizing the exchange of evidence rather than merely the exchange of opinions. 
One strand, rooted in normative democratic theory, argues that inclusive discussion can improve truth-tracking by pooling dispersed information and enabling mutual criticism, while also stressing that these epistemic gains depend on institutional and informational conditions \citep{estlund2018epistemic}. A second strand, in social choice theory, develops formal ``wisdom of deliberating crowds'' results. For example, \citet{dietrichspiekermann2025} propose a probabilistic model of evidence-based opinion formation and derive deliberation-specific jury theorem results, clarifying mechanisms through which deliberation can improve truth-tracking while also identifying regimes in which it can reduce majority competence, especially through overcounting widely shared evidence. A third strand, in game theory, models deliberation as strategic or semi-strategic communication and disclosure. In particular, \citet{dingpivato2021} study round-by-round disclosure of private evidence by Bayesian agents who attempt to persuade others toward their current beliefs, and show that full information pooling arises only under specific structural conditions, including the presence of ``neutral'' agents and features of the decision environment. Relative to these evidence-centered models, our framework deliberately abstracts from reason-giving and evidence transmission in order to isolate a minimal baseline channel of social influence. 

\noindent\textbf{(iv)} \emph{Reduced-form local-update opinion dynamics} replace explicit belief or evidence exchange with simple neighbor-based rules, such as majority dynamics. In \emph{majority dynamics}, individuals repeatedly update by adopting the local majority among their neighbors before a final population-wide vote. \citet{mossel2014majority} study this process as a networked analogue of Condorcet-type aggregation: agents begin with i.i.d.\ private opinions that are slightly biased toward the correct alternative, interaction then induces correlation, and the question is whether the final vote selects the correct outcome with probability tending to one as the population grows. They show that the answer depends sharply on network structure. On sufficiently symmetric graphs, efficient aggregation persists. However, they also construct graph families for which interaction blocks aggregation, so that the probability of an incorrect outcome remains bounded away from zero. For expander graphs, they further show that if the initial bias is sufficiently strong relative to the degree and expansion properties of the graph, repeated majority updates can amplify that advantage and lead to eventual unanimity on the correct alternative. More broadly, these models suggest that collective accuracy depends not only on how votes are counted, but also on who influences whom before the vote. 

Our model differs from majority dynamics in two important respects: it uses an update rule based on imitation and incorporates a persistent exogenous influence in the form of committed leaders. This feature connects it naturally to the large opinion dynamics literature on zealotry, especially in statistical physics, where voter updating is studied under persistent external influence from \emph{committed} or \emph{stubborn} individuals, often called zealots or inflexibles, whose states do not change, or effectively do not change. Recent surveys distinguish early ``flexible'' formulations, such as agents who revert toward a preferred opinion at an idiosyncratic rate, from the now standard inflexible zealot idealization; they also distinguish zealotry in which all zealots support the same state from mixed zealotry, in which multiple zealot camps compete \citep{starnini2025opinion}.  In the two-state voter model with mixed zealotry on well-mixed populations, the presence of even a small number of zealots eliminates absorbing consensus and yields a nondegenerate stationary distribution \citep{chinellato2007, chinellato2015}. For equal zealot numbers, the stationary vote count distribution is approximately Gaussian with a width scaling like the inverse square root of the \emph{number of zealots}, rather than the electorate size, implying that a small committed set can prevent strong concentration even in large electorates \citep{chinellato2015,starnini2025opinion}. Related results extend to networks and to nonlinear or noisy voter variants, where mixed zealotry can generate qualitative changes in the stationary distribution, such as transitions between bimodal and unimodal regimes as zealotry density increases, while one-sided zealotry can drive eventual consensus but substantially affect convergence times \citep{starnini2025opinion}. Our contribution belongs to this zealot-driven class, but reframes the outcome metric in explicitly Condorcet-like terms: rather than asking only how zealots shape polarization or vote shares, we ask how majority accuracy changes, and on which topologies, as one moves from the CJT idealization of independent signals to the long-run distribution generated by competing committed camps under socially influenced opinion dynamics. 

Collectively, these literatures highlight both the promise and the limits of the classical CJT. Crowds are not automatically wise. When voters rely on common cues or learn by imitation, the electorate behaves less like a collection of independent individuals with distinct information and more like a single noisy or biased source, weakening the error-canceling benefit of independence. When influence flows through networks with highly uneven connectivity or concentrated power, communication can amplify incorrect information no less than correct information. The broader lesson is that majority accuracy is context-dependent, shaped by both the informational environment and the social structure through which influence propagates.

\section{Contributions and Roadmap}
\label{Roadmap}

Despite substantial progress in relaxing the assumptions of the CJT, tractable results that explicitly incorporate \emph{networked social influence}, while still yielding sharp characterizations of collective accuracy and its dependence on network structure, remain relatively rare. This paper addresses that gap by linking a Condorcet-type correctness criterion to a parsimonious dynamic model of pre-vote communication: a zealot--contrarian voter process in which otherwise uninformed voters update their vote intentions through imitation or counter-imitation in the presence of competing committed leaders.

Our first contribution is analytical. For a fully connected society, represented by a complete graph, we derive the stationary vote count distribution induced by the dynamics and obtain closed-form expressions for the key quantities governing majority accuracy, including the stationary mean and the correlation structure among free voters. We then extend the analysis to randomly connected electorates modeled as Erd\H{o}s--R\'enyi networks, in which each individual interacts with only a subset of the population. Using a mean-field approximation, we show that the stationary vote count distribution on such random networks coincides with that of the fully connected model after an appropriate rescaling of the parameters. Our second contribution concerns collective accuracy. Building on these distributional results, we establish explicit Condorcet-like guarantees by deriving conditions under which the nonasymptotic majority-correctness advantage holds. In the large electorate regime, we obtain a weak asymptotic guarantee for a fully conformist society. By contrast, for societies with even slight contrarian updating, we show that strict majority correctness converges to $0.5$, revealing an aggregation failure driven by social influence. We also identify parameter regimes in which social interaction lowers strict majority correctness relative to a natural ``no-deliberation'' benchmark that removes interpersonal influence among free voters, thereby establishing a second form of aggregation failure driven by social influence. For Erd\H{o}s--R\'enyi electorates, we further show, under the mean-field approximation validated numerically, that strict majority correctness coincides with that of the fully connected model after the same appropriate rescaling of parameters. 

Finally, we assess strict majority correctness across a broader range of network topologies, including scale-free, ring, and small-world networks, by simulation. The results reinforce the paper’s central mechanism. Sustained pre-vote interaction affects collective accuracy mainly through the dependence structure it generates among free voters. Relative to the Erd\H{o}s--R\'enyi benchmark, scale-free networks produce stronger positive voter--voter correlations and lower strict majority correctness. A natural interpretation is that network heterogeneity changes the channels through which committed-leader influence diffuses. Because highly connected voters are sampled, and hence imitated, disproportionately often, they can emerge as secondary opinion leaders, serving as locally prominent intermediaries in the spread of influence.  This hub-driven common source effect is especially consequential in the high-copying, strongly conformist regime, where repeated imitation generates stronger herding among free-voter states; as contrarian updating becomes more common, it disrupts alignment and weakens that dependence. In our setting, where the majority is taken over free voters and candidate~1 is favored by the committed environment, the stronger dependence induced by scale-free topology can reduce strict majority correctness by increasing the variance of the stationary free-vote count and thereby raising the probability that the majority outcome deviates from its mean. Ring networks and low-rewiring small-world networks show the opposite pattern. Their more local interaction structure weakens long-range coordination, lowers pairwise correlations, and yields higher strict majority correctness. Within the small-world family, this advantage is strongest at low rewiring and narrows as the rewiring probability increases, bringing the network closer to an Erd\H{o}s--R\'enyi graph. More broadly, the simulations numerically preserve the two qualitative comparisons established earlier: across all four topologies, a strict majority remains more accurate than a single post-deliberation free voter, but less accurate than the corresponding no-deliberation benchmark. Overall, these findings show that the same committed-leader environment can yield either higher or lower strict majority correctness depending on how network structure, together with the balance between conformist and contrarian updating, shapes vote correlations and thus the amount of effectively independent information that survives until the vote.

This paper is organized as follows. Section~\ref{voter} introduces the modeling framework used throughout the paper by presenting a zealot--contrarian voter model and characterizing its stationary behavior in fully mixed electorates for both the binary and $m$-alternative cases. Since our Condorcet-type correctness criterion is binary, namely whether the correct alternative~1 attains a strict majority, we develop the $m$-alternative extension only to the extent needed to derive the marginal distribution of the vote count for the correct alternative. We then consider the two-alternative model on randomly connected populations represented by Erd\H{o}s--R\'enyi networks. Building on this foundation, Section~\ref{majority} studies strict majority correctness in both fully mixed and randomly connected populations and establishes Condorcet-like inequalities for finite electorates as well as in the large electorate limit. Section~\ref{simulations} reports simulation results that assess and compare strict majority correctness across a broader range of topologies, including scale-free, small-world, and ring networks. Finally, Section~\ref{discussion} discusses the implications for deliberation-based voting systems and outlines directions for future research.

\section{The zealot--contrarian voter model on fully mixed populations}
\label{voter}

This section sets up the dynamical model of pre-vote social interaction used throughout the paper.
The goal is to obtain an explicit stationary law for the vote count process induced by repeated pairwise influence in the
presence of committed leaders (zealots) and anticonformist responses (contrarian copying). In later sections,
this stationary law will serve as the main input for our analysis of strict majority correctness on networks.

\subsection{Two alternatives}
\label{voter-two}

Consider a population of $n+\alpha_1+\alpha_2$ individuals choosing between two alternatives, labeled $1$ and $2$, in a collective decision (e.g., an election). Following the Condorcet Jury Theorem convention, we interpret alternative~$1$ as the objectively correct alternative. Only $n$ individuals are \emph{free voters} (undecided), while the remaining $\alpha_1$ and $\alpha_2$ are \emph{zealots} (committed influencers) supporting alternatives $1$ and $2$, respectively. Zealots never change their states, whereas free voters update their \emph{vote intentions} through social interaction; for brevity we refer to a free voter's current intention as its \emph{vote}. In network language, we refer to free voters as \emph{free nodes} and zealots as \emph{frozen nodes}. Although $\alpha_1$ and $\alpha_2$ are introduced as counts of zealots, all results below extend to arbitrary real $\alpha_1,\alpha_2>0$, in which case they are naturally interpreted as \emph{influence weights} (effective zealot counts) representing exogenous pressure in favor of each alternative. Each node has a \emph{state} (vote) in $\{1,2\}$.
Because the population is fully mixed (complete interaction graph), the system state at time $t$ is completely determined by
the number of free voters supporting alternative~$1$. Let $X_t\in\{0,1,\dots,n\}$ denote this number. Thus $X_t=k$ means that
$k$ free voters currently support alternative~$1$ and $n-k$ support alternative~$2$.

We use the following update rule.
At each time step, a free node is selected uniformly at random. With probability $\ell\in[0,1)$ (\emph{retention/inertia}) it keeps its current vote. With probability $1-\ell$ (\emph{social update}) it draws an \emph{influence source} as follows: each of the other $n-1$ free voters is selected with weight $1$, and two exogenous sources favoring alternatives $1$ and $2$ are selected with weights $\alpha_1$ and $\alpha_2$, respectively. When $\alpha_1,\alpha_2$ are integers, this is equivalent to sampling uniformly from the remaining $n-1+\alpha_1+\alpha_2$ individuals (free or frozen) in an augmented population that contains $\alpha_1$ and $\alpha_2$ frozen nodes. For non-integer $\alpha_i$, it is simply a weighted-sampling representation of external influence. If the sampled source supports alternative $s\in\{1,2\}$, then the updating voter's new vote is
\begin{equation*}
\text{new vote}=\begin{cases}
\text{copies vote $s$}, & \shortstack[l]{\text{w.p. } $r$\\\text{(conformist imitation)}},\\[0.25em]
\text{switches to vote $3-s$}, & \shortstack[l]{\text{w.p. } $1-r$\\\text{(contrarian response)}}.
\end{cases}
\end{equation*}
The parameter $r\in[0,1]$ tunes the society from fully contrarian ($r=0$) to purely imitative ($r=1$), with $r=\tfrac12$
corresponding to an unbiased ``coin-flip'' response to sampled votes.

Since only one free voter is updated at each step, if $X_t=k$ then $X_{t+1}\in\{k-1,k,k+1\}$. Writing
$D:=n-1+\alpha_1+\alpha_2$, the one-step transition probabilities are
\begin{align}
P_{k,k+1}
&=(1-\ell)\,\frac{n-k}{n}\cdot
\frac{r(k+\alpha_1)+(1-r)(n-k-1+\alpha_2)}{D}, \label{eq:P_up_contrarian}\\
P_{k,k-1}
&=(1-\ell)\,\frac{k}{n}\cdot
\frac{r(n-k+\alpha_2)+(1-r)(k-1+\alpha_1)}{D}, \label{eq:P_down_contrarian}
\end{align}
and $P_{k,k}=1-P_{k,k+1}-P_{k,k-1}$ (with the natural boundary conventions). For example, \eqref{eq:P_up_contrarian}
is the probability that the updated voter is currently in state~$2$, $(n-k)/n$, times the probability of a social update,
$(1-\ell)$, times the probability that the update results in adopting state~$1$: either by imitating a sampled $1$-node
(with probability $r(k+\alpha_1)/D$) or by responding contrarily to a sampled $2$-node (with probability $(1-r)(n-k-1+\alpha_2)/D$).

For $\alpha_1,\alpha_2>0$ and $\ell<1$, we have $P_{k,k+1}>0$ for $k<n$ and $P_{k,k-1}>0$ for $k>0$, so the chain has no absorbing states.
Since $P_{k,k}>0$ for all $k=0,1,\dots,n$, the chain is aperiodic. Moreover, every state communicates with every other through $\pm1$ moves
(irreducibility), and therefore the chain admits a unique stationary distribution.

We next establish the stationary distribution for the binary zealot--contrarian voter model on fully mixed populations.

\begin{theorem}
\label{thm:stationary-binary-contrarian}
Let $n\in\mathbb{N}$, $\alpha_1,\alpha_2>0$, $\ell\in[0,1)$, and $r\in[0,1]$. Under the above dynamics,
$\{X_t\}$ is an irreducible and aperiodic birth--death chain on $\{0,1,\dots,n\}$ with a unique stationary distribution
$\pi(k)=\mathbb{P}(X=k)$, independent of $\ell$.

\smallskip
\noindent\emph{(i) The unbiased case $r=\tfrac12$.}
In stationarity,
\[
X\sim \mathrm{Binomial}\!\left(n,\tfrac12\right),
\qquad
\pi(k)=2^{-n}\binom{n}{k},\qquad k=0,1,\dots,n.
\]

\smallskip
\noindent\emph{(ii) The biased case $r\neq\tfrac12$.}
Let $a:=2r-1$ and define
\begin{equation}
\label{eq:def-theta-beta_maintext}
\theta:=\frac{r\alpha_1+(1-r)(n-1+\alpha_2)}{a},
\qquad
\beta:=\frac{(1-r)\alpha_1+r(n-1+\alpha_2)}{a}-n+1.
\end{equation}
Then the unique stationary distribution is
\begin{equation}
\label{eq:pi-binary-contrarian}
\pi(k)=\binom{n}{k}\,\frac{(\theta)_k\,(\beta)_{n-k}}{(\theta+\beta)_n},
\qquad k=0,1,\dots,n,
\end{equation}
where $(x)_m:=x(x+1)\cdots(x+m-1)$ is the Pochhammer (rising) factorial, with $(x)_0:=1$. Equivalently, using $(x)_m=m!\binom{x+m-1}{m}$,
\[
\pi(k)=
\frac{\displaystyle \binom{\theta+k-1}{k}\,\binom{\beta+n-k-1}{n-k}}
{\displaystyle \binom{\theta+\beta+n-1}{n}},
\qquad k=0,1,\dots,n,
\]
with generalized binomial coefficients extended to real parameters via the gamma function.

In particular, when $r>\tfrac12$ (conformist dominant), the parameters $\theta$ and $\beta$ are positive and \eqref{eq:pi-binary-contrarian} is a genuine Beta--Binomial distribution. When $r<\tfrac12$ (contrarian dominant), the same closed form remains a valid stationary count distribution, but it is not a genuine Beta--Binomial distribution. For $r=\tfrac12$, the stationary distribution is $\mathrm{Binomial}(n,\tfrac12)$ as in part~(i).
\end{theorem}

\begin{proof}
A key simplification is that the stationary law is independent of the inertia parameter $\ell$.
Let $\pi_t(k)=\mathbb{P}(X_t=k)$ and write $\pi_t$ for the row vector with components $\pi_t(0),\dots,\pi_t(n)$.
From the update rule,
\[
\pi_{t+1}=\ell\pi_t+(1-\ell)\pi_t Q=\pi_t[\ell I+(1-\ell)Q]\equiv \pi_t P,
\]
where $Q$ is the \emph{social-update} kernel (set $\ell=0$ in \eqref{eq:P_up_contrarian}--\eqref{eq:P_down_contrarian}).
Since $P=\ell I+(1-\ell)Q$, any stationary distribution of $Q$ is also stationary for $P$, and hence the stationary law does not depend on $\ell$.
Thus it suffices to determine the stationary distribution of $Q$.

Since $Q$ is a birth--death kernel, stationarity is equivalent to detailed balance:
\[
\pi(k)\,Q_{k,k+1}=\pi(k+1)\,Q_{k+1,k}\qquad (k=0,1,\dots,n-1).
\]
Using \eqref{eq:P_up_contrarian}--\eqref{eq:P_down_contrarian} with $\ell=0$ yields
\begin{equation}
\label{eq:ratio_general_maintext}
\frac{\pi(k+1)}{\pi(k)}
=\frac{Q_{k,k+1}}{Q_{k+1,k}}
=\frac{n-k}{k+1}\cdot
\frac{r(k+\alpha_1)+(1-r)(n-k-1+\alpha_2)}
{r(n-k-1+\alpha_2)+(1-r)(k+\alpha_1)}.
\end{equation}

If $r=\tfrac12$, the fraction in \eqref{eq:ratio_general_maintext} equals $1$, so $\pi(k+1)/\pi(k)=(n-k)/(k+1)$ and hence
$\pi(k)\propto \binom{n}{k}$, which normalizes to $\pi(k)=2^{-n}\binom{n}{k}$.

Assume $r\neq\tfrac12$, recalling that $a=2r-1$. With $\theta,\beta$ as in \eqref{eq:def-theta-beta_maintext},
a direct rearrangement gives
\begin{align*}
r(k+\alpha_1)+(1-r)(n-k-1+\alpha_2) &= a\,(k+\theta),\\
r(n-k-1+\alpha_2)+(1-r)(k+\alpha_1) &= a\,(\beta+n-k-1).
\end{align*}
Substituting into \eqref{eq:ratio_general_maintext} yields
\[
\frac{\pi(k+1)}{\pi(k)}
=\frac{n-k}{k+1}\cdot\frac{k+\theta}{\beta+n-k-1}.
\]
Solving this first-order recurrence yields, for $k\ge 1$,
\[
\pi(k)=\pi(0)\prod_{i=0}^{k-1}\frac{\pi(i+1)}{\pi(i)}
=\pi(0)\prod_{i=0}^{k-1}\frac{n-i}{i+1}\cdot\prod_{i=0}^{k-1}\frac{i+\theta}{\beta+n-i-1}.
\]
The first product equals $\binom{n}{k}$, and the second can be written in rising factorial form as
\[
\prod_{i=0}^{k-1}\frac{i+\theta}{\beta+n-i-1}
=\frac{(\theta)_k}{(\beta+n-k)_k}
=(\theta)_k\,\frac{(\beta)_{n-k}}{(\beta)_n},
\]
where the last identity uses $(\beta)_n=(\beta)_{n-k}(\beta+n-k)_k$. Hence
\[
\pi(k)=\pi(0)\binom{n}{k}\,(\theta)_k\,\frac{(\beta)_{n-k}}{(\beta)_n}.
\]
Finally, normalizing via Vandermonde's identity in Pochhammer form,
\[
\sum_{k=0}^{n}\binom{n}{k}\,(\theta)_k\,(\beta)_{n-k}=(\theta+\beta)_n,
\]
yields $\pi(0)=(\beta)_n/(\theta+\beta)_n$ and hence \eqref{eq:pi-binary-contrarian}.
\end{proof}

The stationary law \eqref{eq:pi-binary-contrarian} summarizes the long-run distribution of vote counts after sustained social influence in a fully mixed electorate with committed leaders. In the conformist dominant regime $r>\tfrac12$, $\theta,\beta>0$ and one may equivalently represent the stationary count as $X\mid\Theta\sim\mathrm{Bin}(n,\Theta)$ with $\Theta\sim\mathrm{Beta}(\theta,\beta)$.
In contrast, when $r<\tfrac12$ the model exhibits negative pairwise correlations (underdispersion), so $X$ cannot arise from any mixture of independent Bernoulli votes. As shown in Section~\ref{voter-moments}, \eqref{eq:pi-binary-contrarian} determines both the individual success probability of a free voter (i.e., the probability of voting for the correct alternative~$1$ after deliberation) and the dependence structure among free voters' votes. In particular, the parameter $r$ controls the sign and strength of pairwise correlations: conformist societies ($r>\tfrac12$) generate positive correlations among free voters, contrarian societies ($r<\tfrac12$) generate negative correlations, and $r=\tfrac12$ yields independence in stationarity. These dependence effects will be central when we compare strict majority correctness across interaction networks.


\subsection{Effective competence and pairwise correlation}
\label{voter-moments}

For Condorcet-type questions, two features of the stationary vote profile are especially relevant:
the \emph{individual success probability} of a free voter---i.e., the probability of voting for the correct alternative after the deliberation process---and the \emph{dependence} across voters.
In our setting, both are available in closed form from the stationary law of the vote count $X$ in Theorem~\ref{thm:stationary-binary-contrarian}.
Let $X$ denote the number of free voters who vote for the correct alternative~$1$ in stationarity. Define $p:=\E[X]/n$. Equivalently, $p$ is the marginal probability that a uniformly chosen free voter votes for the correct alternative.
We refer to $p$ as the free voter's ``effective competence'' induced by sustained interaction, while the stationary pairwise correlation captures how social influence amplifies or counteracts aggregation through majority rule.

\medskip

\noindent The next corollary gives closed forms for both the effective competence $p$ and the stationary pairwise correlation between distinct free voters.

\begin{corollary}
\label{cor:moments-corr}
Under the assumptions of Theorem~\ref{thm:stationary-binary-contrarian}, define
\begin{align*}
Y_i &:= \mathbf{1}\{\text{free voter $i$ votes for the correct alternative $1$ in stationarity}\},\\
X &:= \sum_{i=1}^n Y_i.
\end{align*}

\smallskip
\noindent\emph{(i) Effective competence.} In stationarity, the marginal success probability $p:=\Pp(Y_i=1)$ satisfies
\begin{equation}
\label{eq:p-mean}
p=\frac{r\alpha_1+(1-r)(n-1+\alpha_2)}{\alpha_1+\alpha_2+2(1-r)(n-1)}.
\end{equation}
Equivalently, $\E[X]=np$.

\smallskip
\noindent\emph{(ii) Pairwise correlation.}
For any two distinct free voters $i\neq j$,
\begin{equation}
\label{eq:rho}
\Corr(Y_i,Y_j)=
\begin{cases}
0, & r=\tfrac12,\\[6pt]
\dfrac{2r-1}{\alpha_1+\alpha_2+2n-3-2r(n-2)}, & r\neq\tfrac12.
\end{cases}
\end{equation}
In particular, for $n\ge 2$ and $\alpha_1,\alpha_2>0$, the correlation is positive iff $r>\tfrac12$.
\end{corollary}

\begin{proof}
If $r=\tfrac12$, Theorem~\ref{thm:stationary-binary-contrarian}(i) gives $X\sim\mathrm{Binomial}(n,\tfrac12)$, so the votes
$\{Y_i\}$ are i.i.d.\ Bernoulli$(\tfrac12)$ and hence $\Corr(Y_i,Y_j)=0$.

Assume $r\neq\tfrac12$. Because the interaction graph is complete, the stationary law of $(Y_1,\dots,Y_n)$ is exchangeable.
Let $p:=\Pp(Y_i=1)=\E[Y_i]$ denote the \emph{effective competence} of a free voter, and let $q:=\E[Y_iY_j]$ for $i\neq j$. Then
\begin{equation}
\label{eq:exchangeability_identities}
\E[X]=\sum_{i=1}^n \E[Y_i]=np,
\qquad
\E[X(X-1)]=\sum_{i\neq j}\E[Y_iY_j]=n(n-1)q.
\end{equation}
The first identity expresses the mean vote count $\E[X]$ in terms of the effective competence $p$, while the second links the pairwise joint success probability $q$ to $\E[X(X-1)]$.

Under Theorem~\ref{thm:stationary-binary-contrarian}(ii), $X$ has the Beta--Binomial form \eqref{eq:pi-binary-contrarian} with parameters $(\theta,\beta)$, so its
mean vote count $\E[X]$ and second falling factorial moment $\E[X(X-1)]$ are
\begin{equation}
\label{eq:factorial_moments_BB_main}
\E[X]=\frac{n\theta}{\theta+\beta},
\qquad
\E[X(X-1)]=\frac{n(n-1)\theta(\theta+1)}{(\theta+\beta)(\theta+\beta+1)}.
\end{equation}
Combining \eqref{eq:exchangeability_identities}--\eqref{eq:factorial_moments_BB_main} yields the effective competence $p$ and the joint success probability $q$:
\[
p=\frac{\theta}{\theta+\beta},
\qquad
q=\frac{\theta(\theta+1)}{(\theta+\beta)(\theta+\beta+1)}.
\]
Since $Y_i$ is Bernoulli$(p)$, $\Var(Y_i)=p(1-p)$. Moreover,
\[
\Cov(Y_i,Y_j)=\E[Y_iY_j]-\E[Y_i]\E[Y_j]=q-p^2.
\]
Therefore
\[
\Corr(Y_i,Y_j)=\frac{q-p^2}{p(1-p)}=\frac{1}{\theta+\beta+1}.
\]
Finally, using \eqref{eq:def-theta-beta_maintext} to express $\theta+\beta+1$ and the effective competence $p=\theta/(\theta+\beta)$ in terms of $(r,\alpha_1,\alpha_2,n)$, we obtain after some algebra \eqref{eq:rho} and \eqref{eq:p-mean}.
\end{proof}

\subsection{Multiple alternatives and reduction to the correct alternative (\texorpdfstring{$m\ge 2$}{m>=2})}
\label{voter-m}

While the strict majority correctness results developed later are framed as a binary criterion (the electorate either selects the correct alternative~$1$ or it does not), it is sometimes natural to allow for more than two competing alternatives during deliberation. Importantly, the event that alternative~$1$ wins a strict majority among free voters depends only on the stationary count $X_1$ of free votes for alternative~$1$, regardless of how the remaining votes are split among alternatives $2,\dots,m$. The multialternative formulation below therefore serves primarily as a bookkeeping device: it yields a closed-form joint stationary law and, as an immediate consequence, a one-dimensional marginal distribution for $X_1$ that is sufficient for computing strict majority correctness of the correct alternative.

Consider a fully mixed electorate with $n$ free voters and $\alpha_i>0$ zealots
committed to alternative $i\in\{1,\dots,m\}$. As before, the parameters $\alpha_i$ need not be integers; for non-integer $\alpha_i$ they can be interpreted as influence weights (effective zealot strengths) entering the weighted-sampling rule below. Let $\alpha_0:=\sum_{i=1}^m \alpha_i$ and let
\[
X_t=(X_{t,1},\dots,X_{t,m})\in \mathcal{S}_n:=\Big\{k\in\mathbb{Z}_{\ge0}^m:\ \sum_{i=1}^m k_i=n\Big\}
\]
denote the vector of free-voter counts for each alternative at time $t$.

The update rule is the natural multistate analogue. At each time step, a free voter is chosen uniformly. With probability $\ell$
it retains its current choice. Otherwise it samples a second node uniformly from the remaining $n-1+\alpha_0$ individuals.
If the sampled node supports alternative $s$, then the updating voter imitates $s$ with probability $r$, and with probability $1-r$
it adopts a \emph{different} alternative chosen uniformly from the remaining $m-1$ alternatives. (For $m=2$ this reduces to the binary flip rule.)

\noindent The following theorem gives the stationary law of the $m$-alternative zealot--contrarian voter model.

\begin{theorem}
\label{thm:stationary-m-contrarian}
Let $m\ge 2$, $n\in\mathbb{N}$, and $\alpha_i>0$ for $i=1,\dots,m$. Under the above dynamics, $\{X_t\}$ is an irreducible,
aperiodic Markov chain on $\mathcal{S}_n$ with a unique stationary distribution independent of $\ell$.

\smallskip
\noindent\emph{(i) The unbiased case $r=\tfrac1m$.}
In stationarity,
\[
\pi(k_1,\dots,k_m)=\frac{n!}{k_1!\cdots k_m!}\left(\frac1m\right)^n,
\qquad (k_1,\dots,k_m)\in\mathcal{S}_n,
\]
i.e.\ \(X\sim \mathrm{Multinomial}\big(n;\tfrac1m,\dots,\tfrac1m\big)\).

\smallskip
\noindent\emph{(ii) The biased case $r\neq\tfrac1m$.}
Let $T:=n-1+\alpha_0$ and define
\[
\delta:=\frac{(1-r)T}{mr-1},
\qquad
\widetilde{\alpha}_i:=\alpha_i+\delta,
\qquad
\widetilde{\alpha}_0:=\sum_{i=1}^m \widetilde{\alpha}_i=\alpha_0+m\delta.
\]
Then the unique stationary distribution has the Dirichlet--multinomial form:
\begin{equation}
\label{eq:pi-m-contrarian}
\pi(k_1,\dots,k_m)=
\frac{\displaystyle\prod_{i=1}^m \binom{\widetilde{\alpha}_i+k_i-1}{k_i}}
{\displaystyle \binom{\widetilde{\alpha}_0+n-1}{n}},
\qquad (k_1,\dots,k_m)\in\mathcal{S}_n,
\end{equation}
where generalized binomial coefficients are understood in the gamma function sense. Equivalently,
\[
\pi(k_1,\dots,k_m)
=\frac{n!}{k_1!\cdots k_m!}\,
\frac{(\widetilde{\alpha}_1)_{k_1}\cdots(\widetilde{\alpha}_m)_{k_m}}{(\widetilde{\alpha}_0)_{n}}.
\]
\end{theorem}

\smallskip
\noindent The joint stationary law \eqref{eq:pi-m-contrarian} characterizes the stationary probability distribution of the full post-deliberation vote share vector when there are $m$ alternatives. When $mr>1$ (so $\delta>0$ and $\widetilde{\alpha}_i>0$), \eqref{eq:pi-m-contrarian} is a standard Dirichlet--multinomial distribution, with the usual Dirichlet mixture interpretation. When $mr<1$, the same closed form remains a valid stationary distribution, although some $\widetilde{\alpha}_i$ may be nonpositive and the mixture interpretation no longer applies. A proof of Theorem~\ref{thm:stationary-m-contrarian} is given in Appendix~\ref{app:proof-thm2}.

\smallskip
In this paper, however, our Condorcet-like correctness criterion is binary: the probability that a strict majority of free voters selects the correct alternative~$1$ is $\Pp\!\left(X_1 \ge (n+1)/2\right)$, which depends on the multistate model only through the one-dimensional marginal law of $X_1$. Accordingly, the only ingredient from the multistate model used in our analysis is the stationary marginal distribution of $X_1$; the full multistate characterization in \eqref{eq:pi-m-contrarian} is included for completeness and for potential use with other aggregation rules. The following corollary gives this marginal. It has the same Beta--Binomial \emph{form} as in the binary case (Theorem~\ref{thm:stationary-binary-contrarian}), after aggregating alternatives $2,\dots,m$ into ``other''.

\begin{corollary}
\label{cor:bb-marginal-correct}
Under the assumptions of Theorem~\ref{thm:stationary-m-contrarian}, let $X_1$ denote the number of free voters supporting the
correct alternative~$1$ in stationarity. Then
\[
\Pp(X_1=k)=\binom{n}{k}\,\frac{(\widetilde{\alpha}_1)_k\,\big(\widetilde{\alpha}_0-\widetilde{\alpha}_1\big)_{n-k}}{(\widetilde{\alpha}_0)_n},
\qquad k=0,1,\dots,n,
\]
which is only a genuine Beta--Binomial distribution when $mr>1$ (so $\widetilde{\alpha}_1>0$ and $\widetilde{\alpha}_0-\widetilde{\alpha}_1>0$), and otherwise remains a valid stationary distribution (without a Beta mixture interpretation).
\end{corollary}

\subsection{The zealot--contrarian voter model on Erd\H{o}s--R\'enyi networks}
\label{voter-er}

All results so far assumed a well-mixed population, so that at each pre-vote interaction step a free voter draws its influence source uniformly from the rest of the electorate.
We now relax this assumption by allowing interactions among free voters to be \emph{local}, while keeping zealot influence \emph{global}.
Specifically, the free--free connections form an Erd\H{o}s--R\'enyi (ER) graph $G(n,p)$ on the $n$ free voters: each pair of free voters is linked independently with probability $p$, so the expected free degree is $d=p(n-1)$. We treat the free--free edge set as fixed throughout the dynamics. Each free voter is additionally connected to all zealots (equivalently, zealots act as global influence sources available to every free voter).
As in Section~\ref{voter-two}, the parameters $\alpha_1$ and $\alpha_2$ may be non-integers and are then interpreted as influence weights entering the (weighted) neighbor-sampling step.

We focus on the binary case $m=2$, which is the setting relevant for our strict majority correctness criterion. As in Section~\ref{voter-two}, let $X_t\in\{0,1,\dots,n\}$ denote the number of free voters currently supporting the correct alternative~$1$ (so $n-X_t$ support alternative~$2$). At each time step, a free node is selected uniformly at random. With probability $\ell\in[0,1)$ it retains its current vote, and with probability $1-\ell$ it samples \emph{uniformly from its neighbors} (free or zealot). If the sampled neighbor supports alternative $s\in\{1,2\}$, then the updating voter
copies $s$ with probability $r$ and switches to $3-s$ with probability $1-r$.

Fix $k$ and condition on $X_t=k$; for readability, we suppress this conditioning throughout. Under an ER mean-field approximation, we replace the composition of a focal voter's free-neighbor set by its expectation. Thus, for a focal voter in state~$2$, the expected fraction of free neighbors in state~$1$ is approximately $k/(n-1)$, whereas for a focal voter already in state~$1$ the corresponding fraction is approximately $(k-1)/(n-1)$, since the focal voter itself is excluded. Because the actual neighbor draw includes both free neighbors and globally connected zealots, it is convenient to introduce explicit sampling probabilities. In particular,
$\Pp\bigl(2\ \text{samples}\ 1\bigr)$ denotes the probability that a focal free voter currently holding alternative~$2$ samples a neighbor (free or zealot) holding alternative~$1$, and
$\Pp\bigl(1\ \text{samples}\ 1\bigr)$ denotes the analogous probability when the focal voter currently holds alternative~$1$. Under the ER mean-field approximation, these sampling probabilities are approximated by
\begin{equation}
\label{eq:er-s1}
\begin{aligned}
\Pp\bigl(2\ \text{samples}\ 1\bigr)
&\approx \frac{d\,\frac{k}{n-1}+\alpha_1}{d+\alpha_1+\alpha_2},\\
\Pp\bigl(1\ \text{samples}\ 1\bigr)
&\approx \frac{d\,\frac{k-1}{n-1}+\alpha_1}{d+\alpha_1+\alpha_2}.
\end{aligned}
\end{equation}
It is convenient to rewrite \eqref{eq:er-s1} in the complete graph form by introducing the \emph{effective} zealot parameters
\begin{equation}
\label{eq:alpha-eff}
\alpha_1^{\mathrm{eff}}:=\frac{n-1}{d}\,\alpha_1,
\qquad
\alpha_2^{\mathrm{eff}}:=\frac{n-1}{d}\,\alpha_2,
\qquad
\alpha_0^{\mathrm{eff}}:=\alpha_1^{\mathrm{eff}}+\alpha_2^{\mathrm{eff}}.
\end{equation}
Then \eqref{eq:er-s1} becomes
\begin{equation}
\label{eq:er-s1-eff}
\begin{aligned}
\Pp\bigl(2\ \text{samples}\ 1\bigr)
&\approx \frac{k+\alpha_1^{\mathrm{eff}}}{(n-1)+\alpha_0^{\mathrm{eff}}},\\
\Pp\bigl(1\ \text{samples}\ 1\bigr)
&\approx \frac{k-1+\alpha_1^{\mathrm{eff}}}{(n-1)+\alpha_0^{\mathrm{eff}}}.
\end{aligned}
\end{equation}

Given a sampled neighbor, the contrarian channel implies
\[
\Pp(\text{adopt $1$}\mid \text{sample has $1$})=r,
\qquad
\Pp(\text{adopt $1$}\mid \text{sample has $2$})=1-r.
\]
Hence, the probability that a focal voter in state~$2$ \emph{adopts} alternative~$1$ (i.e., updates from $2$ to $1$) is
\begin{equation}
\label{eq:er-u1}
\begin{aligned}
\Pp\bigl(2\ \text{adopts}\ 1\bigr)
&:= r\,\Pp\bigl(2\ \text{samples}\ 1\bigr) \\
&\quad +(1-r)\Bigl(1-\Pp\bigl(2\ \text{samples}\ 1\bigr)\Bigr) \\
&= (1-r)+(2r-1)\,\Pp\bigl(2\ \text{samples}\ 1\bigr).
\end{aligned}
\end{equation}
while the probability that a focal voter in state~$1$ switches to alternative~$2$ (i.e., updates from $1$ to $2$) is
\begin{equation}
\label{eq:er-u2}
\begin{aligned}
\Pp\bigl(1\ \text{adopts}\ 2\bigr)
&:= r\Bigl(1-\Pp\bigl(1\ \text{samples}\ 1\bigr)\Bigr) \\
&\quad +(1-r)\,\Pp\bigl(1\ \text{samples}\ 1\bigr) \\
&= r-(2r-1)\,\Pp\bigl(1\ \text{samples}\ 1\bigr).
\end{aligned}
\end{equation}
Let $Q^{\mathrm{ER}}$ denote the social-update kernel (set $\ell=0$). Since one free node is updated per step, the induced ER mean-field chain for $X_t$ is birth--death with
\begin{equation}
\label{eq:er-Q-bd}
\begin{aligned}
Q^{\mathrm{ER}}_{k,k+1} &= \frac{n-k}{n}\,\Pp\bigl(2\ \text{adopts}\ 1\bigr),\\
Q^{\mathrm{ER}}_{k,k-1} &= \frac{k}{n}\,\Pp\bigl(1\ \text{adopts}\ 2\bigr),\\
Q^{\mathrm{ER}}_{k,k}   &= 1-Q^{\mathrm{ER}}_{k,k+1}-Q^{\mathrm{ER}}_{k,k-1}.
\end{aligned}
\end{equation}
Substituting \eqref{eq:er-s1-eff} into \eqref{eq:er-u1}--\eqref{eq:er-u2} yields the explicit complete graph form
\begin{align}
\label{eq:er-Q-closed-up}
Q^{\mathrm{ER}}_{k,k+1}
&\approx \frac{n-k}{n}\cdot
\frac{r(k+\alpha_1^{\mathrm{eff}})+(1-r)(n-k-1+\alpha_2^{\mathrm{eff}})}
{(n-1)+\alpha_0^{\mathrm{eff}}},
\\
\label{eq:er-Q-closed-down}
Q^{\mathrm{ER}}_{k,k-1}
&\approx \frac{k}{n}\cdot
\frac{r(n-k+\alpha_2^{\mathrm{eff}})+(1-r)(k-1+\alpha_1^{\mathrm{eff}})}
{(n-1)+\alpha_0^{\mathrm{eff}}}.
\end{align}
Comparing \eqref{eq:er-Q-closed-up}--\eqref{eq:er-Q-closed-down} with \eqref{eq:P_up_contrarian}--\eqref{eq:P_down_contrarian} shows that, under the ER mean-field closure, the macroscopic count dynamics on an ER graph coincide algebraically with the complete graph dynamics after the replacements
\[
(\alpha_1,\alpha_2)\mapsto(\alpha_1^{\mathrm{eff}},\alpha_2^{\mathrm{eff}}).
\]
Since the full-step kernel satisfies $P^{\mathrm{ER}}=\ell I+(1-\ell)Q^{\mathrm{ER}}$, the stationary distribution is again independent of $\ell$.

Equations \eqref{eq:er-Q-closed-up}--\eqref{eq:er-Q-closed-down} show that, under the ER mean-field closure, the induced birth--death dynamics for $X_t$ on an ER graph have the same algebraic form as the fully mixed dynamics \eqref{eq:P_up_contrarian}--\eqref{eq:P_down_contrarian} after the parameter substitution
\[
(\alpha_1,\alpha_2)\mapsto(\alpha_1^{\mathrm{eff}},\alpha_2^{\mathrm{eff}}),
\qquad
(\alpha_1^{\mathrm{eff}},\alpha_2^{\mathrm{eff}})\ \text{given by \eqref{eq:alpha-eff}.}
\]
Consequently, all stationary quantities from the fully mixed model carry over to the ER case by the same substitution. In particular, the stationary distribution of $X$ under this approximation is given by \eqref{eq:pi-binary-contrarian} with $(\alpha_1,\alpha_2)$ replaced by $(\alpha_1^{\mathrm{eff}},\alpha_2^{\mathrm{eff}})$ (and reduces to $\mathrm{Binomial}(n,\tfrac12)$ when $r=\tfrac12$). For $r\neq\tfrac12$, with $a=2r-1$ as above, and
\[
\begin{aligned}
\theta^{\mathrm{ER}} &:=\frac{r\alpha_1^{\mathrm{eff}}+(1-r)(n-1+\alpha_2^{\mathrm{eff}})}{a},\\
\beta^{\mathrm{ER}}  &:=\frac{(1-r)\alpha_1^{\mathrm{eff}}+r(n-1+\alpha_2^{\mathrm{eff}})}{a}-n+1,
\end{aligned}
\]
we obtain
\[
\begin{aligned}
\Pp(X=k) &\approx \binom{n}{k}\,
\frac{(\theta^{\mathrm{ER}})_k\,(\beta^{\mathrm{ER}})_{n-k}}
{(\theta^{\mathrm{ER}}+\beta^{\mathrm{ER}})_n},
\qquad k=0,1,\dots,n.
\end{aligned}
\]
Likewise, the pairwise correlation between two distinct free voters in the ER mean-field approximation is obtained by substituting $(\alpha_1^{\mathrm{eff}},\alpha_2^{\mathrm{eff}})$ into the fully mixed expression \eqref{eq:rho}; in particular, for $r\neq\tfrac12$,
\begin{equation}
\label{corr_er}
\Corr(Y_i,Y_j)\ \approx\ \frac{2r-1}{\alpha_1^{\mathrm{eff}}+\alpha_2^{\mathrm{eff}}+2n-3-2r(n-2)}.
\end{equation}
Equation \eqref{corr_er} provides a simple analytic prediction for ER networks, and in Section~\ref{simulations} we assess its accuracy against simulations.

The mapping \eqref{eq:alpha-eff} shows that, relative to the complete graph case ($d=n-1$), global zealot influence is amplified by a factor of $(n-1)/d$ when free--free interactions are sparse. As with any mean-field reduction, the ER approximation \eqref{eq:er-s1} averages over degree heterogeneity and local clustering, and is therefore most reliable for sufficiently dense ER graphs. Performance can degrade on heterogeneous networks, where degree fluctuations alter influence, and on strongly structured topologies (such as rings), where microstates with the same aggregate count $X_t$ need not evolve similarly. Section~\ref{simulations} quantifies these deviations.

\section{Strict majority correctness}
\label{majority}

\subsection{A finite electorate Condorcet-type inequality on fully mixed populations}
\label{sec:majority-complete-finite}

Classical versions of the CJT compare the probability that a majority reaches the correct decision
with the correctness probability of a single voter under independent Bernoulli voting. In the present model, by contrast,
votes are generated endogenously by a pre-vote social influence process with zealots and contrarian responses, so the votes
of free voters are generally dependent in stationarity. The natural Condorcet-type benchmark is therefore the stationary
marginal correctness of a single free voter after the campaign phase, not an exogenous i.i.d.\ competence parameter.
If $X$ denotes the stationary number of free voters supporting the correct alternative~$1$, then this benchmark is
\[
p=\frac{\E[X]}{n},
\]
which is exactly the effective competence in \eqref{eq:p-mean}. The corresponding strict majority success probability is
\[
M_n:=\Pp\!\left(X\ge \frac{n+1}{2}\right)
\]
for odd $n$.

The question is then directly in the spirit of epistemic social choice: whenever a randomly chosen free voter is more likely
than not to support the correct alternative after deliberation, does majority rule improve on that post-deliberation benchmark?

Before turning to the theorem, Fig.~\ref{fig:Mn_minus_p_vs_r_combined} provides a numerical illustration. For odd $n$, it plots $M_n-p$ as a function of $r$ for $n\in\{5,11,31,51\}$ and for three representative
pairs $(\alpha_1,\alpha_2)$. These parameter choices are not intended as an exhaustive exploration; rather, they were selected to span low, intermediate, and stronger zealot influence levels, guided by earlier results for the special case $r=1$, where mixed zealotry was shown to generate qualitatively different stationary laws, including a flat distribution at $\alpha_1=\alpha_2=1$ and increasingly concentrated unimodal behavior for larger values \citep{chinellato2007,chinellato2015}. Accordingly, $(\alpha_1,\alpha_2)=(0.75,0.3)$ corresponds to a low-influence regime below those thresholds, $(2.5,1)$ to an intermediate near-threshold regime, and $(5,2)$ to a stronger-influence regime.
Moreover, by \eqref{eq:p-mean},
\[
p-\frac12
=
\frac{(2r-1)(\alpha_1-\alpha_2)}
{2\bigl[\alpha_1+\alpha_2+2(1-r)(n-1)\bigr]},
\]
so for the cases shown, where $\alpha_1>\alpha_2$, we have $p>\tfrac12$ if and only if $r>\tfrac12$.
The plots therefore make the theoretical comparison visually transparent: across the plotted range,
$M_n-p$ is positive whenever $p>\tfrac12$ and negative whenever $p<\tfrac12$. They also suggest that,
over the parameter range shown, the majority advantage $M_n-p$ becomes larger as the zealot influence weights
increase. In contrast to the classical CJT, the majority advantage $M_n-p$ need not increase monotonically as the electorate size $n$ grows.

\begin{figure}[!htbp]
\centering
\setlength{\abovecaptionskip}{4pt}
\setlength{\belowcaptionskip}{0pt}

\begin{minipage}[t]{0.49\textwidth}
\centering
\includegraphics[width=\linewidth,trim=78 205 95 210,clip]{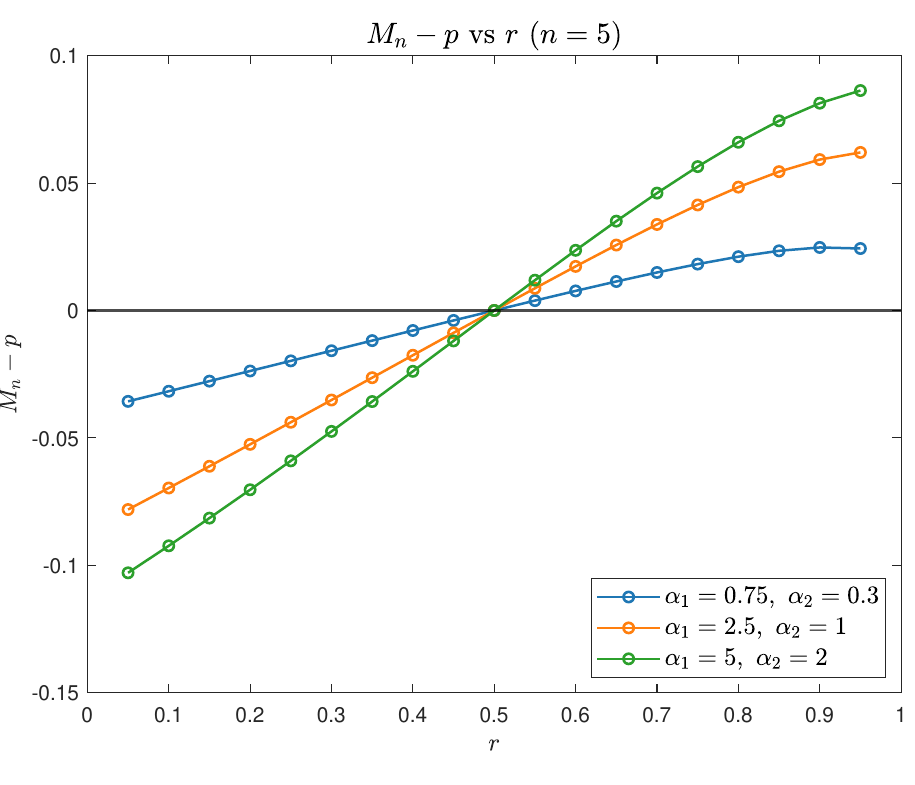}
\end{minipage}
\hfill
\begin{minipage}[t]{0.49\textwidth}
\centering
\includegraphics[width=\linewidth,trim=78 205 95 210,clip]{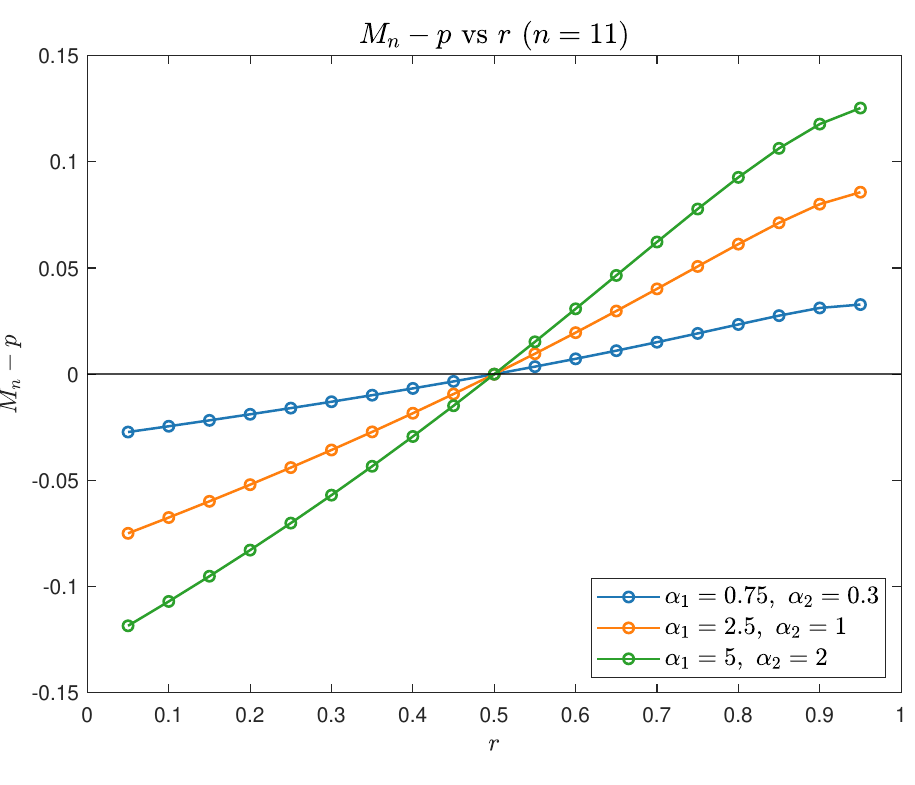}
\end{minipage}

\vspace{-0.35em}

\begin{minipage}[t]{0.49\textwidth}
\centering
\includegraphics[width=\linewidth,trim=78 205 95 210,clip]{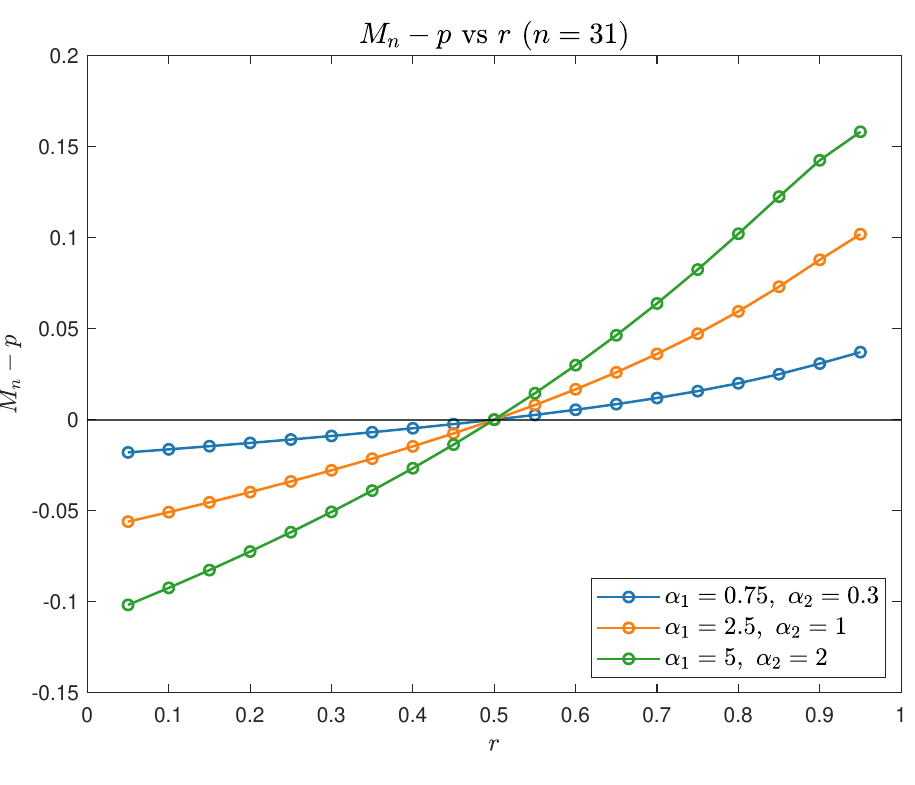}
\end{minipage}
\hfill
\begin{minipage}[t]{0.49\textwidth}
\centering
\includegraphics[width=\linewidth,trim=78 205 95 210,clip]{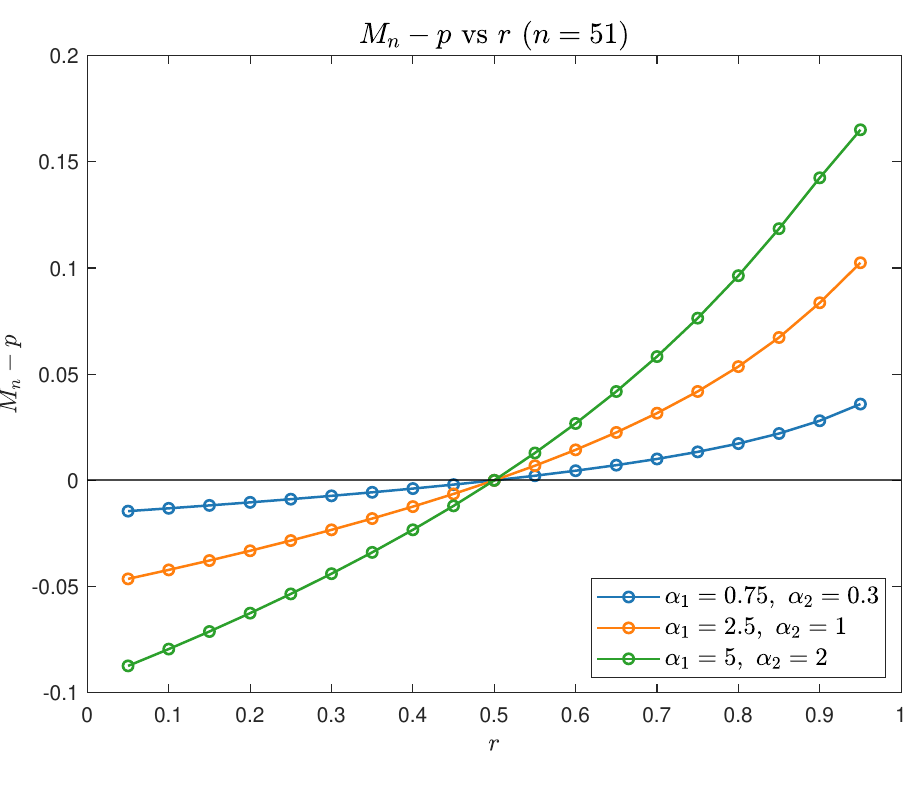}
\end{minipage}

\caption{Numerical illustration of the finite electorate Condorcet-type comparison on the complete graph.
Each panel plots $M_n-p$ against $r$, computed from the closed-form stationary law, for fixed odd
$n\in\{5,11,31,51\}$ and for the three parameter pairs
$(\alpha_1,\alpha_2)=(0.75,0.3)$, $(2.5,1)$, and $(5,2)$.
The parameter choices are intended only to sample low, intermediate, and stronger zealot influence regimes
motivated by the qualitative regime changes identified in earlier mixed-zealotry work \citep{chinellato2007, chinellato2015, braha_deaguiar2017}.
Because $\alpha_1>\alpha_2$, the sign of $p-\tfrac12$ coincides with the sign of $r-\tfrac12$,
so the right half of each panel corresponds to $p>\tfrac12$.
Across the plotted range, $M_n-p$ is positive when $p>\tfrac12$ and negative when $p<\tfrac12$.
The gap also increases with the influence weights. Unlike in the classical CJT, however, the majority advantage $M_n-p$ need not increase monotonically with electorate size.}
\label{fig:Mn_minus_p_vs_r_combined}
\end{figure}

The next theorem gives an affirmative answer for the complete graph baseline, that is, for fully mixed electorates. It is the
finite electorate analogue, within our zealot--contrarian model, of the nonasymptotic comparison in the classical CJT. The proof
works directly from the stationary distribution \eqref{eq:pi-binary-contrarian}, making explicit the parameter regime under which
the comparison holds. For completeness, Appendix~\ref{app:berg-majority-proof} presents a second proof obtained by reducing the
stationary count distribution to the P\'olya--Eggenberger family studied by \citet{berg1993condorcet}.

\begin{theorem}
\label{thm:majority-complete-finite}
Assume the binary fully mixed model of Theorem~\ref{thm:stationary-binary-contrarian}, and let $n=2s+1$ be odd.
Let
\[
p:=\frac{\E[X]}{n}
\]
denote the stationary correctness probability of a uniformly chosen free voter, and let
\[
M_n:=\Pp(X\ge s+1)
\]
denote the strict majority success probability for the correct alternative~$1$. If $p>\tfrac12$, then
\[
M_n>p.
\]
If $r=\tfrac12$, then $p=M_n=\tfrac12$.
\end{theorem}

\begin{proof}
Write $n=2s+1$ and let $\pi(k)=\Pp(X=k)$ denote the stationary distribution of $X$ on $\{0,1,\dots,n\}$.
By definition,
\[
p=\frac{\E[X]}{n}=\sum_{k=0}^n \frac{k}{n}\,\pi(k),
\qquad
M_n=\Pp(X\ge s+1)=\sum_{k=s+1}^n \pi(k).
\]

We first note a general identity:
\begin{align}
 n\,(M_n-p)
 &= n\sum_{k=s+1}^n \pi(k)-\sum_{k=0}^n k\,\pi(k) \notag\\
 &= \sum_{k=s+1}^n (n-k)\,\pi(k)-\sum_{k=0}^{s} k\,\pi(k). \label{eq:Mn_minus_p_identity_1_main}
\end{align}
Changing variables $k\mapsto n-k$ in the first sum yields
\begin{equation}
\label{eq:Mn_minus_p_identity_2_main}
 n\,(M_n-p)
 =\sum_{k=0}^{s} k\bigl(\pi(n-k)-\pi(k)\bigr)
 =\sum_{k=1}^{s} k\bigl(\pi(n-k)-\pi(k)\bigr).
\end{equation}
Therefore it is enough to prove that
\[
\pi(n-k)>\pi(k)\qquad\text{for each }k=1,\dots,s.
\]

If $r=\tfrac12$, then Theorem~\ref{thm:stationary-binary-contrarian}(i) gives
$X\sim\mathrm{Binomial}(n,\tfrac12)$, so $p=\tfrac12$ and, by symmetry, $M_n=\tfrac12$.
Thus only the case $r\neq\tfrac12$ requires proof.

Assume henceforth that $r\neq\tfrac12$. By Theorem~\ref{thm:stationary-binary-contrarian}(ii),
\[
\pi(k)=\binom{n}{k}\frac{(\theta)_k(\beta)_{n-k}}{(\theta+\beta)_n},
\qquad k=0,1,\dots,n,
\]
with
\[
p=\frac{\theta}{\theta+\beta}.
\]
Fix $1\le k\le s$. Using $\binom{n}{n-k}=\binom{n}{k}$,
\begin{align}
\frac{\pi(n-k)}{\pi(k)}
&=\frac{(\theta)_{n-k}(\beta)_k}{(\theta)_k(\beta)_{n-k}} \notag\\
&=\frac{(\theta+k)_{n-2k}}{(\beta+k)_{n-2k}}
=\prod_{j=0}^{n-2k-1}\frac{\theta+k+j}{\beta+k+j}. \label{eq:reflection_ratio_main}
\end{align}
Because $n=2s+1$ and $k\le s$, we have $n-2k\ge 1$, so the product is nonempty.

To show that every factor in \eqref{eq:reflection_ratio_main} is strictly larger than $1$, recall that $a=2r-1$. From the proof of
Theorem~\ref{thm:stationary-binary-contrarian}, for every integer $t\in\{0,1,\dots,n-1\}$,
\[
a(\theta+t)=r(t+\alpha_1)+(1-r)(n-t-1+\alpha_2)>0,
\]
\[
a(\beta+t)=r(n-t-1+\alpha_2)+(1-r)(t+\alpha_1)>0.
\]
Hence $\theta+t$ and $\beta+t$ both have the same sign as $a$ for all $t=0,1,\dots,n-1$.
In particular, $\theta$ and $\beta$ have the same sign as $a$.
Since
\[
p=\frac{\theta}{\theta+\beta}>\frac12,
\]
and $\theta,\beta$ have the same sign, there are only two possibilities:
\[
a>0\quad\text{and}\quad \theta>\beta>0,
\qquad\text{or}\qquad
 a<0\quad\text{and}\quad \theta<\beta<0.
\]
Now, for fixed $k$ and $j=0,1,\dots,n-2k-1$, set $t=k+j$. Then
\[
t\in\{k,k+1,\dots,n-k-1\}\subset\{1,\dots,n-2\}\subset\{0,\dots,n-1\}.
\]
Therefore the sign property established above applies to every factor in \eqref{eq:reflection_ratio_main}. In the first case,
\[
\theta+t>\beta+t>0,
\]
so $(\theta+t)/(\beta+t)>1$. In the second case,
\[
\theta+t<\beta+t<0,
\]
so again $(\theta+t)/(\beta+t)>1$. Therefore every factor in \eqref{eq:reflection_ratio_main} exceeds $1$, and hence
\[
\pi(n-k)>\pi(k)\qquad\text{for all }k=1,\dots,s.
\]

Substituting this inequality into \eqref{eq:Mn_minus_p_identity_2_main} yields $M_n-p>0$, hence $M_n>p$.
\end{proof}

Theorem~\ref{thm:majority-complete-finite} establishes analytically what the right half of
Fig.~\ref{fig:Mn_minus_p_vs_r_combined} shows numerically: when $p>\tfrac12$, one has $M_n-p>0$. It shows that once the stationary dynamics make a uniformly chosen free voter more likely than not to support the correct alternative, majority rule strictly improves on that endogenous post-deliberation benchmark. The left half of Fig.~\ref{fig:Mn_minus_p_vs_r_combined} also shows,
numerically, the reverse ordering when $p<\tfrac12$. The comparison with the ``no-deliberation'' benchmark,
in which free voters do not influence one another and are connected only to zealots, is taken up next.

\subsection{Comparison with a no-deliberation benchmark}
\label{sec:majority-no-deliberation}

Theorem~\ref{thm:majority-complete-finite} is an \emph{internal} Condorcet-type comparison: it asks whether
majority rule improves on the correctness probability of a single free voter \emph{after} the social influence
process has already taken place. From the standpoint of the CJT, however, independence is itself a substantive
assumption, so it is also natural to compare the fully mixed electorate with a counterfactual benchmark in which
free voters never influence one another. This ``no-deliberation'' benchmark isolates the epistemic effect of
endogenous social influence. If deliberation were always beneficial for information aggregation, one would expect
the fully mixed electorate to be at least as accurate as this benchmark. The next theorem shows that the opposite
can occur.

In the no-deliberation benchmark, all free--free edges are removed, while each free voter remains connected to
all zealots. Let $X^{\mathrm{ND}}$ denote the stationary number of free voters supporting the correct
alternative~$1$. Since a free voter then samples only zealots, its stationary success probability is
\begin{equation}
\label{eq:p-ND}
p^{\mathrm{ND}}
=
\frac{r\alpha_1+(1-r)\alpha_2}{\alpha_1+\alpha_2}
=
\frac12+\frac{(2r-1)(\alpha_1-\alpha_2)}{2(\alpha_1+\alpha_2)},
\end{equation}
and the free voters are independent in stationarity. Hence
\begin{equation}
\label{eq:XND-law}
X^{\mathrm{ND}}\sim \mathrm{Binomial}\!\left(n,p^{\mathrm{ND}}\right).
\end{equation}
For odd $n=2s+1$, define the corresponding strict majority success probability by
\begin{equation}
\label{eq:MnND-def}
M_n^{\mathrm{ND}}
:=
\Pp\!\left(X^{\mathrm{ND}}\ge s+1\right).
\end{equation}

Before turning to the theorem, Fig.~\ref{fig:Mn_minus_MnND_vs_r_combined} provides a numerical illustration of this comparison with the counterfactual benchmark in which free voters never influence one another. For odd $n$, it plots $M_n-M_n^{\mathrm{ND}}$ as a function of $r$ for $n\in\{5,11,31,51\}$ and for the same three representative pairs $(\alpha_1,\alpha_2)$ used in Fig.~\ref{fig:Mn_minus_p_vs_r_combined}. Because $\alpha_1>\alpha_2$ in all three cases, we again have $p>\tfrac12$ if and only if $r>\tfrac12$. The figure therefore makes the comparison visually transparent: throughout the plotted range, $M_n-M_n^{\mathrm{ND}}<0$ whenever $p>\tfrac12$, while the reverse ordering $M_n-M_n^{\mathrm{ND}}>0$ is observed when $p<\tfrac12$. It is also noteworthy that, on the right half of the plots, the negative gap becomes more pronounced as $\alpha_1$ and $\alpha_2$ decrease, that is, when zealot strengths are weaker and social influence among free voters is relatively more dominant. Finally, over the values shown, the difference $M_n-M_n^{\mathrm{ND}}$ is neither increasing nor decreasing monotonically as the electorate size $n$ grows.

\begin{figure}[!htbp]
\centering
\setlength{\abovecaptionskip}{4pt}
\setlength{\belowcaptionskip}{0pt}

\begin{minipage}[t]{0.49\textwidth}
\centering
\includegraphics[width=\linewidth,trim=78 205 95 210,clip]{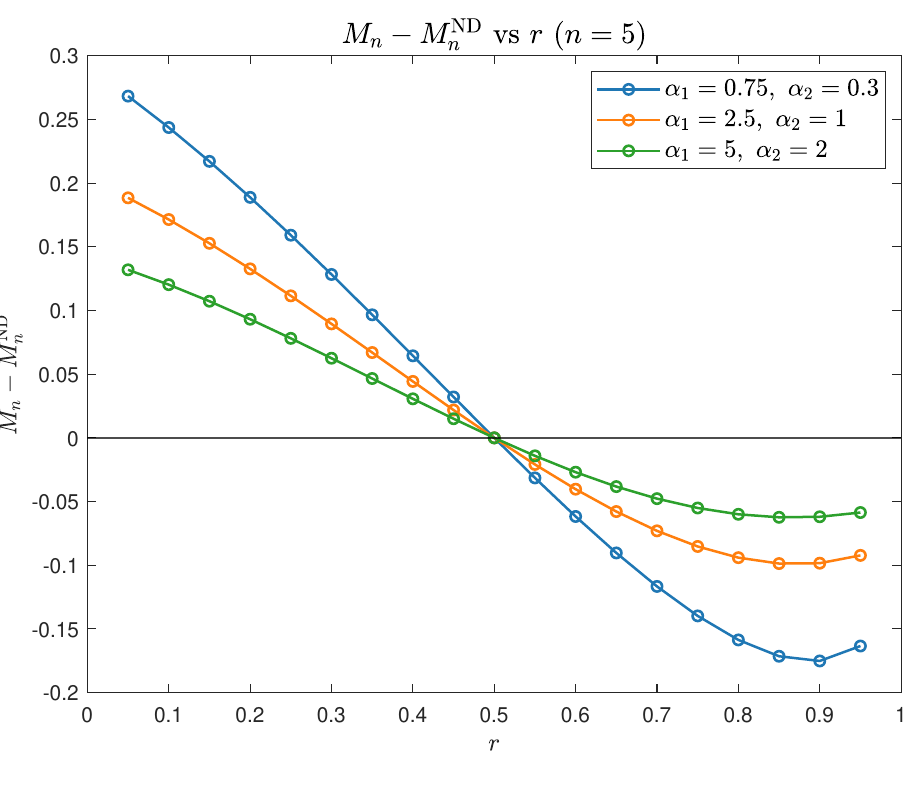}
\end{minipage}
\hfill
\begin{minipage}[t]{0.49\textwidth}
\centering
\includegraphics[width=\linewidth,trim=78 205 95 210,clip]{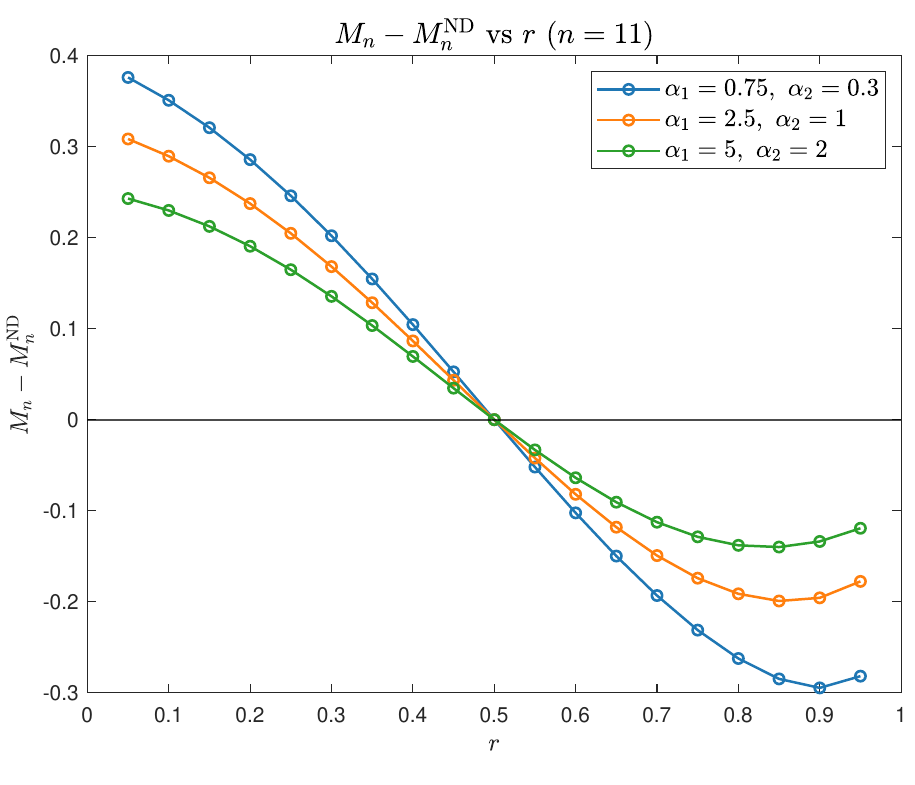}
\end{minipage}

\vspace{-0.35em}

\begin{minipage}[t]{0.49\textwidth}
\centering
\includegraphics[width=\linewidth,trim=78 205 95 210,clip]{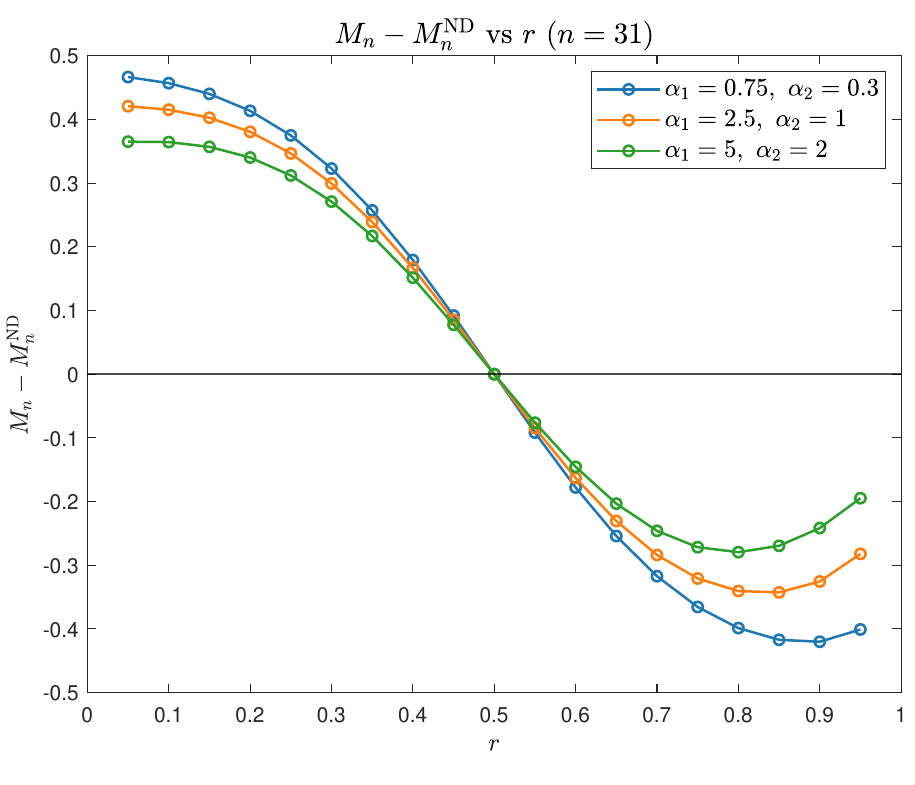}
\end{minipage}
\hfill
\begin{minipage}[t]{0.49\textwidth}
\centering
\includegraphics[width=\linewidth,trim=78 205 95 210,clip]{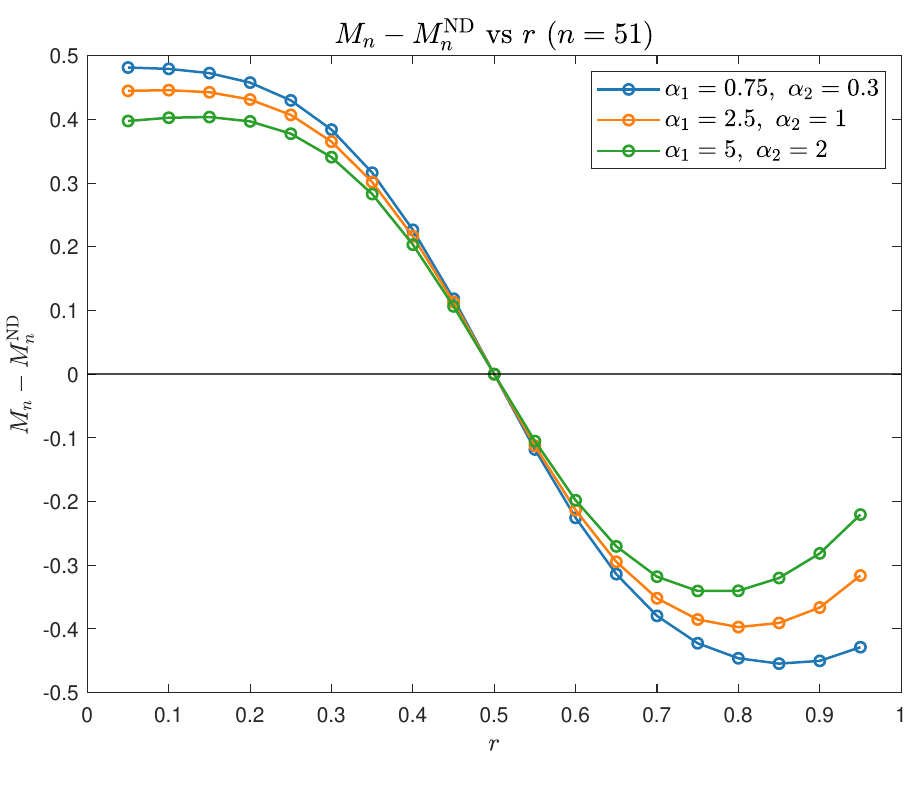}
\end{minipage}

\caption{Numerical illustration of the comparison with the no-deliberation benchmark on the complete graph.
Each panel plots $M_n-M_n^{\mathrm{ND}}$ against $r$, for fixed odd $n\in\{5,11,31,51\}$, using the same three parameter pairs as in Fig.~\ref{fig:Mn_minus_p_vs_r_combined}; see the discussion there for the rationale behind these choices.
Because $\alpha_1>\alpha_2$ in all three cases, the sign of $p-\tfrac12$ again coincides with the sign of $r-\tfrac12$, so the right half of each panel corresponds to $p>\tfrac12$.
Across all panels, $M_n-M_n^{\mathrm{ND}}$ is negative when $p>\tfrac12$, while the reverse ordering is observed when $p<\tfrac12$.
For $p>\tfrac12$, the negative gap becomes more pronounced as $\alpha_1$ and $\alpha_2$ decrease, that is, when zealot influence is weaker and free--free social influence is relatively more dominant.
Over the plotted values, the difference $M_n-M_n^{\mathrm{ND}}$ is neither increasing nor decreasing monotonically as the electorate size $n$ grows.}
\label{fig:Mn_minus_MnND_vs_r_combined}
\end{figure}

\begin{theorem}
\label{thm:majority-no-deliberation}
Assume the binary fully mixed model of Theorem~\ref{thm:stationary-binary-contrarian}, and let $n=2s+1$ be odd.
Let $M_n$ denote the strict majority success probability on the complete graph and let
$M_n^{\mathrm{ND}}$ be the no-deliberation benchmark defined in \eqref{eq:MnND-def}. If
$r>\tfrac12$ and $\alpha_1>\alpha_2>0$, then
\[
M_n<M_n^{\mathrm{ND}}.
\]
\end{theorem}

\begin{proof}
Define
\begin{equation}
\label{eq:gn-def-main}
g_n(u):=\Pp\!\left(\mathrm{Binomial}(n,u)\ge s+1\right),
\qquad u\in[0,1].
\end{equation}
Then
\[
M_n^{\mathrm{ND}}=g_n\!\left(p^{\mathrm{ND}}\right).
\]
A direct differentiation gives
\begin{equation}
\label{eq:gn-derivatives-main}
g_n'(u)=n\binom{n-1}{s}u^{s}(1-u)^{s}>0,
\qquad
g_n''(u)=ns\binom{n-1}{s}u^{s-1}(1-u)^{s-1}(1-2u).
\end{equation}
Hence $g_n$ is strictly increasing on $(0,1)$ and strictly concave on $(\tfrac12,1)$.

By \eqref{eq:p-ND}, the assumptions $r>\tfrac12$ and $\alpha_1>\alpha_2$ imply
\[
p^{\mathrm{ND}}>\frac12.
\]
Moreover, combining \eqref{eq:p-mean} and \eqref{eq:p-ND}, the stationary marginal correctness on the complete graph can be written as
\begin{equation}
\label{eq:p-pND-convex-main}
p
=
\lambda\,p^{\mathrm{ND}}+(1-\lambda)\,\frac12,
\qquad
\lambda=
\frac{\alpha_1+\alpha_2}{\alpha_1+\alpha_2+2(1-r)(n-1)}\in(0,1).
\end{equation}
Therefore,
\begin{equation}
\label{eq:p-between-main}
\frac12<p<p^{\mathrm{ND}}.
\end{equation}

Next, since $r>\tfrac12$, the parameters in \eqref{eq:def-theta-beta_maintext} satisfy
\[
\theta=
\frac{r\alpha_1+(1-r)(n-1+\alpha_2)}{2r-1}>0,
\qquad
\beta=
\frac{(1-r)(n-1+\alpha_1)+r\alpha_2}{2r-1}>0.
\]
Thus \eqref{eq:pi-binary-contrarian} is an ordinary beta--binomial law, so it admits the standard beta-mixture representation: there exists
\[
\Theta\sim\mathrm{Beta}(\theta,\beta)
\]
such that
\[
X\mid \Theta \sim \mathrm{Binomial}(n,\Theta),
\qquad
\E[\Theta]=\frac{\theta}{\theta+\beta}=p.
\]
Consequently,
\begin{equation}
\label{eq:Mn-beta-mixture-main}
M_n=\E[g_n(\Theta)].
\end{equation}

We now show that
\begin{equation}
\label{eq:beta-mixture-upper-main}
\E[g_n(\Theta)]<g_n(\E[\Theta])=g_n(p).
\end{equation}
Because $g_n$ is concave only on $(\tfrac12,1)$, Jensen's inequality cannot be applied directly to $\Theta$. To isolate the part of the beta law that lies to the right of $\tfrac12$, define
\[
h_n(u):=g_n(u)-\frac12.
\]
Since $n$ is odd, binomial symmetry yields
\[
h_n(1-u)=-h_n(u).
\]

Let $f$ denote the density of $\Theta$. Splitting the integral for \(\E[h_n(\Theta)]\) at \(\tfrac12\) and making the change of variables \(v=1-u\) on \((0,\tfrac12)\) yields
\[
\E[h_n(\Theta)]
=
\int_{1/2}^{1} h_n(u)\bigl(f(u)-f(1-u)\bigr)\,du.
\]
Since $\E[\Theta]=p>\tfrac12$, we have $\theta>\beta$. Therefore, for every $u\in(\tfrac12,1)$,
\[
\frac{f(u)}{f(1-u)}
=
\left(\frac{u}{1-u}\right)^{\theta-\beta}>1,
\]
so
\[
w(u):=f(u)-f(1-u)>0,
\qquad \tfrac12<u<1.
\]
Set
\[
C:=\int_{1/2}^{1} w(u)\,du \in (0,1),
\qquad
q(u):=\frac{w(u)}{C},
\qquad \tfrac12<u<1.
\]
Then $q$ is a probability density on $(\tfrac12,1)$. Let $U$ be a random variable with density $q$. By construction,
\begin{equation}
\label{eq:EhU-main}
\E[h_n(\Theta)]=C\,\E[h_n(U)].
\end{equation}

We next compute the mean of $U$. Using the definition of $w$ and the change of variables $u\mapsto 1-u$ on $(0,\tfrac12)$,
\[
\int_{1/2}^{1}\Bigl(u-\frac12\Bigr)w(u)\,du
=
\int_0^1\Bigl(u-\frac12\Bigr)f(u)\,du
=
p-\frac12.
\]
Hence
\[
\E[U]-\frac12
=
\frac{1}{C}\int_{1/2}^{1}\Bigl(u-\frac12\Bigr)w(u)\,du
=
\frac{p-\frac12}{C},
\]
that is,
\[
\E[U]
=
\frac12+\frac{p-\frac12}{C}
\ge p,
\]
because $C\le 1$.

Now $h_n$ is strictly concave on $(\tfrac12,1)$ by \eqref{eq:gn-derivatives-main}, and $U$ is nondegenerate on that interval. Therefore Jensen's inequality is strict:
\[
\E[h_n(U)]< h_n(\E[U]).
\]
Also, since $h_n$ is concave on $(\tfrac12,1)$ and $h_n(\tfrac12)=0$, the secant slope \(\frac{h_n(x)}{x-\tfrac12}\) is nonincreasing for \(x\in(\tfrac12,1)\). Because $\E[U]\ge p>\tfrac12$, it follows that
\[
\frac{h_n(\E[U])}{\E[U]-\tfrac12}
\le
\frac{h_n(p)}{p-\tfrac12},
\]
and hence
\[
h_n(\E[U])
\le
\frac{\E[U]-\tfrac12}{p-\tfrac12}\,h_n(p)
=
\frac{h_n(p)}{C}.
\]
Combining this with \eqref{eq:EhU-main}, we obtain
\[
\E[h_n(\Theta)]
=
C\,\E[h_n(U)]
<
C\,h_n(\E[U])
\le
h_n(p).
\]
Therefore,
\[
M_n-\frac12
=
\E[h_n(\Theta)]
<
h_n(p)
=
g_n(p)-\frac12,
\]
which proves \eqref{eq:beta-mixture-upper-main}.

Finally, \eqref{eq:p-between-main} and the strict monotonicity of $g_n$ yield
\[
M_n<g_n(p)<g_n\!\left(p^{\mathrm{ND}}\right)=M_n^{\mathrm{ND}}.
\]
\end{proof}

Theorem~\ref{thm:majority-no-deliberation} sharpens the message of Section~\ref{sec:majority-complete-finite}. Although majority rule still improves on the \emph{post-deliberation} benchmark $p$ once $p>\tfrac12$, the deliberation process itself can reduce collective accuracy relative to an otherwise identical electorate with no endogenous interaction among free voters. More precisely, the theorem establishes analytically what the right half of Fig.~\ref{fig:Mn_minus_MnND_vs_r_combined} suggests numerically: when $p>\tfrac12$, one has $M_n-M_n^{\mathrm{ND}}<0$. The left half of Fig.~\ref{fig:Mn_minus_MnND_vs_r_combined} likewise suggests, numerically, the reverse ordering when $p<\tfrac12$. In that sense, the result identifies a failure of epistemic social aggregation caused by social influence: free--free interaction creates dependence and pulls the stationary marginal competence toward $1/2$, and the resulting loss can outweigh the within-model majority advantage established in Theorem~\ref{thm:majority-complete-finite}.

\subsection{Large electorate asymptotics on fully mixed populations}
\label{sec:majority-large-n}

Sections~\ref{sec:majority-complete-finite} and \ref{sec:majority-no-deliberation} address the \emph{nonasymptotic} side of the CJT: for a fixed odd electorate, does majority rule improve on a single free voter, and how does it compare with an independent no-deliberation benchmark? We now turn to the asymptotic side and ask what happens to strict majority correctness as the number of free voters grows. In this subsection, to make the dependence on electorate size explicit, let $X_n$ denote the stationary number of free voters supporting the correct alternative~$1$ in an electorate of size $n$, let
\[
p_n:=\frac{\E[X_n]}{n},
\]
and, for odd $n$, let
\[
M_n:=\Pp\!\left(X_n\ge \frac{n+1}{2}\right).
\]
As shown below, the large electorate answer depends sharply on whether the society is fully conformist ($r=1$) or contains even a small contrarian component ($\tfrac12<r<1$). We first consider the fully conformist case. The next theorem identifies the exact large-$n$ limit and shows that it is \emph{strictly larger} than the stationary correctness probability of a single free voter.

\begin{theorem}
\label{thm:majority-large-n-r1}
Assume the binary fully mixed model with $r=1$ and $\alpha_1>\alpha_2>0$. Then
\[
p_n=\frac{\alpha_1}{\alpha_1+\alpha_2}
\qquad\text{for every }n,
\]
and, along odd $n$,
\[
M_n\longrightarrow \Pp\!\left(\Theta>\frac12\right),
\qquad
\Theta\sim \mathrm{Beta}(\alpha_1,\alpha_2).
\]
Moreover,
\[
\Pp\!\left(\Theta>\frac12\right)
>
\frac{\alpha_1}{\alpha_1+\alpha_2}
=
p_n
>
\frac12.
\]
\end{theorem}

\begin{proof}
When $r=1$, Theorem~\ref{thm:stationary-binary-contrarian} gives $\theta=\alpha_1$ and $\beta=\alpha_2$, so the stationary distribution is the ordinary beta--binomial law
\[
X_n\mid \Theta \sim \mathrm{Binomial}(n,\Theta),
\qquad
\Theta\sim \mathrm{Beta}(\alpha_1,\alpha_2).
\]
In particular,
\[
p_n=\frac{\E[X_n]}{n}=\E[\Theta]=\frac{\alpha_1}{\alpha_1+\alpha_2}.
\]

Fix $\Theta=u$. By the law of large numbers, $\mathrm{Binomial}(n,u)/n\to u$ in probability, so
\[
\Pp\!\left(X_n\ge \frac{n+1}{2}\,\middle|\,\Theta=u\right)\longrightarrow \mathbf{1}_{\{u>1/2\}}
\qquad (u\neq\tfrac12).
\]
Since the pointwise limit above is established only for $u\neq\tfrac12$, we need it to hold for $\Theta$-almost every value of $u$. Because $\Theta$ has a continuous density, $\Pp(\Theta=\tfrac12)=0$, so the exceptional point carries no mass. Dominated convergence therefore yields
\[
M_n=\E\!\left[\Pp\!\left(X_n\ge \frac{n+1}{2}\,\middle|\,\Theta\right)\right]
\longrightarrow
\Pp\!\left(\Theta>\frac12\right).
\]

It remains to compare this limit with $p_n$. Let
\[
f(u)=\frac{u^{\alpha_1-1}(1-u)^{\alpha_2-1}}{B(\alpha_1,\alpha_2)},
\qquad 0<u<1,
\]
denote the beta density. Then
\begin{align*}
\Pp\!\left(\Theta>\frac12\right)-\E[\Theta]
&=
\int_{1/2}^{1} f(u)\,du-\int_{0}^{1}u\,f(u)\,du \\
&=
\int_{1/2}^{1} (1-u)\bigl(f(u)-f(1-u)\bigr)\,du.
\end{align*}
For $u\in(\tfrac12,1)$,
\[
\frac{f(u)}{f(1-u)}
=
\left(\frac{u}{1-u}\right)^{\alpha_1-\alpha_2}
>1
\]
because $\alpha_1>\alpha_2$. Hence $f(u)-f(1-u)>0$ on $(\tfrac12,1)$, so the last integral is strictly positive. Therefore
\[
\Pp\!\left(\Theta>\frac12\right)>\E[\Theta]=\frac{\alpha_1}{\alpha_1+\alpha_2},
\]
which is the desired inequality.
\end{proof}

Thus, in a fully conformist society, weakly informative signals are amplified rather than washed out: the large electorate majority remains more accurate than any single free voter, although unlike the classical CJT it does not converge to certainty. The picture changes completely once $r$ drops below~$1$, as the next result shows.

\begin{theorem}
\label{thm:majority-large-n-mixed}
Fix $\alpha_1,\alpha_2>0$ and $r\in(\tfrac12,1)$. Then
\[
\frac{X_n-n/2}{\sqrt n}
\xrightarrow{d}
\mathcal{N}\!\left(0,\frac{1}{8(1-r)}\right).
\]
Consequently, along odd $n$,
\[
M_n\longrightarrow \frac12.
\]
If in addition $\alpha_1>\alpha_2$, then $p_n>\tfrac12$ for every $n$, but
\[
p_n\longrightarrow \frac12
\qquad\text{and}\qquad
M_n\longrightarrow \frac12.
\]
\end{theorem}

\begin{proof}[Proof sketch]
The full proof is given in Appendix~\ref{app:large-n-majority-proof}; here we outline only the main steps. For $r\in(\tfrac12,1)$, Theorem~\ref{thm:stationary-binary-contrarian} gives a beta--binomial stationary law with parameters
\[
\theta_n=\frac{r\alpha_1+(1-r)(n-1+\alpha_2)}{2r-1},
\qquad
\beta_n=\frac{(1-r)\alpha_1+r(n-1+\alpha_2)}{2r-1}-n+1.
\]
These satisfy
\[
\theta_n+\beta_n=\frac{\alpha_1+\alpha_2+2(1-r)(n-1)}{2r-1},
\qquad
\theta_n-\beta_n=\alpha_1-\alpha_2.
\]
Hence both $\theta_n$ and $\beta_n$ grow linearly with $n$, while their difference stays fixed. Equivalently, if
\[
P_n\sim \mathrm{Beta}(\theta_n,\beta_n),
\qquad
X_n\mid P_n\sim \mathrm{Binomial}(n,P_n),
\]
then $P_n$ concentrates around $1/2$ with fluctuations of order $n^{-1/2}$. Writing
\[
\frac{X_n-n/2}{\sqrt n}
=
\sqrt n\!\left(P_n-\frac12\right)
+
\frac{X_n-nP_n}{\sqrt n},
\]
the first term converges to a centered normal law with variance $(2r-1)/[8(1-r)]$, while the second converges to an asymptotically independent centered normal law with variance $1/4$. Their sum therefore converges to
\[
\mathcal{N}\!\left(0,\frac{1}{8(1-r)}\right),
\]
which yields $M_n\to\tfrac12$ because the limiting distribution is continuous and centered.

Finally, when $\alpha_1>\alpha_2$, Equation \eqref{eq:p-mean} gives
\[
p_n-\frac12
=
\frac{(2r-1)(\alpha_1-\alpha_2)}
{2\bigl[\alpha_1+\alpha_2+2(1-r)(n-1)\bigr]}
>0,
\]
so $p_n>\tfrac12$ for every $n$, but also $p_n\to\tfrac12$ as $n\to\infty$.
\end{proof}

Theorem~\ref{thm:majority-large-n-mixed} is striking from a CJT perspective. For every finite odd electorate with $\alpha_1>\alpha_2$ and $r>\tfrac12$, Theorem~\ref{thm:majority-complete-finite} shows that majority rule still improves on the endogenous post-deliberation benchmark $p_n$. Yet if $r<1$ is held fixed and $n$ grows, this advantage disappears: both $p_n$ and $M_n$ collapse to $1/2$. In particular, the classical no-deliberation benchmark behaves very differently. When $\alpha_1>\alpha_2$ and $r>\tfrac12$, \eqref{eq:p-ND} gives $p^{\mathrm{ND}}>\tfrac12$, so by the classical CJT the corresponding independent benchmark satisfies $M_n^{\mathrm{ND}}\to1$. Thus endogenous social influence can destroy asymptotic wisdom completely: instead of converging to certainty, strict majority correctness converges to random guessing. Together, Theorems~\ref{thm:majority-large-n-r1} and \ref{thm:majority-large-n-mixed} identify a tipping point at $r=1$: perfect conformism retains a nontrivial asymptotic accuracy advantage, whereas any fixed amount of contrarian updating collapses that advantage in the large electorate limit.

\section{Simulations}
\label{simulations}

In this section we use simulations to study how strict majority correctness depends on network topology beyond the complete graph and Erd\H{o}s--R\'enyi networks treated analytically in the previous sections. Our goals are threefold. First, we compare strict majority correctness across several sparse and structured topologies, namely Erd\H{o}s--R\'enyi (ER), scale free (SF), ring, and small world (SW) networks. Second, on these networks, we examine the robustness of the two Condorcet-type results established in Sections~\ref{sec:majority-complete-finite} and \ref{sec:majority-no-deliberation} by comparing strict majority correctness both with post-deliberation marginal correctness and with the no-deliberation benchmark. Third, we assess the accuracy of the ER mean-field approximation developed in Section~\ref{voter-er}, both for strict majority correctness and for pairwise voter correlation.

We organize the simulation results as follows. Figs.~3--5 compare the simulated strict majority correctness $M_n^G$ across the four network families, where $G$ denotes the underlying topology. Fig.~6 then examines $M_n^G$ more broadly within the SW family as the rewiring probability varies. Finally, Figs.~7 and~8 return to the two Condorcet-type comparisons from Sections~\ref{sec:majority-complete-finite} and \ref{sec:majority-no-deliberation} and study how they behave on ER, SF, ring, and SW networks.

Throughout this section, the simulations use $n=501$ free voters and retention probability $\ell=0$. For the ER networks, free--free edges are generated independently with probability $\frac{d}{n-1}$, where the expected degree of free voters is $d=8$. For the SF case we use Barab\'asi--Albert networks \citep{albert2002statistical} with seed size $m_0=9$ and attachment parameter $m=4$, yielding mean degree approximately $2m=8$. For the ring case we use a deterministic ring lattice with average free voter degree $8$, in which each free voter is connected to neighbors at offsets $\pm1,\ldots,\pm4$ with periodic boundary conditions. For the SW case we construct the networks from the same ring lattice with average degree $8$ by rewiring each clockwise edge independently to a uniformly chosen non neighbor, while keeping the graph undirected and simple \citep{watts1998collective}. Thus all four topologies are compared at essentially the same average connectivity.

For ER, SF, and SW, we simulate $1000$ fixed graph realizations for each value of the conformity rate $r$. For the ring network, which is deterministic, we instead average over $1000$ simulation runs on the same graph. In all cases, the process is run for a burn in period of 1000 Monte Carlo sweeps (each consisting of $n$ update attempts), followed by 1000 stationary samples separated by 10 sweeps, and the reported quantities are averaged over realizations.

To compare sparse networks with the complete graph theory, we use the rescaling described in Section~\ref{voter-er}. Specifically, if $(\alpha_1,\alpha_2)$ denotes the complete graph parameter pair, then the corresponding network simulations use
\[
\alpha_i^{\mathrm{G}}=\alpha_i\,\frac{d}{n-1},
\qquad i=1,2,
\]
with $d=8$. This is equivalent to comparing the network simulations with the complete graph formulas under the effective parameter mapping \eqref{eq:alpha-eff}. For the pairwise correlation plots shown in Figs.~3--5, the simulated correlation is estimated from the stationary samples of $X_t$ using the identities from Section~\ref{voter-moments}. Let
\[
\widehat p^{\,G}=\frac{\overline X}{n},
\qquad
\widehat q^{\,G}=\frac{\overline{X(X-1)}}{n(n-1)},
\]
where the bars denote averages over stationary samples and $G$ labels the underlying network topology. Then the simulated correlation for a randomly chosen pair of free voters is estimated by
\[
\widehat{\mathrm{Corr}}^{\,G}
=
\frac{\widehat q^{\,G}-\bigl(\widehat p^{\,G}\bigr)^2}{\widehat p^{\,G}\bigl(1-\widehat p^{\,G}\bigr)}.
\]
In Figs.~3--5, we compare these estimates for all four network topologies with the ER mean-field counterpart \eqref{corr_er}.

We begin with Figs.~3--5, which compare the simulated strict majority correctness across the four network families. In Figs.~\ref{fig:ER_SF_Mn_corr_vs_r}--\ref{fig:ER_SW_Mn_corr_vs_r}, the three rows correspond throughout to the parameter pairs $(\alpha_1,\alpha_2)=(0.75,0.3)$, $(2.5,1)$, and $(5,2)$. We first compare ER and SF networks in Fig.~\ref{fig:ER_SF_Mn_corr_vs_r}. In each row, the middle panel plots the simulated strict majority correctness for the ER and SF networks together with the complete graph benchmark, while the right panel plots the corresponding pairwise voter--voter correlation. Fig.~\ref{fig:ER_SF_Mn_corr_vs_r} shows three robust qualitative patterns. First, for all three parameter pairs, the ER curves lie very close to the complete graph benchmark, providing numerical support for the ER mean-field reduction developed in Section~\ref{voter-er}. Second, across the full range of $r$ shown in the figure, the SF topology yields lower strict majority correctness than the ER topology. Third, the correlation panels explain this gap: the SF networks systematically generate stronger positive voter--voter correlations than the ER networks. Collectively, these plots support the interpretation that topology affects collective accuracy primarily through the dependence structure it induces among free voters. Relative to ER, the heterogeneous hub structure of the SF network creates a stronger common source effect and amplifies alignment among free voters' opinions. In the present setting, where alternative~1 is favored by the committed environment and majority accuracy is evaluated over the free voters, this stronger positive dependence reduces the effective amount of independent information aggregated by majority rule and lowers strict majority correctness.

\begin{figure}[!htbp]
\centering
\captionsetup{width=0.97\textwidth}

\noindent
\begin{minipage}[c]{0.10\textwidth}
\centering
\scriptsize
\makebox[\linewidth][c]{%
\begin{tabular}{@{}r@{}l@{}}
$\alpha_1$&$=0.75$\\[-0.05em]
$\alpha_2$&$=0.3$
\end{tabular}}
\end{minipage}\hspace{0.01\textwidth}%
\begin{minipage}[c]{0.43\textwidth}
\centering
\includegraphics[width=\linewidth]{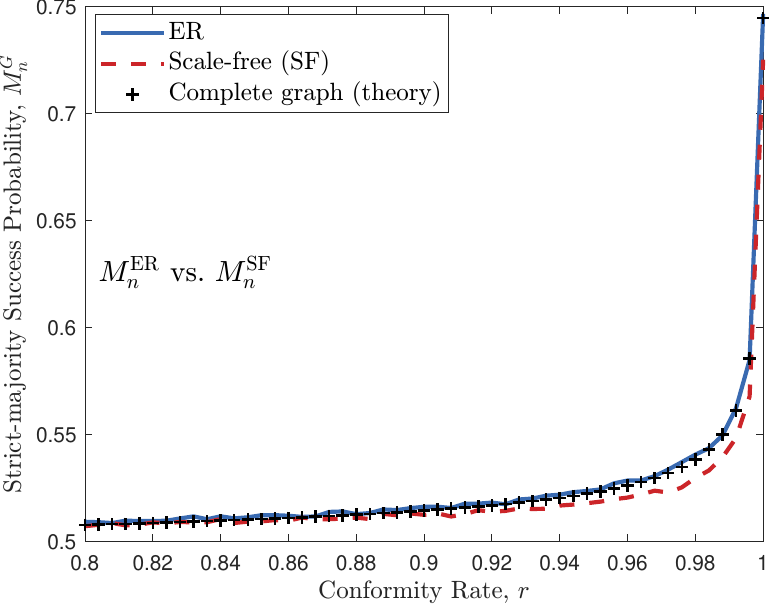}
\end{minipage}\hspace{0.01\textwidth}%
\begin{minipage}[c]{0.43\textwidth}
\centering
\includegraphics[width=\linewidth]{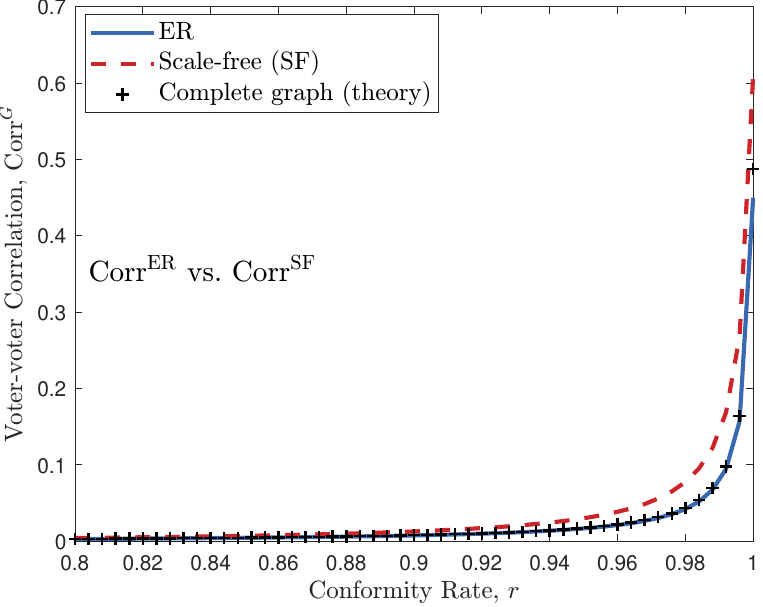}
\end{minipage}

\vspace{0.35em}

\noindent
\begin{minipage}[c]{0.10\textwidth}
\centering
\scriptsize
\makebox[\linewidth][c]{%
\begin{tabular}{@{}r@{}l@{}}
$\alpha_1$&$=2.5$\\[-0.05em]
$\alpha_2$&$=1$
\end{tabular}}
\end{minipage}\hspace{0.01\textwidth}%
\begin{minipage}[c]{0.43\textwidth}
\centering
\includegraphics[width=\linewidth]{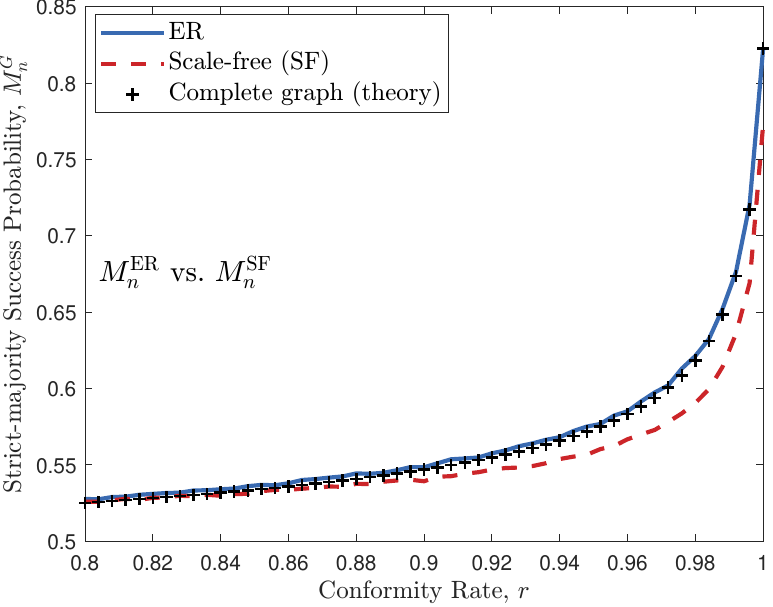}
\end{minipage}\hspace{0.01\textwidth}%
\begin{minipage}[c]{0.43\textwidth}
\centering
\includegraphics[width=\linewidth]{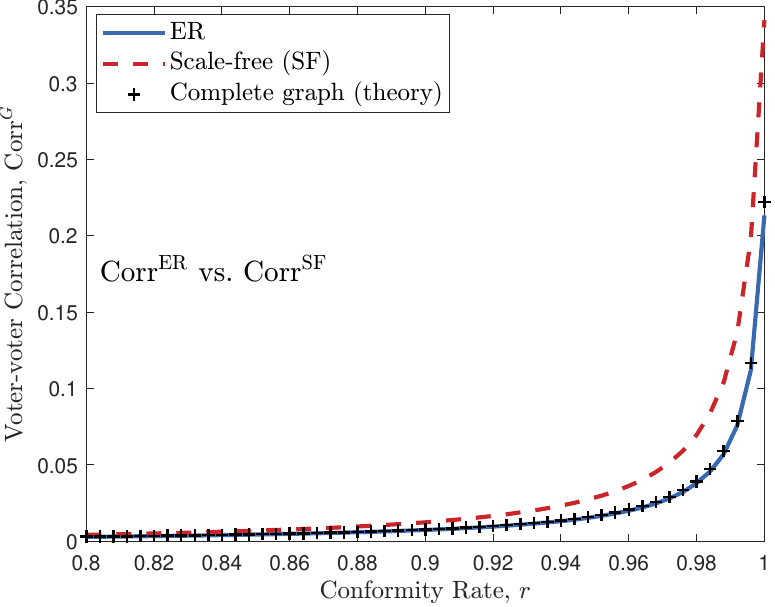}
\end{minipage}

\vspace{0.35em}

\noindent
\begin{minipage}[c]{0.10\textwidth}
\centering
\scriptsize
\makebox[\linewidth][c]{%
\begin{tabular}{@{}r@{}l@{}}
$\alpha_1$&$=5$\\[-0.05em]
$\alpha_2$&$=2$
\end{tabular}}
\end{minipage}\hspace{0.01\textwidth}%
\begin{minipage}[c]{0.43\textwidth}
\centering
\includegraphics[width=\linewidth]{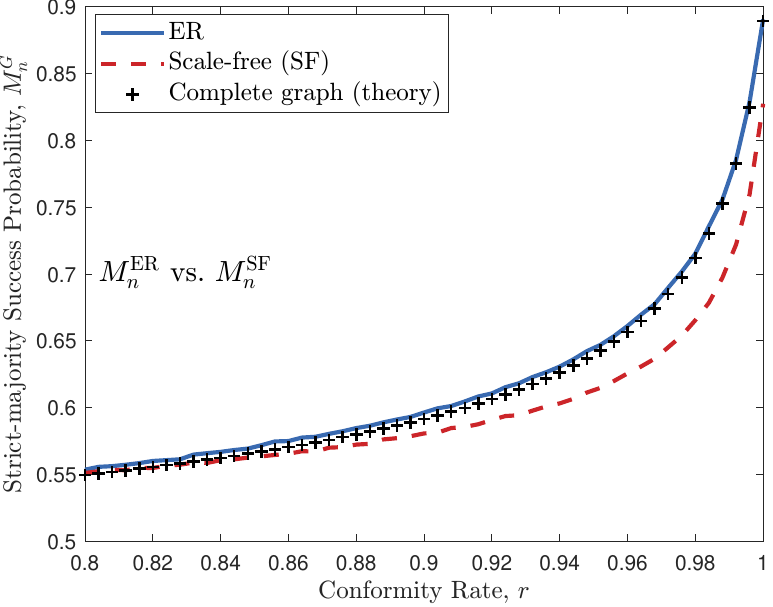}
\end{minipage}\hspace{0.01\textwidth}%
\begin{minipage}[c]{0.43\textwidth}
\centering
\includegraphics[width=\linewidth]{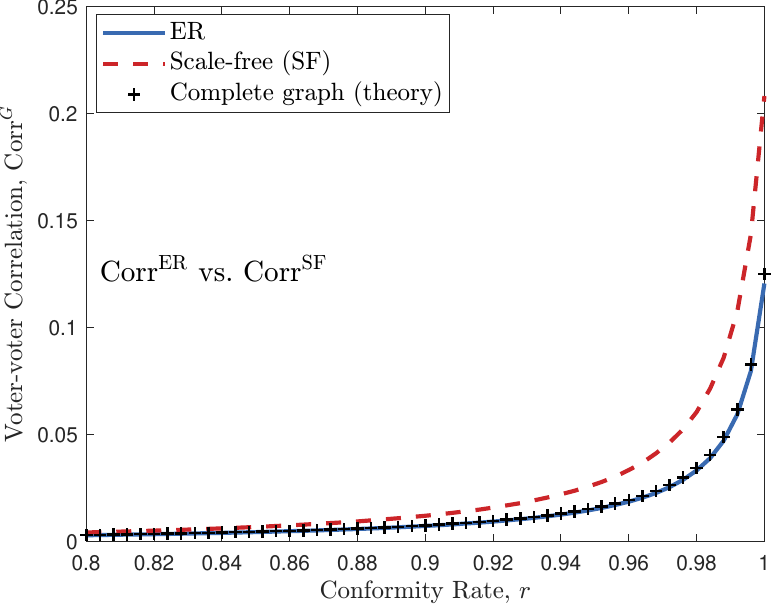}
\end{minipage}

\caption{Comparison of Erd\H{o}s--R\'enyi (ER) and scale free (SF) networks for the zealot--contrarian voter model, with $n=501$ free voters and conformity rate $r\in[0.8,1]$, where topology-dependent differences are most pronounced. The three rows correspond to the parameter pairs indicated in the left column. In each row, the middle panel shows the simulated strict majority correctness $M_n^G$ for ER and SF networks together with the complete graph benchmark, and the right panel shows the corresponding pairwise voter--voter correlation. Zealot strengths in the network simulations are rescaled according to the ER mean-field mapping from Section~\ref{voter-er}. Across all three parameter pairs, the ER results closely track the complete graph theory, whereas the SF topology yields lower strict majority correctness and higher positive voter--voter correlation.}
\label{fig:ER_SF_Mn_corr_vs_r}
\end{figure}

We next compare ER and ring networks, shown in Fig.~\ref{fig:ER_RING_Mn_corr_vs_r}. The row structure is the same as in Fig.~\ref{fig:ER_SF_Mn_corr_vs_r}, with strict majority correctness in the middle panels and the corresponding pairwise voter--voter correlation in the right panels. Fig.~\ref{fig:ER_RING_Mn_corr_vs_r} shows the reverse ordering from Fig.~\ref{fig:ER_SF_Mn_corr_vs_r}. For all three parameter pairs, the ring topology yields higher strict majority correctness than ER throughout the plotted range. At the same time, the correlation panels show that ring networks systematically generate lower pairwise voter--voter correlations than ER. The complete graph benchmark again remains closely aligned with the ER simulation curves, consistent with the ER mean-field reduction developed in Section~\ref{voter-er}. The ring topology departs from this ER-like behavior in the direction of weaker dependence and better majority performance. Relative to ER, the local structure of the ring network suppresses long range coordination and weakens the common source effects induced by social influence. In the present setting, this weaker dependence leaves more effective information for majority aggregation and therefore raises strict majority correctness.

\begin{figure}[!htbp]
\centering
\captionsetup{width=0.97\textwidth}

\noindent
\begin{minipage}[c]{0.10\textwidth}
\centering
\scriptsize
\makebox[\linewidth][c]{%
\begin{tabular}{@{}r@{}l@{}}
$\alpha_1$&$=0.75$\\[-0.05em]
$\alpha_2$&$=0.3$
\end{tabular}}
\end{minipage}\hspace{0.01\textwidth}%
\begin{minipage}[c]{0.43\textwidth}
\centering
\includegraphics[width=\linewidth]{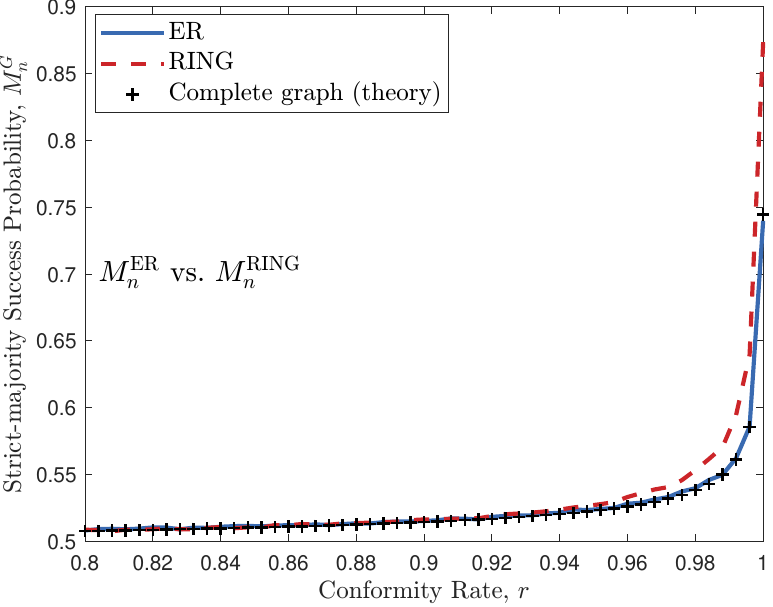}
\end{minipage}\hspace{0.01\textwidth}%
\begin{minipage}[c]{0.43\textwidth}
\centering
\includegraphics[width=\linewidth]{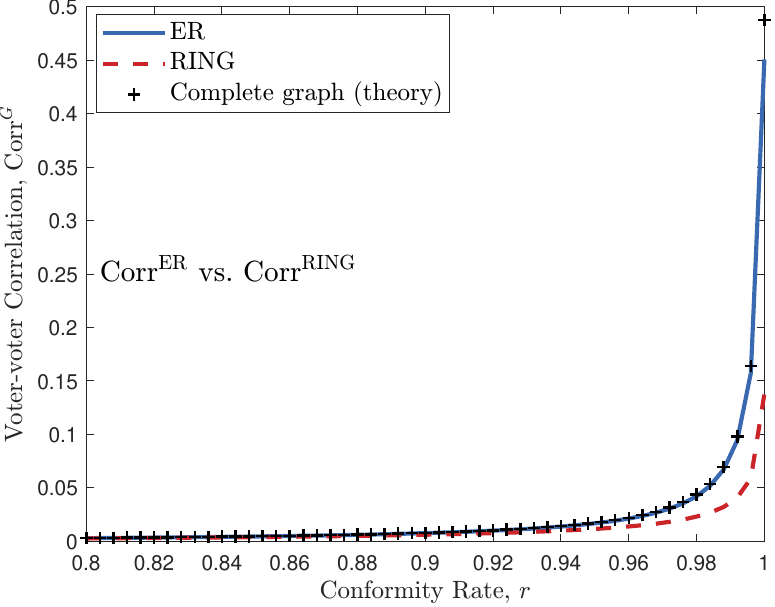}
\end{minipage}

\vspace{0.35em}

\noindent
\begin{minipage}[c]{0.10\textwidth}
\centering
\scriptsize
\makebox[\linewidth][c]{%
\begin{tabular}{@{}r@{}l@{}}
$\alpha_1$&$=2.5$\\[-0.05em]
$\alpha_2$&$=1$
\end{tabular}}
\end{minipage}\hspace{0.01\textwidth}%
\begin{minipage}[c]{0.43\textwidth}
\centering
\includegraphics[width=\linewidth]{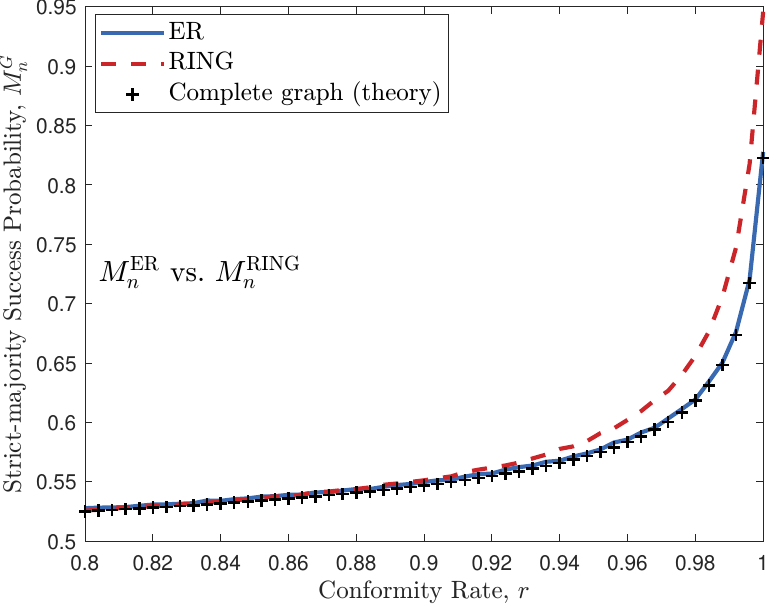}
\end{minipage}\hspace{0.01\textwidth}%
\begin{minipage}[c]{0.43\textwidth}
\centering
\includegraphics[width=\linewidth]{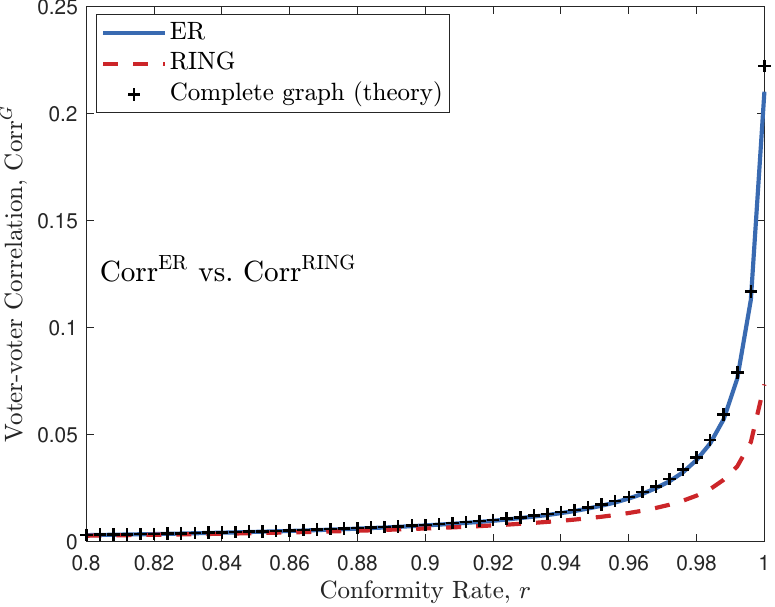}
\end{minipage}

\vspace{0.35em}

\noindent
\begin{minipage}[c]{0.10\textwidth}
\centering
\scriptsize
\makebox[\linewidth][c]{%
\begin{tabular}{@{}r@{}l@{}}
$\alpha_1$&$=5$\\[-0.05em]
$\alpha_2$&$=2$
\end{tabular}}
\end{minipage}\hspace{0.01\textwidth}%
\begin{minipage}[c]{0.43\textwidth}
\centering
\includegraphics[width=\linewidth]{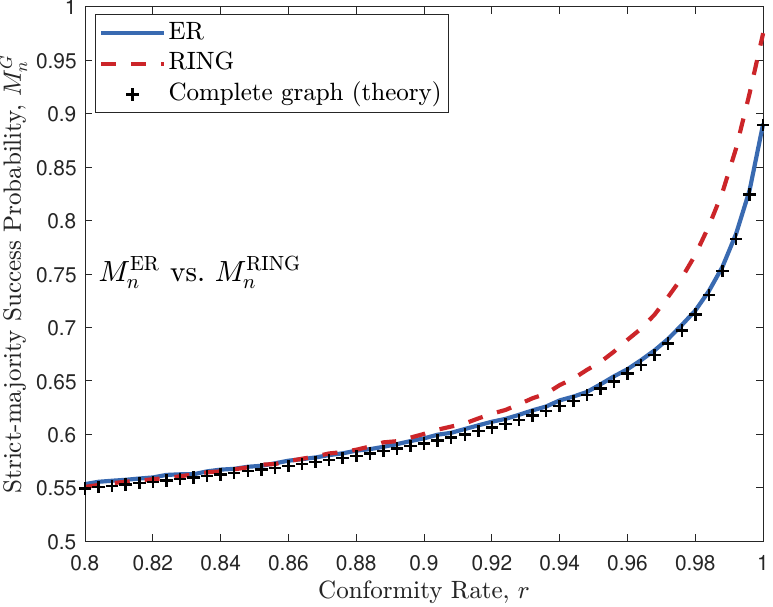}
\end{minipage}\hspace{0.01\textwidth}%
\begin{minipage}[c]{0.43\textwidth}
\centering
\includegraphics[width=\linewidth]{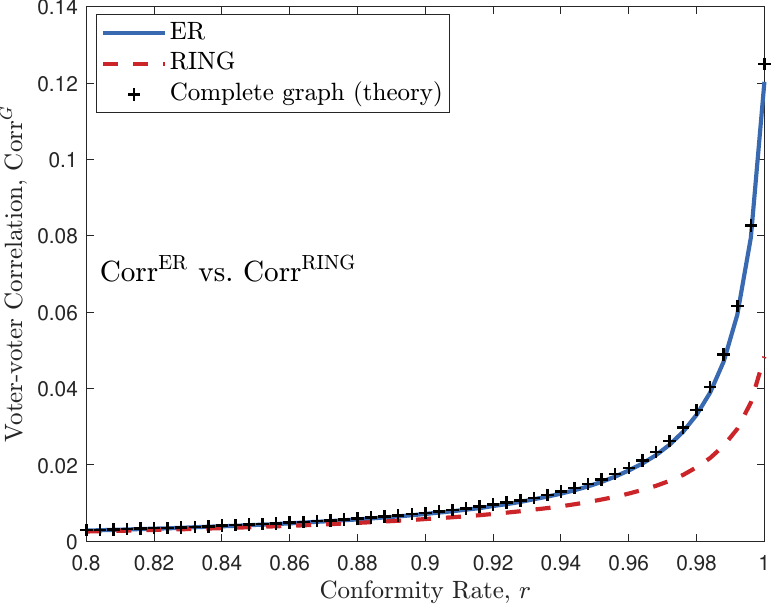}
\end{minipage}

\caption{Comparison of Erd\H{o}s--R\'enyi (ER) and ring networks for the zealot--contrarian voter model, with $n=501$ free voters and conformity rate $r\in[0.8,1]$, where topology-dependent differences are most pronounced. The three rows correspond to the parameter pairs indicated in the left column. In each row, the middle panel shows the simulated strict majority correctness $M_n^G$ for ER and ring networks together with the complete graph benchmark, and the right panel shows the corresponding pairwise voter--voter correlation. Zealot strengths in the network simulations are rescaled according to the ER mean-field mapping from Section~\ref{voter-er}. Across all three parameter pairs, the ER results again closely track the complete graph theory, whereas the ring topology yields higher strict majority correctness and lower positive voter--voter correlation than the ER topology.}
\label{fig:ER_RING_Mn_corr_vs_r}
\end{figure}

We finally compare ER and SW networks, shown in Fig.~\ref{fig:ER_SW_Mn_corr_vs_r}. In these simulations, the SW networks are generated with rewiring probability $p_{\mathrm{rewire}}=0.01$. The layout is again the same as in Figs.~\ref{fig:ER_SF_Mn_corr_vs_r} and \ref{fig:ER_RING_Mn_corr_vs_r}. Fig.~\ref{fig:ER_SW_Mn_corr_vs_r} shows that the SW topology also improves on ER throughout the plotted range. For all three parameter pairs, the SW curves for strict majority correctness lie above the ER curves, while the correlation panels show that the SW networks generate systematically lower pairwise voter--voter correlations than ER. The complete graph benchmark again remains closely aligned with the ER curves, consistent with the ER mean-field reduction from Section~\ref{voter-er}. Together with Figs.~\ref{fig:ER_SF_Mn_corr_vs_r} and \ref{fig:ER_RING_Mn_corr_vs_r}, Fig.~\ref{fig:ER_SW_Mn_corr_vs_r} reinforces a single conclusion: across these topologies, stronger positive dependence is associated with lower collective accuracy, whereas weaker dependence is associated with higher collective accuracy.

\begin{figure}[!htbp]
\centering
\captionsetup{width=0.97\textwidth}

\noindent
\begin{minipage}[c]{0.10\textwidth}
\centering
\scriptsize
\makebox[\linewidth][c]{%
\begin{tabular}{@{}r@{}l@{}}
$\alpha_1$&$=0.75$\\[-0.05em]
$\alpha_2$&$=0.3$
\end{tabular}}
\end{minipage}\hspace{0.01\textwidth}%
\begin{minipage}[c]{0.43\textwidth}
\centering
\includegraphics[width=\linewidth]{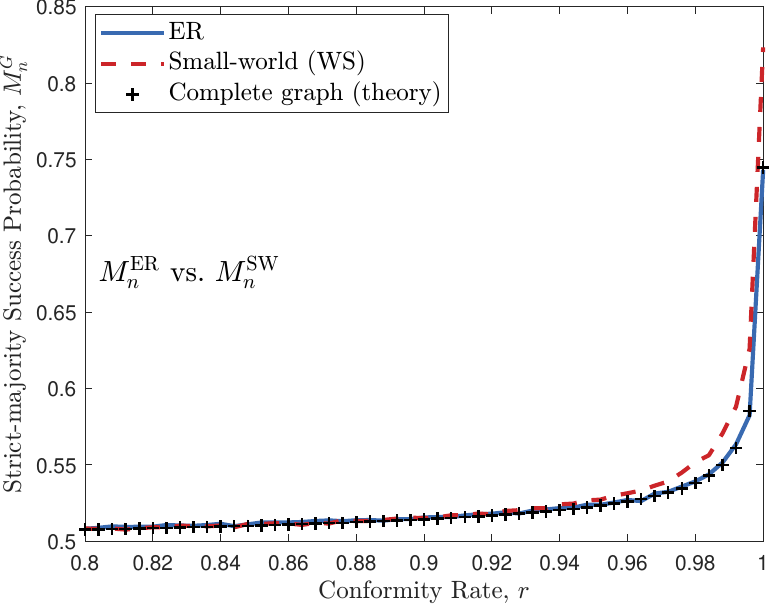}
\end{minipage}\hspace{0.01\textwidth}%
\begin{minipage}[c]{0.43\textwidth}
\centering
\includegraphics[width=\linewidth]{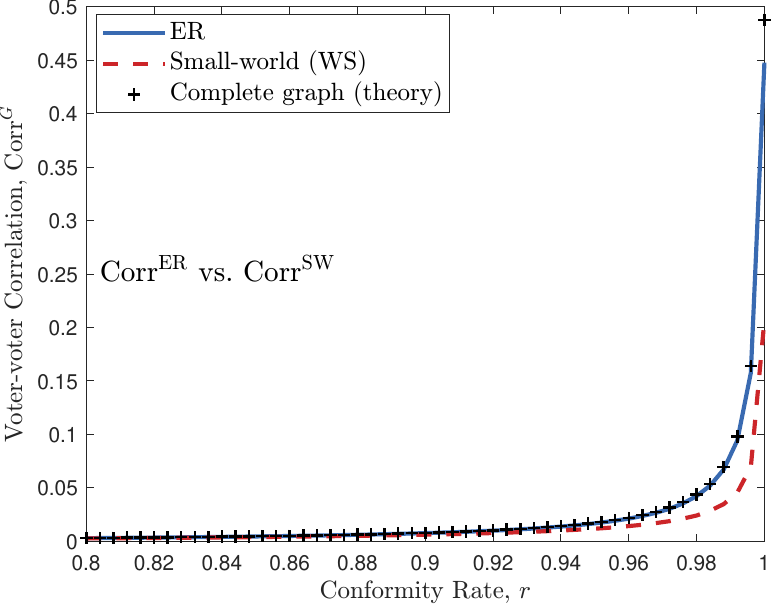}
\end{minipage}

\vspace{0.35em}

\noindent
\begin{minipage}[c]{0.10\textwidth}
\centering
\scriptsize
\makebox[\linewidth][c]{%
\begin{tabular}{@{}r@{}l@{}}
$\alpha_1$&$=2.5$\\[-0.05em]
$\alpha_2$&$=1$
\end{tabular}}
\end{minipage}\hspace{0.01\textwidth}%
\begin{minipage}[c]{0.43\textwidth}
\centering
\includegraphics[width=\linewidth]{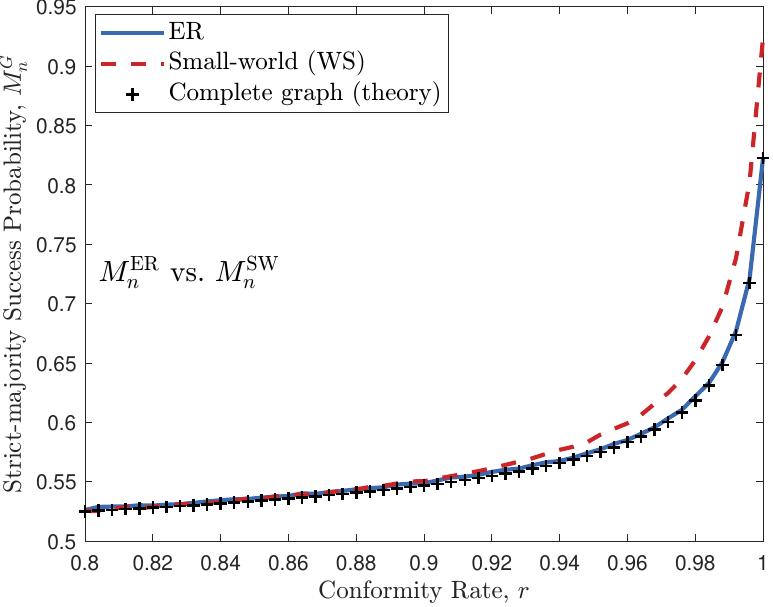}
\end{minipage}\hspace{0.01\textwidth}%
\begin{minipage}[c]{0.43\textwidth}
\centering
\includegraphics[width=\linewidth]{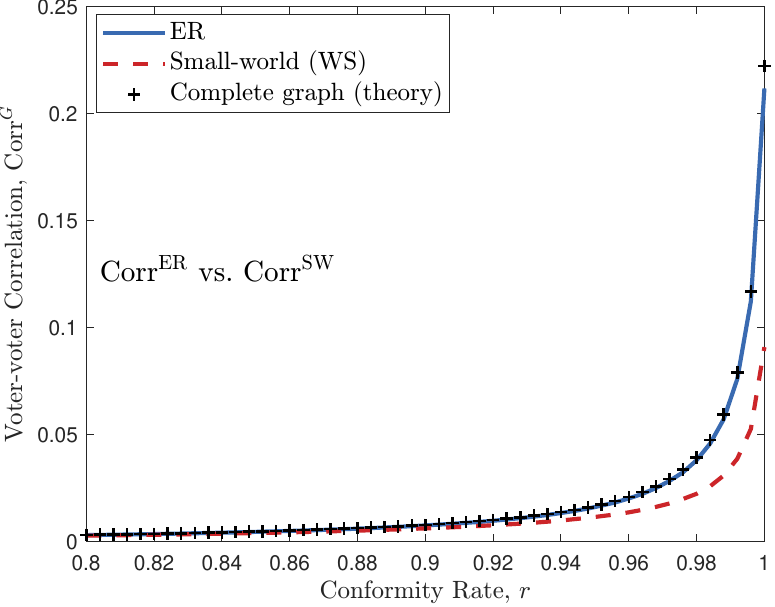}
\end{minipage}

\vspace{0.35em}

\noindent
\begin{minipage}[c]{0.10\textwidth}
\centering
\scriptsize
\makebox[\linewidth][c]{%
\begin{tabular}{@{}r@{}l@{}}
$\alpha_1$&$=5$\\[-0.05em]
$\alpha_2$&$=2$
\end{tabular}}
\end{minipage}\hspace{0.01\textwidth}%
\begin{minipage}[c]{0.43\textwidth}
\centering
\includegraphics[width=\linewidth]{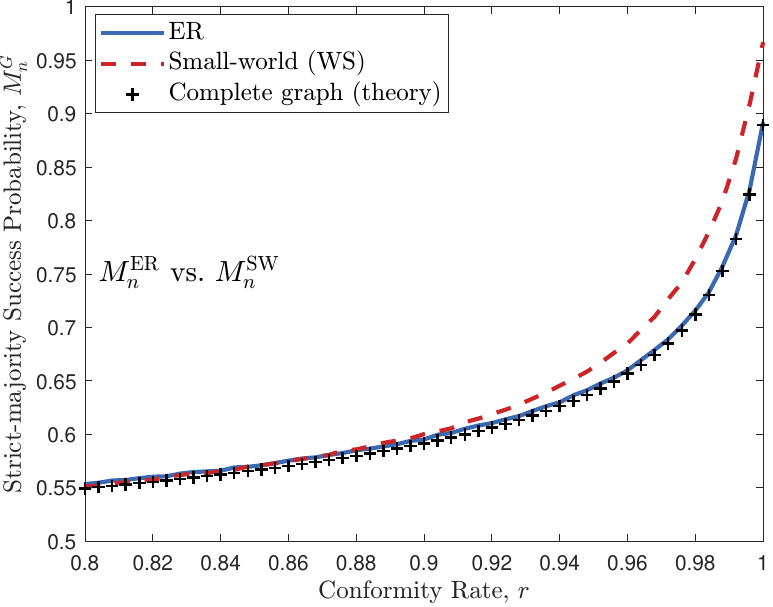}
\end{minipage}\hspace{0.01\textwidth}%
\begin{minipage}[c]{0.43\textwidth}
\centering
\includegraphics[width=\linewidth]{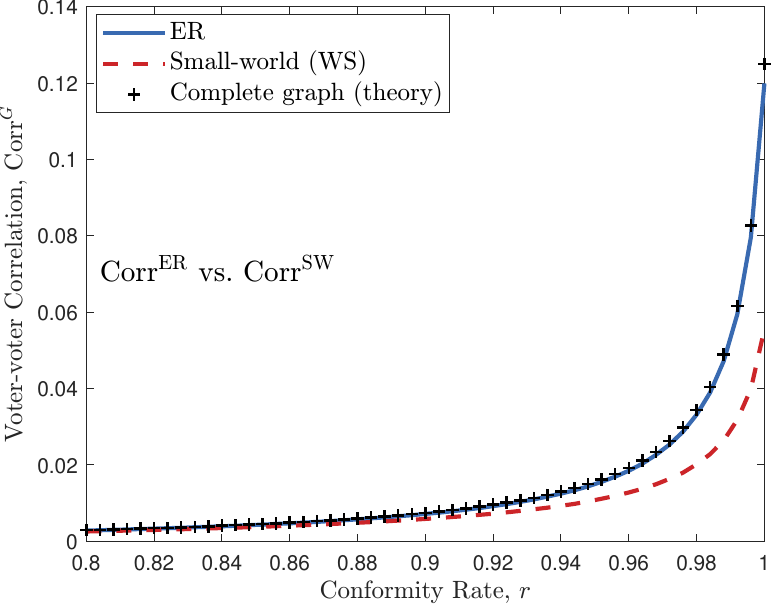}
\end{minipage}

\caption{Comparison of Erd\H{o}s--R\'enyi (ER) and small world (SW) networks for the zealot--contrarian voter model, with $n=501$ free voters, conformity rate $r\in[0.8,1]$, and SW rewiring probability $p_{\mathrm{rewire}}=0.01$. The three rows correspond to the parameter pairs indicated in the left column. In each row, the middle panel shows the simulated strict majority correctness $M_n^G$ for ER and SW networks together with the complete graph benchmark, and the right panel shows the corresponding pairwise voter--voter correlation. Zealot strengths in the network simulations are rescaled according to the ER mean-field mapping from Section~\ref{voter-er}. Across all three parameter pairs, the ER results again closely track the complete graph theory, whereas the SW topology yields higher strict majority correctness and lower positive voter--voter correlation than the ER topology.}
\label{fig:ER_SW_Mn_corr_vs_r}
\end{figure}

To examine the SW family more broadly, Fig.~\ref{fig:ER_SW_rewire_Mn_vs_r} varies the SW rewiring probability across $p_{\mathrm{rewire}}\in\{0.01,0.05,0.10,0.20\}$ while keeping the same three parameter pairs used in Figs.~\ref{fig:ER_SF_Mn_corr_vs_r}--\ref{fig:ER_SW_Mn_corr_vs_r}. Each row corresponds to one pair $(\alpha_1,\alpha_2)$, and each column corresponds to one rewiring probability. Fig.~\ref{fig:ER_SW_rewire_Mn_vs_r} shows that, for all three parameter pairs and for each rewiring probability considered, the SW topology continues to yield higher strict majority correctness than the ER topology throughout the plotted range of $r$. At the same time, the advantage of the SW topology becomes progressively smaller as $p_{\mathrm{rewire}}$ increases from $0.01$ to $0.20$, and the SW curves move closer to the ER curves. This is consistent with the same dependence based mechanism emphasized above. As the rewiring probability increases, the SW network becomes more random and structurally closer to an Erd\H{o}s--R\'enyi network, so the correlation reduction characteristic of low rewiring SW networks weakens. Correspondingly, the gap between $M_n^{\mathrm{SW}}$ and $M_n^{\mathrm{ER}}$ shrinks. Fig.~\ref{fig:ER_SW_rewire_Mn_vs_r} therefore shows that the improvement in strict majority correctness produced by the SW family is strongest at low rewiring and fades as the topology becomes more random.

\begin{figure}[!htbp]
\centering
\captionsetup{width=0.98\textwidth}

\noindent
\begin{minipage}[c]{0.10\textwidth}
\centering
\scriptsize
\end{minipage}\hspace{0.005\textwidth}%
\begin{minipage}[c]{0.21\textwidth}
\centering
\scriptsize
\makebox[\linewidth][c]{\hspace*{0.25\linewidth}$p_{\mathrm{rewire}}=0.01$\hspace*{-0.25\linewidth}}
\end{minipage}\hspace{0.005\textwidth}%
\begin{minipage}[c]{0.21\textwidth}
\centering
\scriptsize
\makebox[\linewidth][c]{\hspace*{0.25\linewidth}$p_{\mathrm{rewire}}=0.05$\hspace*{-0.25\linewidth}}
\end{minipage}\hspace{0.005\textwidth}%
\begin{minipage}[c]{0.21\textwidth}
\centering
\scriptsize
\makebox[\linewidth][c]{\hspace*{0.25\linewidth}$p_{\mathrm{rewire}}=0.10$\hspace*{-0.25\linewidth}}
\end{minipage}\hspace{0.005\textwidth}%
\begin{minipage}[c]{0.21\textwidth}
\centering
\scriptsize
\makebox[\linewidth][c]{\hspace*{0.25\linewidth}$p_{\mathrm{rewire}}=0.20$\hspace*{-0.25\linewidth}}
\end{minipage}

\vspace{0.25em}

\noindent
\begin{minipage}[c]{0.10\textwidth}
\centering
\scriptsize
\makebox[\linewidth][c]{%
\begin{tabular}{@{}r@{}l@{}}
$\alpha_1$&$=0.75$\\[-0.05em]
$\alpha_2$&$=0.3$
\end{tabular}}
\end{minipage}\hspace{0.005\textwidth}%
\begin{minipage}[c]{0.21\textwidth}
\centering
\includegraphics[width=\linewidth]{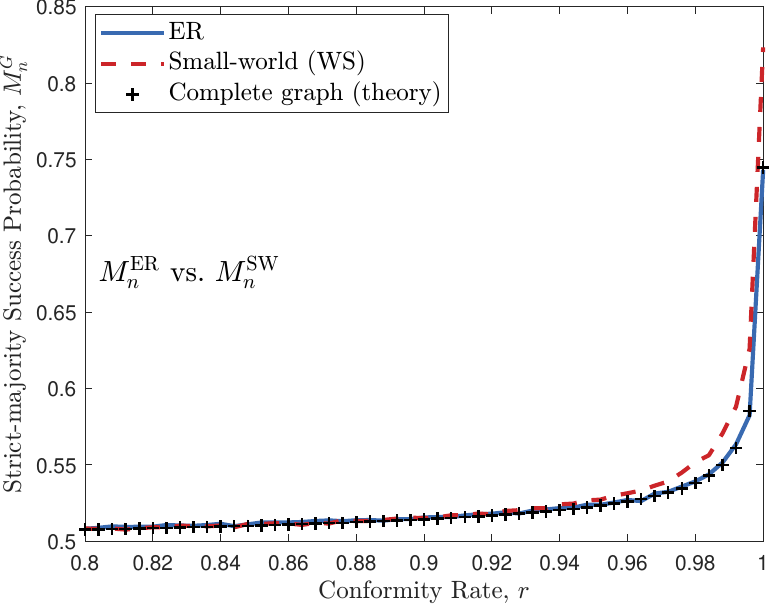}
\end{minipage}\hspace{0.005\textwidth}%
\begin{minipage}[c]{0.21\textwidth}
\centering
\includegraphics[width=\linewidth]{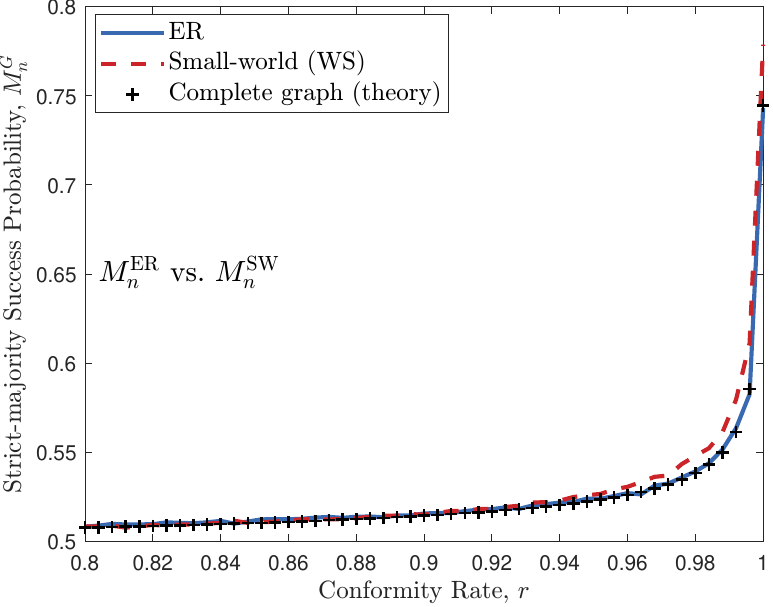}
\end{minipage}\hspace{0.005\textwidth}%
\begin{minipage}[c]{0.21\textwidth}
\centering
\includegraphics[width=\linewidth]{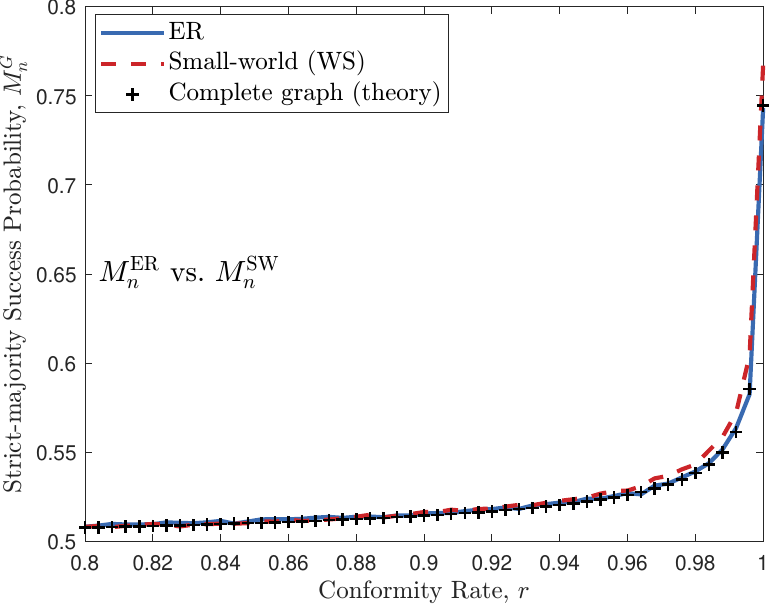}
\end{minipage}\hspace{0.005\textwidth}%
\begin{minipage}[c]{0.21\textwidth}
\centering
\includegraphics[width=\linewidth]{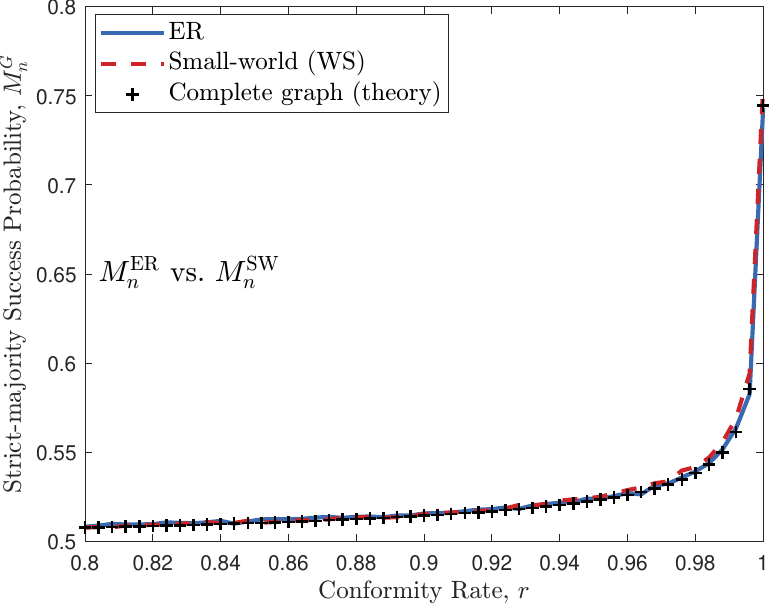}
\end{minipage}

\vspace{0.25em}

\noindent
\begin{minipage}[c]{0.10\textwidth}
\centering
\scriptsize
\makebox[\linewidth][c]{%
\begin{tabular}{@{}r@{}l@{}}
$\alpha_1$&$=2.5$\\[-0.05em]
$\alpha_2$&$=1$
\end{tabular}}
\end{minipage}\hspace{0.005\textwidth}%
\begin{minipage}[c]{0.21\textwidth}
\centering
\includegraphics[width=\linewidth]{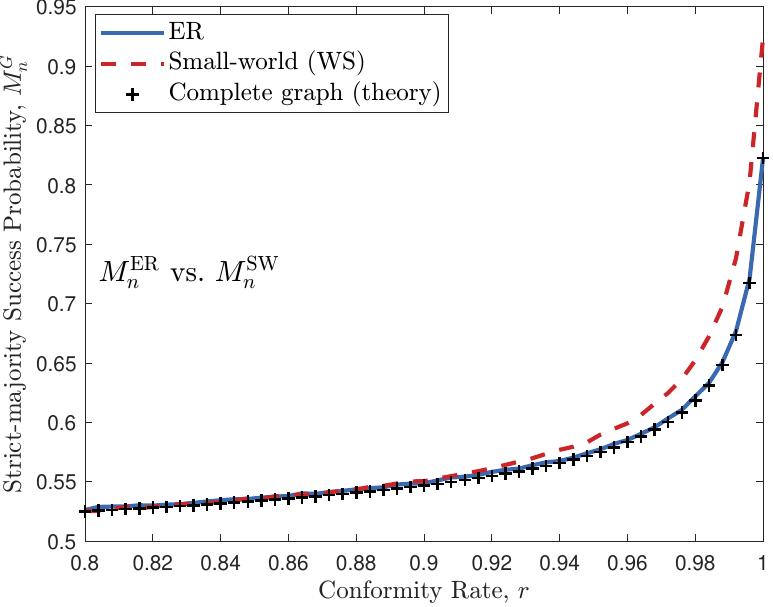}
\end{minipage}\hspace{0.005\textwidth}%
\begin{minipage}[c]{0.21\textwidth}
\centering
\includegraphics[width=\linewidth]{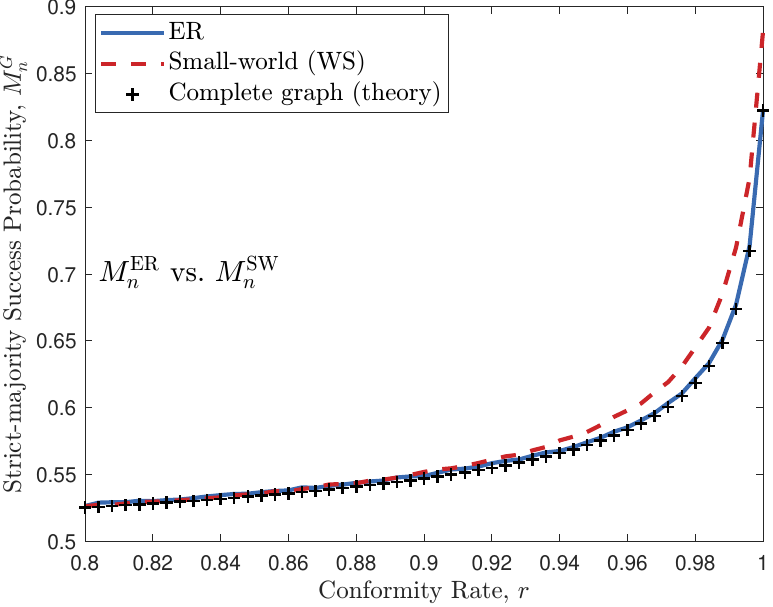}
\end{minipage}\hspace{0.005\textwidth}%
\begin{minipage}[c]{0.21\textwidth}
\centering
\includegraphics[width=\linewidth]{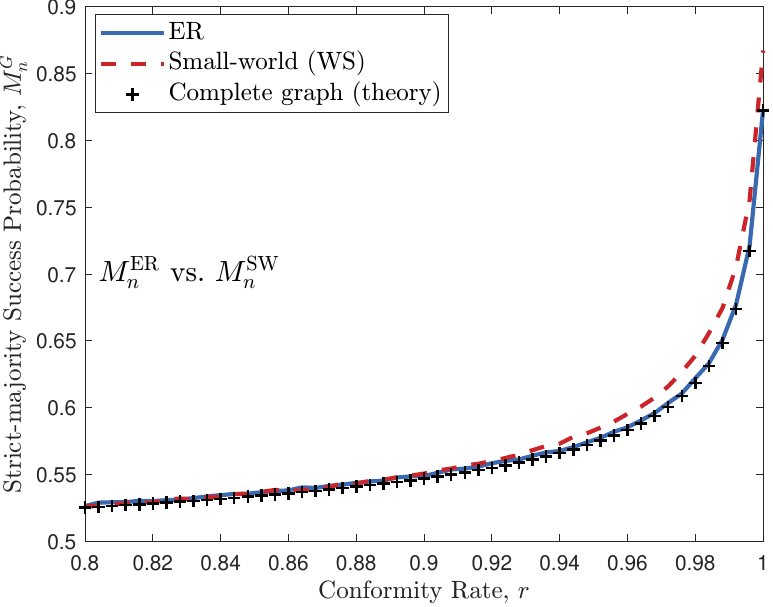}
\end{minipage}\hspace{0.005\textwidth}%
\begin{minipage}[c]{0.21\textwidth}
\centering
\includegraphics[width=\linewidth]{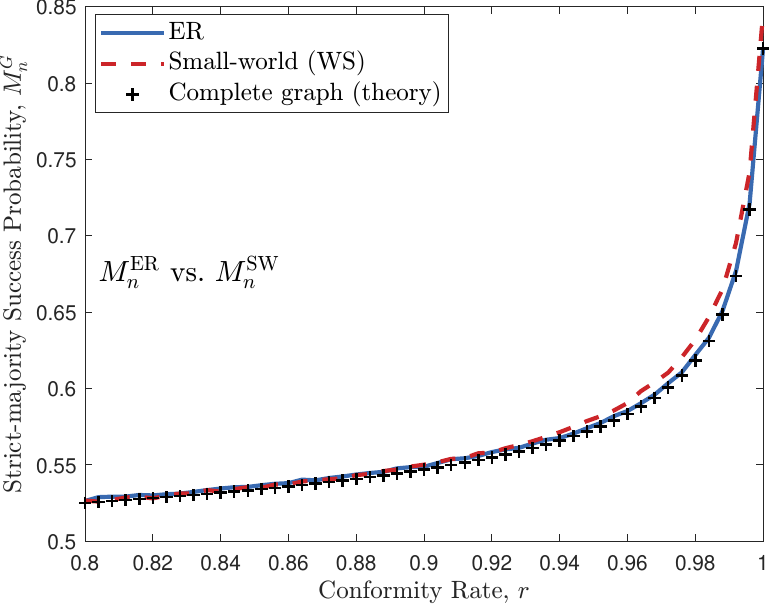}
\end{minipage}

\vspace{0.25em}

\noindent
\begin{minipage}[c]{0.10\textwidth}
\centering
\scriptsize
\makebox[\linewidth][c]{%
\begin{tabular}{@{}r@{}l@{}}
$\alpha_1$&$=5$\\[-0.05em]
$\alpha_2$&$=2$
\end{tabular}}
\end{minipage}\hspace{0.005\textwidth}%
\begin{minipage}[c]{0.21\textwidth}
\centering
\includegraphics[width=\linewidth]{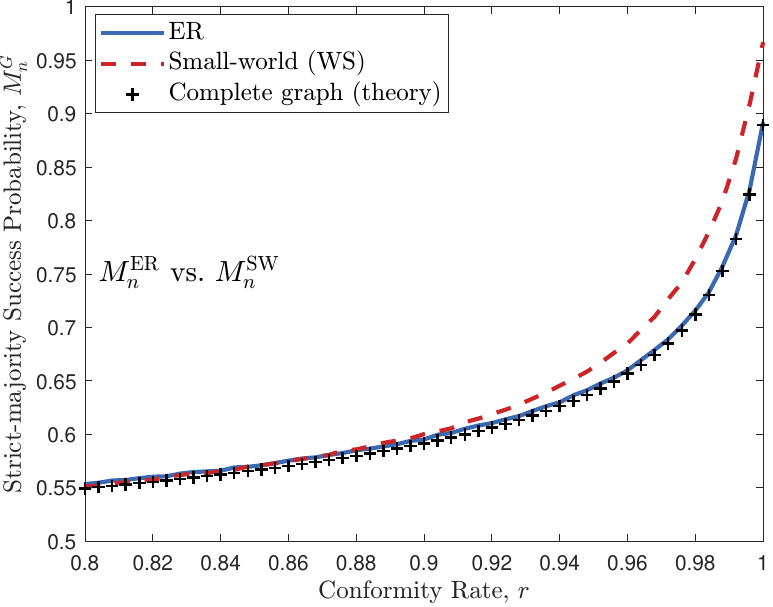}
\end{minipage}\hspace{0.005\textwidth}%
\begin{minipage}[c]{0.21\textwidth}
\centering
\includegraphics[width=\linewidth]{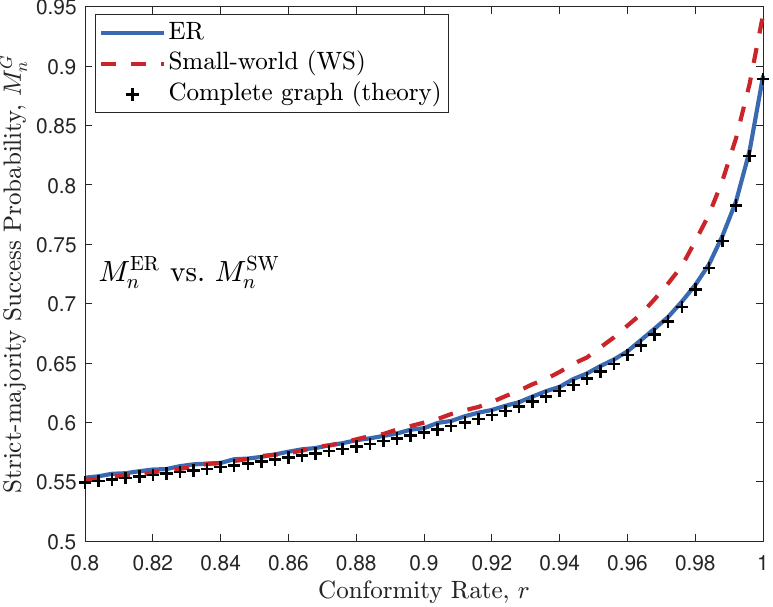}
\end{minipage}\hspace{0.005\textwidth}%
\begin{minipage}[c]{0.21\textwidth}
\centering
\includegraphics[width=\linewidth]{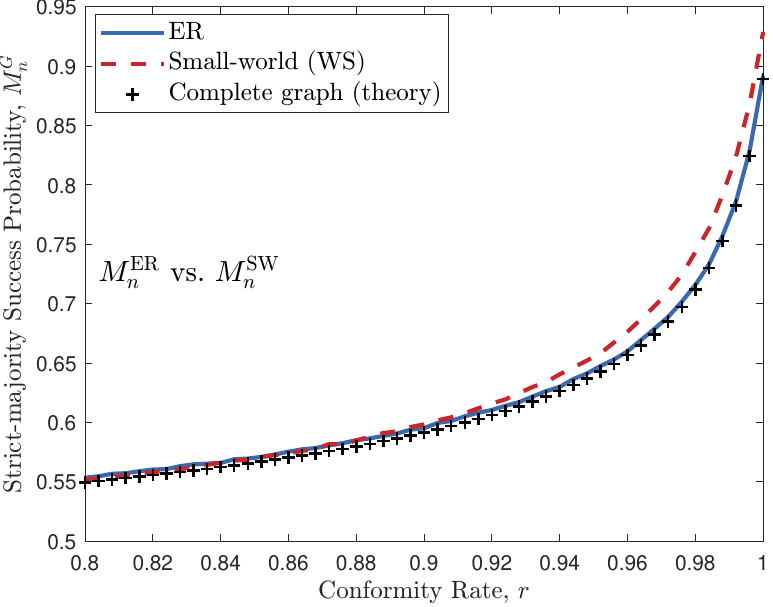}
\end{minipage}\hspace{0.005\textwidth}%
\begin{minipage}[c]{0.21\textwidth}
\centering
\includegraphics[width=\linewidth]{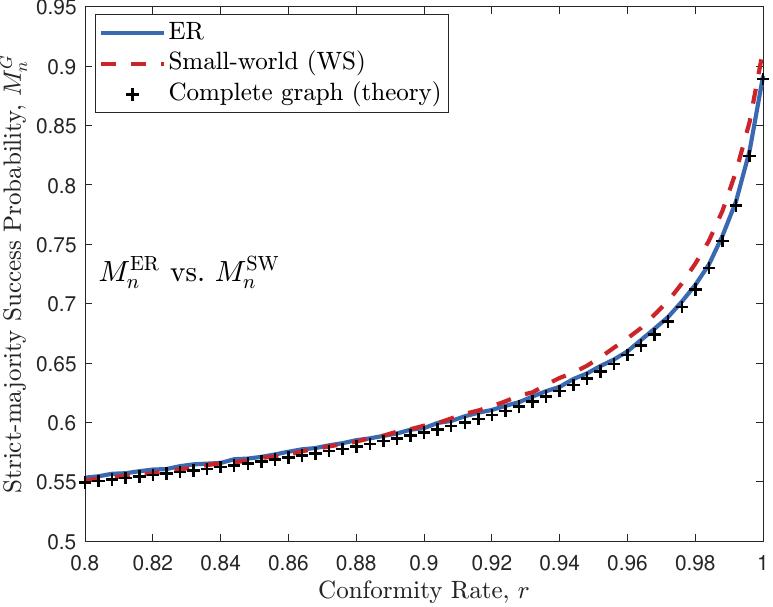}
\end{minipage}

\caption{Strict majority correctness $M_n^G$ for Erd\H{o}s--R\'enyi (ER) and small world (SW) networks as a function of the conformity rate $r$, across a range of rewiring probabilities $p_{\mathrm{rewire}}$. The three rows correspond to the parameter pairs indicated in the left column, and the four columns correspond to $p_{\mathrm{rewire}}=0.01$, $0.05$, $0.10$, and $0.20$, respectively. In every panel, the simulated ER and SW curves are shown together with the complete graph benchmark. Across all three parameter pairs and all four rewiring probabilities, the SW topology yields higher strict majority correctness than ER, but the gap narrows as $p_{\mathrm{rewire}}$ increases.}
\label{fig:ER_SW_rewire_Mn_vs_r}
\end{figure}

We now return to the two Condorcet-type comparisons from Sections~\ref{sec:majority-complete-finite} and \ref{sec:majority-no-deliberation} and examine how they behave on sparse and structured networks. Fig.~\ref{fig:all_topologies_Mn_minus_p} studies the analogue of Theorem~\ref{thm:majority-complete-finite} by plotting $M_n^G-p^G$ for ER, SF, ring, and SW networks, where $p^G$ denotes the stationary marginal correctness of a single free voter on topology $G$. Throughout this comparison we use $(\alpha_1,\alpha_2)=(5,2)$, $n=501$, $d=8$, and, for the SW case, $p_{\mathrm{rewire}}=0.01$. In each panel, the black ``+'' markers denote the complete graph benchmark. Fig.~\ref{fig:all_topologies_Mn_minus_p} shows that the finite electorate Condorcet-type inequality from Section~\ref{sec:majority-complete-finite} remains visible on all four network families considered here. In every panel, $M_n^G-p^G$ stays positive throughout the displayed range of $r$, so the simulated majority remains more accurate than a single post-deliberation free voter. The ER curve is again in excellent agreement with the complete graph theory, providing further numerical support for the mean-field approximation from Section~\ref{voter-er}. The SF, ring, and SW panels deviate somewhat from the complete graph benchmark, but the deviations are modest and do not alter the sign of the comparison. Thus, at least for this parameter regime, the majority advantage established in Theorem~\ref{thm:majority-complete-finite} appears robust to substantial network heterogeneity and structure.

\begin{figure}[!htbp]
\centering
\captionsetup{width=0.95\textwidth}
\includegraphics[width=0.95\textwidth]{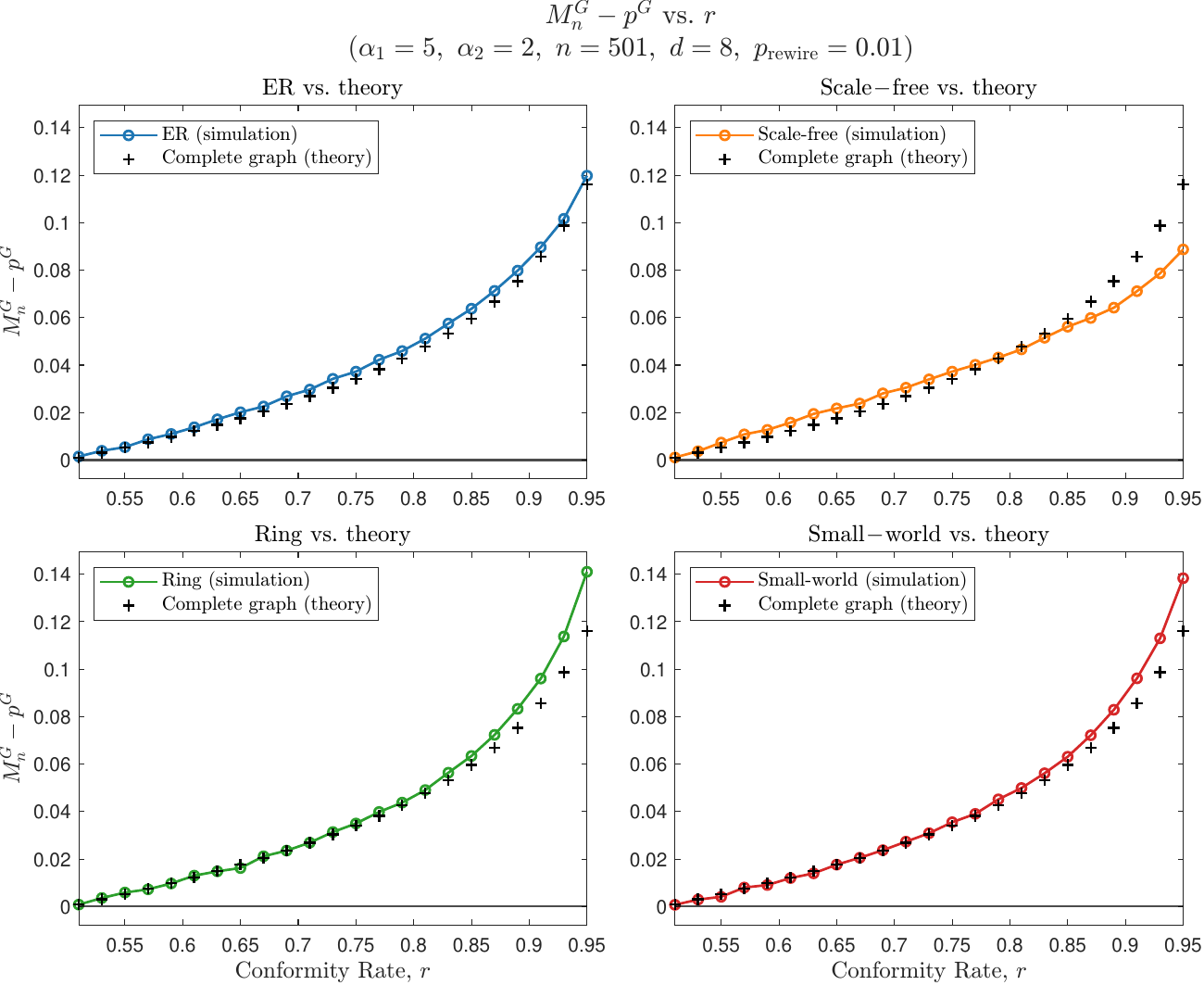}
\caption{Network analogue of the comparison in Theorem~\ref{thm:majority-complete-finite}. The four panels show $M_n^G-p^G$ as a function of the conformity rate $r$ for Erd\H{o}s--R\'enyi (ER), scale free (SF), ring, and small world (SW) networks, with $(\alpha_1,\alpha_2)=(5,2)$, $n=501$, $d=8$, and $p_{\mathrm{rewire}}=0.01$ for the SW networks. Here $M_n^G$ is the simulated strict majority correctness on topology $G$, and $p^G$ is the corresponding stationary marginal correctness of a single free voter after the campaign phase. In each panel, the black ``+'' markers denote the complete graph benchmark. Across all four topologies, the plotted quantity remains positive throughout the displayed range of $r$, numerically extending the inequality from Theorem~\ref{thm:majority-complete-finite} beyond complete graphs.}
\label{fig:all_topologies_Mn_minus_p}
\end{figure}

Fig.~\ref{fig:all_topologies_Mn_minus_MnND} turns to the second Condorcet-type comparison and studies the network analogue of Theorem~\ref{thm:majority-no-deliberation}. It plots $M_n^G-M_n^{\mathrm{ND}}$ for the same four topologies and the same parameter values. Here $M_n^{\mathrm{ND}}$ is the strict majority success probability in the no-deliberation benchmark, where all free--free edges are removed while each free voter remains exposed to the zealots. As in Fig.~\ref{fig:all_topologies_Mn_minus_p}, the black ``+'' markers denote the complete graph benchmark. Fig.~\ref{fig:all_topologies_Mn_minus_MnND} shows that the no-deliberation comparison from Section~\ref{sec:majority-no-deliberation} is likewise robust on networks. For every topology shown, $M_n^G-M_n^{\mathrm{ND}}$ remains negative throughout the plotted range, so endogenous interaction among free voters continues to reduce strict majority correctness relative to the corresponding no-deliberation benchmark. Once again, the ER curve lies very close to the complete graph theory, providing additional support for the mean-field reduction from Section~\ref{voter-er}. The SF, ring, and SW topologies show visible but still moderate departures from the complete graph benchmark, yet all preserve the same qualitative ordering. In this sense, the aggregation failure identified analytically in Theorem~\ref{thm:majority-no-deliberation} appears numerically robust well beyond the fully mixed setting.

\begin{figure}[!htbp]
\centering
\captionsetup{width=0.95\textwidth}
\includegraphics[width=0.95\textwidth]{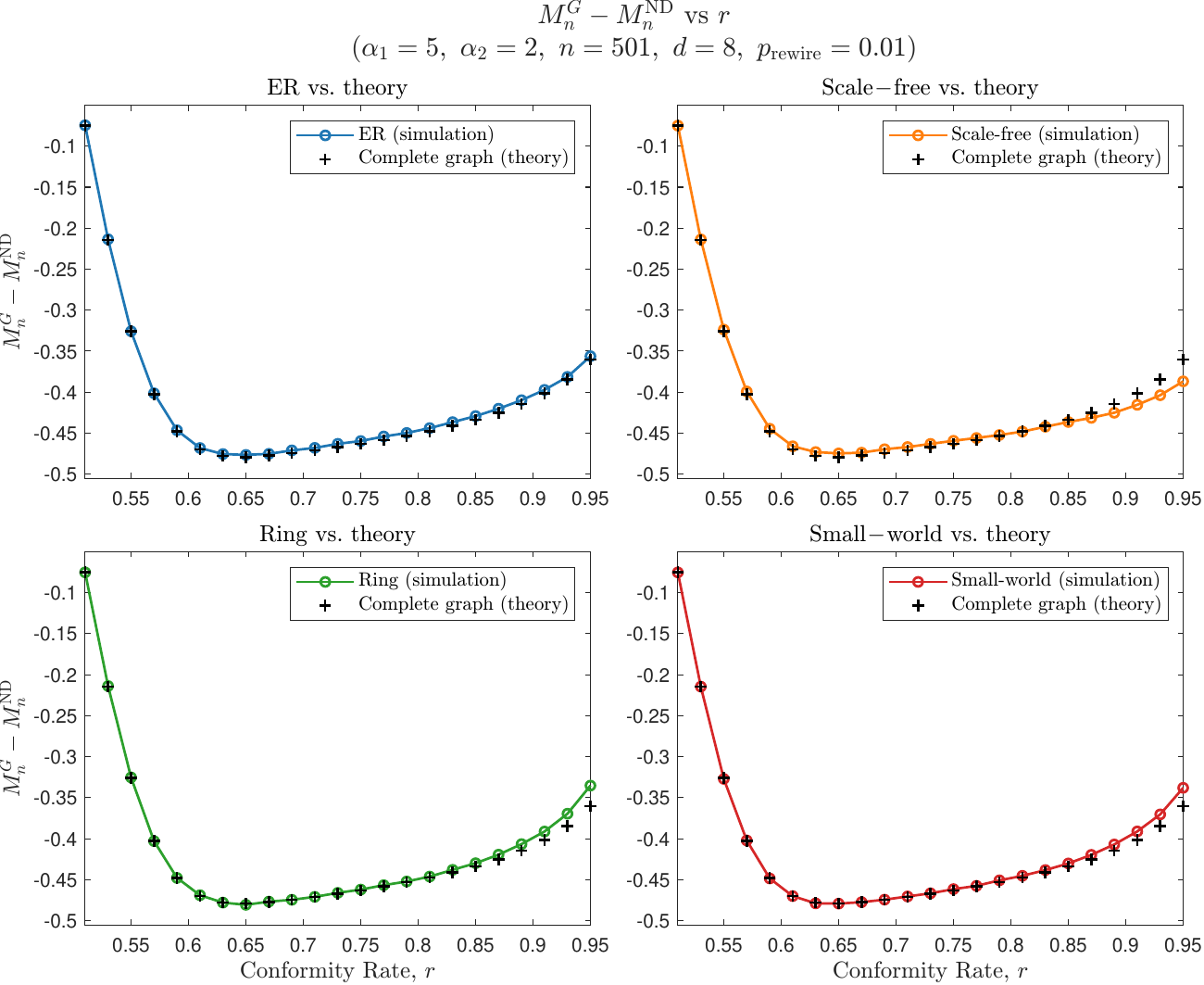}
\caption{Network analogue of the comparison in Theorem~\ref{thm:majority-no-deliberation}. The four panels show $M_n^G-M_n^{\mathrm{ND}}$ as a function of the conformity rate $r$ for Erd\H{o}s--R\'enyi (ER), scale free (SF), ring, and small world (SW) networks, with $(\alpha_1,\alpha_2)=(5,2)$, $n=501$, $d=8$, and $p_{\mathrm{rewire}}=0.01$ for the SW networks. Here $M_n^G$ is the simulated strict majority correctness on topology $G$, and $M_n^{\mathrm{ND}}$ is the strict majority success probability in the no-deliberation benchmark in which free voters do not influence one another. In each panel, the black ``+'' markers denote the complete graph benchmark. Across all four topologies, the plotted quantity remains negative throughout the displayed range of $r$, numerically extending the inequality from Theorem~\ref{thm:majority-no-deliberation} beyond complete graphs.}
\label{fig:all_topologies_Mn_minus_MnND}
\end{figure}

\section{Discussion}
\label{discussion}

This paper studies a simple but consequential departure from the classical Condorcet picture of collective choice. Instead of treating votes as independent draws, we model the campaign phase explicitly and allow repeated pre-vote interaction under the influence of committed leaders, which we call zealots. Across the sparse network settings studied here, a common theme emerges: social interaction changes collective accuracy mainly by changing the dependence structure of votes. In our framework, strict-majority performance is governed by the voter--voter correlations generated jointly by the pre-vote dynamics, the strength and balance of zealot influence, the conformity--contrarian mix, and the structure of the underlying interaction network.

For fully mixed electorates, the model yields an explicit long-run distribution of votes in both the binary and multi-alternative settings, which in turn permits an exact analysis of strict-majority correctness. The finite-electorate results show that when the zealots favor the correct alternative and conformity dominates contrarian updating, a strict majority is more likely to be correct than a single post-deliberation free voter. At the same time, the paper identifies a stronger comparison benchmark. Relative to the no-deliberation counterfactual, in which free voters respond only to zealot influence and are independent in stationarity, endogenous social interaction lowers strict-majority correctness. Put differently, pre-vote interaction can generate a finite electorate Condorcet-type gain relative to an individual after deliberation, while still producing an aggregation failure relative to the corresponding no-deliberation benchmark.

The large electorate result shows that the finite electorate advantage does not persist asymptotically. Whenever there is a persistent tendency toward contrarian behavior, for example through resistance to a view because it is gaining broad support or through persistent noise in social copying, both single-voter correctness and strict-majority correctness converge to $1/2$, the level of random choice, as the electorate grows. This stands in sharp contrast to the independent no-deliberation benchmark, under which strict-majority correctness converges to one. The intuition is that even a small tendency to resist what is becoming popular prevents the early advantage of the correct alternative from accumulating through the electorate. Communication no longer works as a one-way channel that steadily reinforces the informational advantage of the correct alternative. Instead, as support for the correct alternative spreads, it also generates resistance, so the informational gain from social interaction is gradually neutralized rather than amplified. Surprisingly, however, this collapse stops exactly at the boundary case of full conformity. When resistance disappears altogether, social influence no longer offsets its own informational effect; it carries the initial advantage forward in a cumulative way. That is why full conformity emerges as the tipping point: it is the unique case in which the epistemic benefit of pre-vote interaction survives asymptotically.

The simulation results on sparse networks extend this message in a way that is especially relevant for social choice and political behavior. Under the numerically validated Erd\H{o}s--R\'enyi mean-field approximation, a randomly connected electorate behaves much like the fully mixed model after an appropriate parameter rescaling. By contrast, other topologies produce systematically different levels of correlation among votes and therefore systematically different levels of strict-majority correctness. Scale-free networks perform worse than Erd\H{o}s--R\'enyi networks and generate stronger positive voter--voter correlations. A natural interpretation is that highly connected hubs act as secondary influencers: once many free voters are exposed to the same hubs, their votes become more synchronized, and the majority combines fewer independent judgments.  Ring networks and low rewiring small-world networks show the opposite tendency. Because influence remains more local, they weaken long-range coordination, preserve more diversity across neighborhoods, and generate lower voter--voter correlations. In this sense, topology matters because it changes how much effectively independent judgments survives until the vote. The paper therefore contributes to the literature on correlated voting not by imposing correlation statically from the outset, but by showing dynamically how it can arise through repeated communication on a network. 
This perspective also helps position the paper relative to the relevant literature. In social choice, our results complement work on correlated voting (see Section~\ref{literat}) by showing how dependence can arise from sustained interaction before the vote, rather than from a single common influence affecting everyone at once. In the voter model and network dynamics literatures, the analysis adds an explicitly epistemic criterion: the main object of interest is not only consensus, polarization, or stationary vote shares, but the probability that a strict majority selects the correct alternative. More broadly, the paper aligns with the political networks literature, which has emphasized that social relationships are central to politics and that treating political behavior as independent across individuals can obscure important mechanisms of influence and aggregation \citep{lazer2011}.

The connection to real voting systems is suggestive. Actual electorates rarely aggregate fully isolated private judgments. Instead, they are influenced by campaigns, parties, media, advocacy organizations, interpersonal discussion, and digitally mediated flows of information. In that broad sense, the model is relevant to settings such as legislative voting, where party leaders or whips may resemble persistent, zealot-like sources of influence, and where members also influence one another through pre-vote interaction on a network. It is also relevant to deliberative bodies in which discussion occurs within smaller groups before broader aggregation. The zealot--contrarian voter model is not meant as a literal description of any particular institution. Rather, it provides a tractable framework for thinking about how communication structure affects the informational value of majority decisions.

These findings also help clarify the role of correlation. If many voters are repeatedly influenced by the same zealots, hubs, or shared paths of influence , then their votes become more similar, and the resulting majority may reflect fewer effectively independent judgments than the vote count suggests. In that sense, majority rule becomes more fragile as an information aggregation device. By contrast, communication structures that keep interaction partly local, as in ring networks or low rewiring small-world networks, can preserve more variation across neighborhoods and thereby leave more informational content in the final vote. This makes the small-world results especially interesting from an institutional point of view. They suggest that systems with substantial local discussion and only limited long-range shortcuts may sometimes preserve more useful diversity than highly centralized systems. More broadly, they raise the possibility that hierarchical or federated decision processes, in which opinions are first formed within local communities and only then aggregated more widely, may differ epistemically from systems in which influence is rapidly funneled through a small set of highly connected voters.

From this viewpoint, the model suggests a useful distinction for the design of democratic institutions. Reforms that raise average voter competence need not improve collective accuracy if they also sharply increase vote correlation through repeated interaction among free voters, especially when network structure produces highly connected secondary influencers. But the implication is not simply that decentralized networks are better. What matters is the joint effect of influence strength and network architecture on the covariance structure of votes. This is one reason the contrast here among Erd\H{o}s--R'enyi, scale-free, ring, and small-world systems is informative. It suggests an epistemic perspective on campaign communication in which the key question is not only how competent voters are on average, but also how patterns of social interaction compress many nominally separate votes into a smaller number of effectively independent judgments.

The paper also suggests several concrete empirical directions. One is controlled laboratory or online experiments in which participants answer incentivized binary factual or policy questions, interact over assigned communication networks, and are exposed to a small number of experimentally seeded committed participants favoring competing alternatives. By varying network structure, especially the degree of centralization, as well as the persistence of seeded influence, such experiments could test whether majority accuracy changes in the direction predicted here. A second direction is field experiments in deliberative mini-publics, classrooms, civic forums, or citizens’ assemblies, using questions with independently verifiable answers rather than live electoral contests. In these settings, the relevant outcomes would include not only the final majority decision, but also the evolution of vote correlations and the extent to which highly connected participants shape the flow of influence. A third direction is to use high-resolution communication or social-media data to test whether more centralized discussion structures are associated with higher vote covariance and lower majority reliability on questions with verifiable answers.

The broader implication is straightforward. Majority rule should be understood not only as an aggregation rule, but also as the endpoint of a social process. Once the campaign phase is modeled explicitly, the key question is no longer simply whether voters are individually competent, but how pre-vote interaction and network structure jointly determine how much independent information remains at the moment of voting. The results of this paper suggest that this issue belongs naturally at the intersection of social choice theory, political behavior, and network science.

\appendix
\numberwithin{equation}{section}
\renewcommand{\theequation}{\thesection\arabic{equation}}

\section{Proof of Theorem~\ref{thm:stationary-m-contrarian}}
\label{app:proof-thm2}

We provide a detailed balance proof of the stationary distribution for the $m$-alternative zealot--contrarian voter model on a fully mixed population.
As in the binary case, the stationary law does not depend on the inertia parameter $\ell$, since $P=\ell I+(1-\ell)Q$ and any stationary distribution of the social-update kernel $Q$ is also stationary for $P$. We therefore set $\ell=0$ and work with $Q$.

Let $k=(k_1,\dots,k_m)\in\mathcal{S}_n$ and fix distinct $i,j\in\{1,\dots,m\}$. Write $k':=k+e_i-e_j$ for the state obtained by moving one free vote from alternative $j$ to alternative $i$.
Let $\alpha_0:=\sum_{s=1}^m \alpha_s$ and $T:=n-1+\alpha_0$.

\subsection*{Transition probabilities}

In a social update, we first choose a free voter currently supporting alternative $j$ with probability $k_j/n$.
Next we sample a second individual uniformly from the remaining $T$ individuals (free or zealot), so the sampled individual's vote equals $s$ with probability $(k_s+\alpha_s)/T$.
Given a sampled vote $s$, the updating voter adopts vote $i$ with probability $r$ if $s=i$, and with probability $(1-r)/(m-1)$ if $s\neq i$ (a contrarian step chooses uniformly among the $m-1$ alternatives other than $s$).
Hence, for $i\neq j$,
\begin{align}
Q\big(k\to k'\big)
&=\frac{k_j}{n}\left[r\,\frac{k_i+\alpha_i}{T}+\frac{1-r}{m-1}\left(1-\frac{k_i+\alpha_i}{T}\right)\right] \notag\\
&=\frac{k_j}{n}\left[\frac{1-r}{m-1}+\frac{mr-1}{m-1}\cdot\frac{k_i+\alpha_i}{T}\right]. \label{eq:Q_multistate_appendix}
\end{align}
When $r=\tfrac1m$, the bracket equals $1/m$ (independent of $k$), so $Q(k\to k')=(k_j/n)\cdot(1/m)$.

For $r\neq\tfrac1m$, define the shift and shifted parameters as in Theorem~\ref{thm:stationary-m-contrarian}:
\begin{equation}
\delta:=\frac{(1-r)T}{mr-1},
\qquad
\widetilde{\alpha}_i:=\alpha_i+\delta,
\qquad
\widetilde{\alpha}_0:=\sum_{s=1}^m \widetilde{\alpha}_s.
\end{equation}
Then \eqref{eq:Q_multistate_appendix} can be rewritten as
\begin{equation}
\label{eq:Q_multistate_tilde_appendix}
Q\big(k\to k+e_i-e_j\big)=\frac{k_j}{n}\cdot \frac{mr-1}{m-1}\cdot \frac{k_i+\widetilde{\alpha}_i}{T}.
\end{equation}

\subsection*{Proof of part (i): $r=\tfrac1m$}

When $r=\tfrac1m$, the updating voter chooses its new vote uniformly from $\{1,\dots,m\}$, regardless of the sampled vote.
Consequently, at the configuration level (tracking each free voter's vote), the Markov chain is invariant under permutations of voters and alternatives.
In stationarity, each free voter is uniform over the $m$ alternatives, which implies that the induced count vector $X=(X_1,\dots,X_m)$ is multinomial:
\begin{equation}
\Pp(X=k)=\frac{n!}{k_1!\cdots k_m!}\left(\frac1m\right)^n,\qquad k\in\mathcal{S}_n.
\end{equation}

\subsection*{Proof of part (ii): $r\neq\tfrac1m$}

Define, for $k\in\mathcal{S}_n$,
\begin{equation}
\label{eq:pi_candidate_appendix}
\pi(k):=\frac{n!}{k_1!\cdots k_m!}\,
\frac{(\widetilde{\alpha}_1)_{k_1}\cdots(\widetilde{\alpha}_m)_{k_m}}{(\widetilde{\alpha}_0)_n}.
\end{equation}
(When $mr>1$ we have $\widetilde{\alpha}_i>0$ and this is the standard Dirichlet--multinomial distribution; when $mr<1$ the same closed form still arises from detailed balance and remains a valid stationary distribution for the finite Markov chain.)

To verify stationarity, it suffices to check detailed balance on each nearest-neighbor move $k\to k'=k+e_i-e_j$.
Using \eqref{eq:pi_candidate_appendix} and the identities
\begin{equation}
\begin{aligned}
\frac{k_j!}{(k_j-1)!}&=k_j, &
\frac{(k_i+1)!}{k_i!}&=k_i+1,\\
\frac{(\widetilde{\alpha}_i)_{k_i+1}}{(\widetilde{\alpha}_i)_{k_i}}&=\widetilde{\alpha}_i+k_i, &
\frac{(\widetilde{\alpha}_j)_{k_j-1}}{(\widetilde{\alpha}_j)_{k_j}}&=\frac{1}{\widetilde{\alpha}_j+k_j-1},
\end{aligned}
\end{equation}
we obtain
\begin{equation}
\label{eq:pi_ratio_appendix}
\frac{\pi(k')}{\pi(k)}
=\frac{k_j}{k_i+1}\cdot \frac{k_i+\widetilde{\alpha}_i}{k_j-1+\widetilde{\alpha}_j}.
\end{equation}
On the other hand, the transition probabilities \eqref{eq:Q_multistate_tilde_appendix} give
\begin{equation}
\begin{aligned}
\frac{Q(k\to k')}{Q(k'\to k)}
&=
\frac{\frac{k_j}{n}\cdot \frac{mr-1}{m-1}\cdot \frac{k_i+\widetilde{\alpha}_i}{T}}
{\frac{k_i+1}{n}\cdot \frac{mr-1}{m-1}\cdot \frac{k_j-1+\widetilde{\alpha}_j}{T}}\\
&=
\frac{k_j}{k_i+1}\cdot \frac{k_i+\widetilde{\alpha}_i}{k_j-1+\widetilde{\alpha}_j}.
\end{aligned}
\end{equation}
Comparing with \eqref{eq:pi_ratio_appendix} shows that $\pi(k)\,Q(k\to k')=\pi(k')\,Q(k'\to k)$ for every neighboring pair, hence $\pi$ satisfies detailed balance and is stationary for $Q$.

Finally, normalization in \eqref{eq:pi_candidate_appendix} follows from the multivariate Vandermonde identity
\begin{equation}
\sum_{k\in\mathcal{S}_n}\frac{n!}{k_1!\cdots k_m!}\,(\widetilde{\alpha}_1)_{k_1}\cdots(\widetilde{\alpha}_m)_{k_m}
=(\widetilde{\alpha}_0)_n,
\end{equation}
which can be obtained by multiplying the binomial series $(1-z)^{-\widetilde{\alpha}_i}=\sum_{k_i\ge 0}(\widetilde{\alpha}_i)_{k_i}z^{k_i}/k_i!$ over $i=1,\dots,m$ and equating coefficients of $z^n$ after multiplying by $n!$.
Since the chain is finite, irreducible, and aperiodic, this stationary distribution is unique.

\section{A second proof of Theorem~\ref{thm:majority-complete-finite} via the P\'olya--Eggenberger family}
\label{app:berg-majority-proof}

The proof in Section~\ref{sec:majority-complete-finite} proceeds directly from the stationary distribution derived from the
zealot--contrarian dynamics. For completeness, we present a second proof based on a reparameterization into the P\'olya--Eggenberger
family studied by \citet{berg1993condorcet}. The two formulations should nevertheless be kept conceptually distinct.

Berg studies a P\'olya--Eggenberger family written in terms of a marginal success parameter $p$ and a dependence
parameter $\psi$. He then analyzes majority functions within that family. His Condorcet-type inequality is therefore static in nature: once a
vote count distribution belongs to the P\'olya--Eggenberger family, the comparison between strict majority accuracy and individual accuracy
depends only on that distribution, not on the micro-dynamics that generated it.

By contrast, our stationary distribution arises from an explicit opinion dynamics model with parameters $(\alpha_1,\alpha_2,r,n)$. As shown in Theorem~\ref{thm:stationary-binary-contrarian}, it
is written in the form
\begin{equation}
\pi(k)=\binom{n}{k}\frac{(\theta)_k(\beta)_{n-k}}{(\theta+\beta)_n},
\qquad k=0,1,\dots,n,
\label{eq:B1}
\end{equation}
with $(\theta,\beta)$ determined by \eqref{eq:def-theta-beta_maintext}. We show below that every stationary distribution realizable by our
zealot--contrarian dynamics belongs to the P\'olya--Eggenberger family. The converse is not true: our dynamics do not generate all probability
laws covered by the admissible parameter region of the P\'olya--Eggenberger family. The contribution here is therefore dynamic. We derive the
stationary law \eqref{eq:B1} from an explicit voter interaction process and relate $(\alpha_1,\alpha_2,r)$ to the induced parameters $(p,\psi)$.
That correspondence is enough to import Berg's strict majority correctness inequality.

To see the algebraic correspondence, let
\begin{equation}
p=\frac{\theta}{\theta+\beta},
\qquad q=1-p=\frac{\beta}{\theta+\beta},
\qquad \psi=\frac{1}{\theta+\beta}.
\end{equation}
For $m\in\mathbb{Z}_{\ge 0}$, we write the generalized ascending factorial
\begin{equation}
x^{[m,\psi]}:=x(x+\psi)\cdots (x+(m-1)\psi),
\qquad x^{[0,\psi]}:=1
\end{equation}
in terms of the Pochhammer (rising) factorial used in Theorem~\ref{thm:stationary-binary-contrarian}:
\begin{equation}
x^{[m,\psi]}=\psi^m\left(\frac{x}{\psi}\right)_m.
\end{equation}
With $p/\psi=\theta$, $q/\psi=\beta$, and $1/\psi=\theta+\beta$, the P\'olya--Eggenberger distribution can be written as
\begin{align}
b_n(k;p,\psi)
&:=\binom{n}{k}\frac{p^{[k,\psi]}q^{[n-k,\psi]}}{1^{[n,\psi]}} \notag\\
&=\binom{n}{k}\frac{\psi^k(\theta)_k\,\psi^{n-k}(\beta)_{n-k}}{\psi^n(\theta+\beta)_n} \notag\\
&=\binom{n}{k}\frac{(\theta)_k(\beta)_{n-k}}{(\theta+\beta)_n}
=\pi(k).
\end{align}
Thus the stationary distribution of the zealot--contrarian voter model belongs algebraically to the P\'olya--Eggenberger family, with marginal parameter
\begin{equation}
p=\frac{\theta}{\theta+\beta}.
\end{equation}

\begin{proof}[Second proof of Theorem~\ref{thm:majority-complete-finite}]
Assume $r\neq \tfrac12$ and $n=2s+1$ odd. By the above reparameterization, the stationary law of $X$ coincides with Berg's
P\'olya--Eggenberger distribution $b_n(k;p,\psi)$. Berg's ``Condorcet's jury theorem II'' for odd electorates states that the corresponding majority
function satisfies
\begin{equation}
\sum_{k=s+1}^{n} b_n(k;p,\psi)>p
\qquad\text{whenever }p>\tfrac12
\end{equation}
for admissible parameter values, including negative-dependence regimes. Applying this gives
\begin{equation}
M_n=\sum_{k=s+1}^{n}\pi(k)=\sum_{k=s+1}^{n} b_n(k;p,\psi)>p,
\end{equation}
which is exactly the claim of Theorem~\ref{thm:majority-complete-finite}. If $r=\tfrac12$, then the claim reduces to the already noted
identity $M_n=p=\tfrac12$. This appendix proof is only a distributional shortcut: the main theorem does not rely on Berg's formulation,
and the key point is that the stationary vote count distribution generated by the zealot--contrarian dynamics falls into a family for which
Berg's inequality applies.
\end{proof}

\section{Proof of Theorem~\ref{thm:majority-large-n-mixed}}
\label{app:large-n-majority-proof}

This appendix gives the full proof of the large-electorate result for the mixed regime $r\in(\tfrac12,1)$, where conformist updating remains dominant but contrarian updating is still present.

\begin{proof}
Write
\[
a:=2r-1>0,
\]
and recall
\[
\theta_n=\frac{r\alpha_1+(1-r)(n-1+\alpha_2)}{a},
\qquad
\beta_n=\frac{(1-r)\alpha_1+r(n-1+\alpha_2)}{a}-n+1.
\]
By Theorem~\ref{thm:stationary-binary-contrarian}, $X_n$ has the beta--binomial law
\[
\Pp(X_n=k)=\binom{n}{k}\frac{(\theta_n)_k(\beta_n)_{n-k}}{(\theta_n+\beta_n)_n},
\qquad k=0,1,\dots,n.
\]
Since $r>\tfrac12$, this is an ordinary beta--binomial distribution, so there exists a mixing variable
\[
P_n\sim \mathrm{Beta}(\theta_n,\beta_n)
\]
such that
\[
X_n\mid P_n\sim \mathrm{Binomial}(n,P_n).
\]
The parameters satisfy
\begin{equation}
\label{eq:sumdiff-large-app}
\theta_n+\beta_n=\frac{\alpha_1+\alpha_2+2(1-r)(n-1)}{a},
\qquad
\theta_n-\beta_n=\alpha_1-\alpha_2.
\end{equation}

The proof separates the fluctuation of the latent support level $P_n$ from the conditional binomial fluctuation around that latent level.

We begin with the asymptotics of the beta mixing variable. Let
\[
c=\frac{1-r}{a}>0.
\]
Then \eqref{eq:sumdiff-large-app} implies
\[
\frac{\theta_n}{n}\longrightarrow c,
\qquad
\frac{\beta_n}{n}\longrightarrow c.
\]
Because $P_n\sim\mathrm{Beta}(\theta_n,\beta_n)$ and both shape parameters diverge linearly in $n$, the standard large-parameter normal approximation for the beta distribution gives
\[
\frac{P_n-\mu_n}{\sigma_n}
\xrightarrow{d}
\mathcal{N}(0,1),
\]
where
\[
\mu_n=\frac{\theta_n}{\theta_n+\beta_n},
\qquad
\sigma_n^2=\frac{\theta_n\beta_n}{(\theta_n+\beta_n)^2(\theta_n+\beta_n+1)}.
\]
Using \eqref{eq:sumdiff-large-app},
\[
\mu_n
=
\frac12+\frac{\alpha_1-\alpha_2}{2(\theta_n+\beta_n)}
=
\frac12+O\!\left(\frac1n\right),
\]
while a direct calculation from $\theta_n/n\to c$ and $\beta_n/n\to c$ yields
\[
n\sigma_n^2
\longrightarrow
\frac{1}{8c}
=
\frac{2r-1}{8(1-r)}.
\]
Now write
\[
A_n=\sqrt n\!\left(P_n-\frac12\right)
=
(\sqrt n\,\sigma_n)\frac{P_n-\mu_n}{\sigma_n}
+
\sqrt n\!\left(\mu_n-\frac12\right).
\]
The first factor converges in distribution to $\mathcal N(0,1)$, while
\[
\sqrt n\,\sigma_n \longrightarrow \sqrt{\frac{2r-1}{8(1-r)}}
\qquad\text{and}\qquad
\sqrt n\!\left(\mu_n-\frac12\right)\longrightarrow 0.
\]
Therefore, by Slutsky's theorem,
\begin{equation}
\label{eq:mixing-clt-large-app}
A_n
\xrightarrow{d}
\mathcal{N}\!\left(0,\frac{2r-1}{8(1-r)}\right).
\end{equation}
In particular,
\begin{equation}
\label{eq:mixing-prob-large-app}
P_n\longrightarrow \frac12
\qquad\text{in probability.}
\end{equation}

Next we control the conditional binomial fluctuation. Fix $\delta\in(0,\tfrac12)$ and $t\in\mathbb{R}$. For $p\in[\delta,1-\delta]$, let
\[
Z_{n,p}=\frac{B_{n,p}-np}{\sqrt n},
\qquad
B_{n,p}\sim \mathrm{Binomial}(n,p).
\]
The classical de Moivre--Laplace theorem implies that, for each fixed $p\in(0,1)$,
\[
Z_{n,p}\xrightarrow{d}\mathcal{N}(0,p(1-p)).
\]
Because $p$ is restricted to the compact interval $[\delta,1-\delta]$, the standard characteristic-function proof gives this approximation uniformly in $p$. In particular, uniformly for $p\in[\delta,1-\delta]$,
\begin{equation}
\label{eq:de Moivre--Laplace}
\E\!\left[e^{itZ_{n,p}}\right]
=
\exp\!\left(-\frac{t^2}{2}p(1-p)\right)+o(1).
\end{equation}
Now let
\[
Z_n=\frac{X_n-nP_n}{\sqrt n}.
\]
Then
\begin{equation}
\label{eq:decomposition-large-app}
\frac{X_n-\frac{n}{2}}{\sqrt{n}}
=
\underbrace{\sqrt{n}\left(P_n-\frac12\right)}_{A_n}
+
\underbrace{\frac{X_n-nP_n}{\sqrt{n}}}_{Z_n}.
\end{equation}
To identify the joint limit of $(A_n,Z_n)$, fix $s,t\in\mathbb{R}$ and write
\[
E_n=\{P_n\in[\delta,1-\delta]\}.
\]
We decompose the joint characteristic function according to this event:
\begin{align}
\E\!\left[e^{isA_n+itZ_n}\right]
&=
\E\!\left[e^{isA_n+itZ_n}\mathbf{1}_{E_n}\right]
+
\E\!\left[e^{isA_n+itZ_n}\mathbf{1}_{E_n^c}\right]
\nonumber\\
&=
\E\!\left[\E\!\left[e^{isA_n+itZ_n}\mathbf{1}_{E_n}\mid P_n\right]\right]
+
\E\!\left[e^{isA_n+itZ_n}\mathbf{1}_{E_n^c}\right]
\nonumber\\
&=
\E\!\left[
e^{isA_n}\mathbf{1}_{E_n}\E\!\left[e^{itZ_n}\mid P_n\right]
\right]
+
\E\!\left[e^{isA_n+itZ_n}\mathbf{1}_{E_n^c}\right],
\label{eq:split-large-app}
\end{align}
where in the last line we used that $A_n$ and $\mathbf{1}_{E_n}$ are measurable functions of $P_n$. By \eqref{eq:mixing-prob-large-app}, $\Pp(E_n^c)\to0$, and therefore the second term in \eqref{eq:split-large-app} tends to zero.

On $E_n$, the conditional law of $Z_n$ given $P_n=p$ is the same as that of $Z_{n,p}$. Hence \eqref{eq:de Moivre--Laplace} implies
\[
\E\!\left[e^{itZ_n}\mid P_n\right]
=
\exp\!\left(-\frac{t^2}{2}P_n(1-P_n)\right)+o(1),
\]
where the $o(1)$ term is uniform on $E_n$. Substituting into \eqref{eq:split-large-app} gives
\[
\E\!\left[e^{isA_n+itZ_n}\right]
=
\E\!\left[
e^{isA_n}\exp\!\left(-\frac{t^2}{2}P_n(1-P_n)\right)\mathbf{1}_{E_n}
\right]
+o(1).
\]
Since $\Pp(E_n^c)\to0$ and the integrand is bounded by $1$ in modulus, we may remove the indicator and obtain
\begin{equation}
\label{eq:first-reduction-large-app}
\E\!\left[e^{isA_n+itZ_n}\right]
=
\E\!\left[
e^{isA_n}\exp\!\left(-\frac{t^2}{2}P_n(1-P_n)\right)
\right]
+o(1).
\end{equation}

Because $P_n\to\tfrac12$ in probability, we have
\[
P_n(1-P_n)\longrightarrow \frac14
\qquad\text{in probability,}
\]
and thus
\[
\exp\!\left(-\frac{t^2}{2}P_n(1-P_n)\right)
\longrightarrow
e^{-t^2/8}
\qquad\text{in probability.}
\]
Since the difference is bounded by $2$, it also converges to zero in $L^1$, and therefore
\[
\E\!\left[
e^{isA_n}\exp\!\left(-\frac{t^2}{2}P_n(1-P_n)\right)
\right]
-
e^{-t^2/8}\E\!\left[e^{isA_n}\right]
\longrightarrow 0.
\]
Combining this with \eqref{eq:first-reduction-large-app} and \eqref{eq:mixing-clt-large-app}, we obtain
\begin{equation}
\label{eq:joint-limit-large-app}
\E\!\left[e^{isA_n+itZ_n}\right]
\longrightarrow
\exp\!\left(
-\frac{s^2}{2}\cdot\frac{2r-1}{8(1-r)}
-\frac{t^2}{8}
\right).
\end{equation}
Thus $(A_n,Z_n)$ converges jointly to a centered bivariate normal vector $(A,Z)$ with independent components,
\[
A\sim \mathcal{N}\!\left(0,\frac{2r-1}{8(1-r)}\right),
\qquad
Z\sim \mathcal{N}\!\left(0,\frac14\right).
\]

By \eqref{eq:decomposition-large-app} and \eqref{eq:joint-limit-large-app},
\[
\frac{X_n-n/2}{\sqrt n}
\xrightarrow{d}
A+Z.
\]
Since $A$ and $Z$ are independent centered normals, the limit is centered normal with variance
\[
\frac{2r-1}{8(1-r)}+\frac14
=
\frac{1}{8(1-r)}.
\]
Therefore
\[
\frac{X_n-n/2}{\sqrt n}
\xrightarrow{d}
\mathcal{N}\!\left(0,\frac{1}{8(1-r)}\right).
\]
For odd $n$,
\[
M_n
=
\Pp\!\left(X_n>\frac n2\right)
=
\Pp\!\left(\frac{X_n-n/2}{\sqrt n}>0\right).
\]
Because the limiting normal distribution is continuous and symmetric about $0$,
\[
M_n\longrightarrow \frac12.
\]

Finally, if $\alpha_1>\alpha_2$, then \eqref{eq:p-mean} gives
\[
p_n-\frac12
=
\frac{(2r-1)(\alpha_1-\alpha_2)}
{2\bigl[\alpha_1+\alpha_2+2(1-r)(n-1)\bigr]}
>0,
\]
so $p_n>\tfrac12$ for every $n$, and the same expression shows that $p_n\to\tfrac12$ as $n\to\infty$.
\end{proof}

\vspace{1em}
\noindent\textbf{Funding} The authors received no financial support for the research, authorship, or publication of this article.

\bibliographystyle{chicago}
\bibliography{discussion_newrefs}

\end{document}